\documentclass[10pt,reqno,oneside]{amsart}

\usepackage{cancel,bm, amssymb}
\usepackage{color}

\usepackage{enumitem}

\usepackage{color}
\usepackage{graphicx,framed}
\usepackage[colorlinks=true, pdfstartview=FitV, linkcolor=blue, citecolor=blue, urlcolor=blue]{hyperref}
\definecolor{shadecolor}{rgb}{0.95, 0.95, 0.86}

\renewcommand{\d}{{\mathrm d}}

\numberwithin{equation}{section}

\newtheorem{theo}{Theorem}[section]

\newtheorem{rem}[theo]{Remark}
\newtheorem{problem}[theo]{Riemann-Hilbert Problem}

\newtheorem{prop}[theo]{Proposition} 
\newtheorem{cor}[theo]{Corollary}

\begin{document}

\title[Asymptotic behavior of a log gas in the bulk scaling limit I.]{On the asymptotic behavior of a log gas in the bulk scaling limit in the presence of a varying external potential I.}

\author{Thomas Bothner}
\address{Centre de recherches math\'ematiques,
Universit\'e de Montr\'eal, Pavillon Andr\'e-Aisenstadt, 2920 Chemin de la tour, Montr\'eal, Qu\'ebec H3T 1J4, Canada}
\email{bothner@crm.umontreal.ca}

\author{Percy Deift}
\address{Courant Institute of Mathematical Sciences, 251 Mercer St., New York, NY 10012, U.S.A.}
\email{deift@cims.nyu.edu}

\author{Alexander Its}
\address{Department of Mathematical Sciences,
Indiana University-Purdue University Indianapolis,
402 N. Blackford St., Indianapolis, IN 46202, U.S.A.}
\email{itsa@math.iupui.edu}

\author{Igor Krasovsky}
\address{Department of Mathematics, Imperial College, London SW7 2AZ, United Kingdom}
\email{i.krasovsky@imperial.ac.uk}



\thanks{T.B. acknowledges the support of Concordia University through a postdoctoral fellow top-up award as well as the hospitality of the Banff International Research Station where part of this work was completed. P.D. acknowledges support of NSF Grant DMS-1300965. A.I. acknowledges support of NSF Grants DMS-1001777 and DMS-1361856, and of SPbGU grant N 11.38.215.2014, as well as the hospitality of the Berlin Technical University and the Imperial College of London where part of this work was completed. I.K. is grateful for the hospitality of Indiana University-Purdue University Indianapolis in April 2013 and acknowledges support of the European Community Seventh Framework grant ``Random and Integrable Models in Mathematical Physics". The authors would like to thank Jinho Baik for bringing \cite{BP} to their attention. Also the authors would like to thank the referees for their careful reading of the original version of this paper, and for their comments and questions, which led the authors to reformulate, correct and streamline the paper in a variety of ways.}

\begin{abstract}
We study the determinant $\det(I-\gamma K_s), 0<\gamma <1$, of the integrable Fredholm operator $K_s$ acting on the interval $(-1,1)$ with kernel $K_s(\lambda, \mu)= \frac{\sin s(\lambda - \mu)}{\pi (\lambda-\mu)}$. This determinant arises in the analysis of a log-gas of interacting particles in the bulk-scaling limit, at inverse temperature $\beta=2$, in the presence of an external potential $v=-\frac{1}{2}\ln(1-\gamma)$ supported on an interval of length $\frac{2s}{\pi}$. We evaluate, in particular, the double scaling limit of $\det(I-\gamma K_s)$ as $s\rightarrow\infty$ and $\gamma\uparrow 1$, in the region $0\leq\kappa=\frac{v}{s}=-\frac{1}{2s}\ln(1-\gamma)\leq 1-\delta$, for any fixed $0<\delta<1$. This problem was first considered by Dyson in \cite{Dy1}.
\end{abstract}

\date{\today}
\maketitle

\section{Introduction and statement of results}\label{sec0}

Let $K_s$ be the trace class operator with kernel
\begin{equation*}
	K_s(\lambda,\mu) = \frac{\sin s(\lambda-\mu)}{\pi(\lambda-\mu)},\ \ s>0
\end{equation*}
acting on $L^2(-1,1)$. Consider the determinant $\det(I-\gamma K_s)$, where $0\leq\gamma\leq 1$. We write $\gamma$ in the form
\begin{equation*}
	\gamma=1-e^{-2v}=1-e^{-2\kappa s},\hspace{0.5cm} \kappa=\frac{v}{s},\hspace{0.5cm} 0\leq v\leq+\infty.
\end{equation*}
The determinant $\det(I-\gamma K_s)$ arises in the following way. In the bulk scaling limit for $N\times N$ Hermitian matrices from the Gaussian Unitary Ensemble (GUE), $N\rightarrow\infty$, one considers the eigenvalues $\{\eta_j\}_{j=1}^N$ of the matrices in a neighborhood of a point, say $\eta_B$, in the bulk of the spectrum, on the scale where the average number of eigenvalues is one per unit length. More precisely we consider $x_j=\rho_N(\eta_j-\eta_B),1\leq j\leq N$, where $\rho_N\d\eta,N\rightarrow\infty$, is the expected number of eigenvalues in an interval $\d\eta$ near $\eta_B$: Clearly the expected number of eigenvalues per unit $x$-length is now one. As $N\rightarrow\infty$, the $x_j$'s converge to a determinantal random point process (see \cite{Sosh} and the references therein to the seminal work of A. Lenard) with correlation kernel $K(\lambda,\mu)=\frac{\sin\pi(\lambda-\mu)}{\pi(\lambda-\mu)}$. Hence, by general theory, if $\phi\in L^{\infty}(\mathbb{R})$ with support inside a bounded set $B$,
\begin{equation*}
	\mathbb{E}\Big(\prod_{j}\big(1+\phi(x_j)\big)\Big)=\det(I+K\phi)\Big|_{L^2(B)}.
\end{equation*}
In particular (see \cite{Dy1}) if $\mathcal{V}(x)$ is an external potential such that $\mathcal{V}(x)=2v>0$ for $-\frac{s}{\pi}<x<\frac{s}{\pi},s>0$,\footnote{Dyson \cite{Dy1} uses ``$\pi s$" wherever we use ``$s$". In particular, he considers the kernel $\frac{\sin\pi s(\lambda-\mu)}{\pi(\lambda-\mu)}$ on $L^2(-1,1)$, etc.} and zero otherwise, then
\begin{equation}\label{Dcorr}
	\mathbb{E}\left(\exp\bigg[-\sum_{j}\mathcal{V}(x_j)\bigg]\right)=\det(I-\gamma K)\Big|_{L^2(-\frac{s}{\pi},\frac{s}{\pi})}=\det(I-\gamma K_s)\Big|_{L^2(-1,1)}.
\end{equation}
For GUE, the eigenvalues (``particles") exhibit repulsion proportional to $e^{-2\sum_{j<k}\ln|x_j-x_k|^{-1}}$ and, in analogy with thermodynamics, the point process $\{x_j\}$ is often called a log-gas, or a Coulomb gas, at inverse temperature $\beta=2$. This terminology was introduced by F. Dyson in the early $60$'s.\smallskip
%

The effect of the potential $2v>0$ is to push the particles $\{x_j\}$ out of the region $(-\frac{s}{\pi},\frac{s}{\pi})$. When $v>0$ is very small compared to $s$, we should have $\det(I-\gamma K_s)\sim 1$. However, as $v$ grows with respect to $s$, we anticipate that particles begin to leak out from the edges of the interval $(-\frac{s}{\pi},\frac{s}{\pi})$ into the potential free region $(-\infty,-\frac{s}{\pi})\cup(\frac{s}{\pi},\infty)$. Eventually, when $v\gg s$, there are no particles left in $(-\frac{s}{\pi},\frac{s}{\pi})$. This is consistent with the well known fact in random matrix theory that $\det(I-K_s)=\lim_{v\rightarrow\infty}\det(I-\gamma K_s)$ is the probability (or equivalently, the expectation) that there are no (bulk-scaled) eigenvalues in $(-\frac{s}{\pi},\frac{s}{\pi})$.\smallskip

We are interested in the behavior of $\det(I-\gamma K_s)$ as $s$ and $v\rightarrow\infty$ in the $(s,v)$ plane, compare Figure \ref{Deiftdia}.
\begin{figure}[tbh]
\begin{center}
\includegraphics[width=0.35\textwidth]{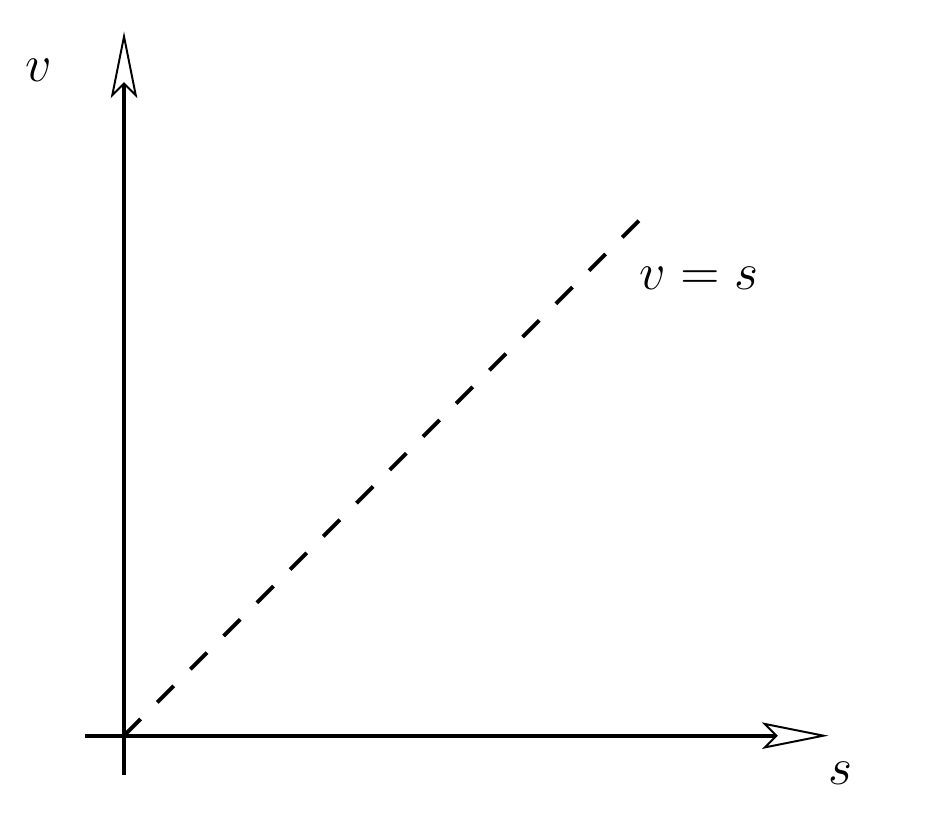}
\end{center}
\caption{The double scaling diagram}
\label{Deiftdia}
\end{figure}

For a fixed $0\leq v\leq\infty$, the behavior is known: in the case $v=\infty$, i.e. $\gamma=1$, as $s\rightarrow\infty$
\begin{equation}\label{DW}
	\det(I-K_s) = \exp\left[-\frac{s^2}{2}\right]s^{-\frac{1}{4}}c_0\Big(1+\mathcal{O}\big(s^{-1}\big)\Big),
\end{equation}
where
\begin{equation}\label{DWconstant}
	c_0 = \exp\left[\frac{1}{12}\ln 2+3\zeta'(-1)\right]
\end{equation}
and $\zeta(\cdot)$ is the Riemann zeta function. This result was conjectured by Dyson in \cite{Dy0}, and eventually proved in \cite{E,K,DIIZ}. Many authors contributed to the analysis of $\det(I-K_s)$, including des Cloizeaux, Mehta and Widom (see \cite{DIIZ}, and \cite{DII} for a historical review). For $v\in[0,\infty)$ fixed, as $s\rightarrow\infty$,
\begin{equation}\label{BWBB1}
	\det(I-\gamma K_s)=\exp\left[-\frac{4vs}{\pi}\right](4s)^{\frac{2v^2}{\pi^2}}b^2(v)\Big(1+\mathcal{O}\left(s^{-1}\right)\Big),
\end{equation}
where
\begin{equation}\label{BWBB2}
	b(v) = e^{(1+\gamma)\frac{v^2}{\pi^2}}\prod_{k=1}^{\infty}\left(1+\frac{v^2}{\pi^2k^2}\right)^ke^{-\frac{v^2}{\pi^2 k}}
\end{equation}
and $\gamma$ is Euler's constant. This result was proved by Basor and Widom in 1983 \cite{BW} and then independently by Budylin and Buslaev in 1995 \cite{BB}.
\begin{rem}
Note that $b(v)$ can be expressed in terms of the Barnes $G$-Function (see \cite{BE,NIST})
\begin{equation}\label{BG}
	b(v) = G\left(1+\frac{iv}{\pi}\right)G\left(1-\frac{iv}{\pi}\right).
\end{equation}
As the asymptotic behavior of $G(1+z)$ as $z\rightarrow\infty,|\textnormal{arg}\,z|\leq\pi-\eta,\,0<\eta<\pi,$ is known (see e.g. \cite{BE,NIST}), the asymptotic behavior of $b(v)$ as $v\rightarrow\infty$ is known through \eqref{BG}. We find
\begin{equation*}
	\ln b(v) = -\frac{v^2}{\pi^2}\ln\left(\frac{v}{\pi}\right)+\frac{3v^2}{2\pi^2}-\frac{1}{6}\ln\left(\frac{v}{\pi}\right)+2\zeta'(-1)+\mathcal{O}\left(v^{-1}\right),\ \ v\rightarrow+\infty.
\end{equation*}
\end{rem}
\begin{rem}\label{Busref} Although \eqref{BWBB1}, \eqref{BWBB2} were first proven in \cite{BW,BB} for fixed $0\leq v<\infty$, we will derive the following stronger estimation in the sequel \cite{BDIK} to the current paper -- in this expansion $v$ does not have to be held fixed any longer: $\exists\,s_0,c_1,c_2>0$ such that
\begin{equation}\label{BBrefined}
	\det(I-\gamma K_s)=\exp\left[-\frac{4vs}{\pi}\right](4s)^{\frac{2v^2}{\pi^2}}b^2(v)\Big(1+r(s,v)\Big)
\end{equation}
where $s\geq s_0,0\leq v<s^{\frac{1}{3}}$ and
\begin{equation*}
	\big|r(s,v)\big|<c_1\frac{v}{s}+c_2\frac{v^3}{s}.
\end{equation*}
This more general estimate is needed to control certain error-terms in Theorem \ref{theo1} below (see also Remark \ref{bohiPa}). The error term $r(s,v)$ can be differentiated with respect to $s$.
\end{rem}

We are interested, in particular, how the asymptotic behavior of $\det(I-\gamma K_s)$ is transformed from exponential decay $\exp\left[-\frac{4vs}{\pi}\right]$ in \eqref{BWBB1}, to super-exponential decay $\exp\left[-\frac{s^2}{2}\right]$ in \eqref{DW}, as $v$ increases from $0$ to $\infty$ (in this connection, see again Remark \ref{bohiPa}).\bigskip

To this end, we collect certain a priori bounds for $\det(I-\gamma K_s)$ which will motivate subsequent steps: From the formula
\begin{equation*}
	K_s(\lambda,\mu) = \frac{\sin s(\lambda-\mu)}{\pi(\lambda-\mu)} = \frac{1}{2\pi}\int_{-s}^se^{it\lambda}\,\overline{e^{it\mu}}\,\d t,
\end{equation*}
it follows that $K_s$ is a positive definite, trace class operator in $L^2(-1,1)$ with norm $\|K_s\|<1$ and trace norm
\begin{equation}\label{a0}
	\left\|K_s\right\|_1=\int_{-1}^1K_s(\lambda,\lambda)d\lambda = \frac{2s}{\pi}.
\end{equation}
Moreover $K_s$ is an increasing function of $s$, $\langle f,K_sf\rangle\leq \langle f,K_{s'}f\rangle$ for $s<s'$ and $f\in L^2(-1,1)$, and so $\det(I-\gamma K_s)$ is a decreasing function of $s$: Clearly, $\det(I-\gamma K_s)$ is also a decreasing function of $v$.
The eigenvalues $\{\lambda_n=\lambda_n(s)\}_{n=0}^{\infty}$ of $K_s$ satisfy $1>\lambda_0\geq\lambda_1\geq\ldots\geq\lambda_n\geq\ldots\geq 0$
and for any fixed $n\geq 0$,
\begin{equation}\label{a1}
	\lambda_n(s) = 1-\frac{\sqrt{\pi}}{n!}2^{3n+2}s^{n+\frac{1}{2}}e^{-2s}\Big(1+\mathcal{O}\left(s^{-1}\right)\Big)
\end{equation}
as $s\rightarrow\infty$ (see \cite{Sl}). In particular, as $s\rightarrow\infty$,
\begin{equation}\label{a2}
	\lambda_0(s)=1-4\sqrt{\pi s}e^{-2s}\Big(1+\mathcal{O}\left(s^{-1}\right)\Big).
\end{equation}
In order to derive the bounds for $\det(I-\gamma K_s)$ first differentiate the relation
\begin{equation*}
	\ln\det(I-\gamma K_s) = \textnormal{tr}\big(\ln(I-\gamma K_s)\big)
\end{equation*}
with respect to $\gamma$ and reintegrate, obtaining
\begin{equation}\label{a3}
	\det(I-\gamma K_s) = \exp\left[-\int_0^{\gamma}\textnormal{tr}\left(\left(I-\mu K_s\right)^{-1}K_s\right)d\mu\right]
\end{equation}
and so
\begin{equation}\label{a4}
	e^{-\frac{4vs}{\pi}}\leq e^{\frac{2s}{\pi\lambda_0}\ln(1-\lambda_0\gamma)}\leq \det(I-\gamma K_s)\leq e^{-\frac{2s\gamma}{\pi}}.
\end{equation}
We see from \eqref{a4} that $\det(I-\gamma K_s)$ is bounded below by $\exp\left[-\frac{4vs}{\pi}\right]$ which is, in particular, consistent with \eqref{BWBB1} as $s\rightarrow\infty$. From \eqref{a2}, for $s$ sufficiently large, say $s>s_0$,
\begin{equation}\label{a5}
	1-\lambda_0\gamma = 4\sqrt{\pi s}e^{-2s}\Big(1+\mathcal{O}\left(s^{-1}\right)\Big)+e^{-2v}+\mathcal{O}\left(\sqrt{s}e^{-2(s+v)}\right),
\end{equation}
and we see that the ratio 
\begin{equation*}
	\kappa\equiv\frac{v}{s},
\end{equation*}
which is proportional to the energy per unit interval, plays a crucial role. For if $\kappa>1$, then to leading order $\ln\left(1-\lambda_0\gamma\right)\sim-2s$ and so
\begin{equation*}
	\det(I-\gamma K_s)\gtrsim e^{-\frac{4s^2}{\pi}}
\end{equation*}
which is consistent with the leading order behavior in \eqref{DW}, $\det(I-K_s)\sim e^{-\frac{s^2}{2}}$. On the other hand if $\kappa<1$, then $1-\lambda_0\gamma\sim e^{-2v}$ and so
\begin{equation*}
	\det(I-\gamma K_s)\gtrsim e^{-\frac{4vs}{\pi}}
\end{equation*}
which is consistent with \eqref{BWBB1}, but inconsistent with \eqref{DW} if $\kappa<\frac{\pi}{8}<1$. This gives a rough measure of the transition of the asymptotics from \eqref{BWBB1} to \eqref{DW}. Note also from \eqref{a4}, that if $v=o(s^{-1})$, then $\det(I-\gamma K_s)\rightarrow 1$ as $s\rightarrow\infty$, and very few particles are expelled from $(-\frac{s}{\pi},\frac{s}{\pi})$. And if $v\sim s^{-1}$, then $\det(I-\gamma K_s)\geq\textnormal{const}>0,s\rightarrow\infty$. Thus in order for the potential to have a noticeable effect on the fluid, one must have $vs\rightarrow\infty$, i.e. $\kappa s^2\rightarrow\infty$.\bigskip

In \cite{Dy1}, Dyson considered $\det(I-\gamma K_s)$ as $s,v\rightarrow\infty$ subject to $0<\kappa=\frac{v}{s}<1$. In this limit he showed that 
\begin{equation}\label{Dyres}
	\det(I-\gamma K_s)\sim C(v)\theta_4\big(sV(\kappa)|\,\tau(\kappa)\big)\exp\left[-\frac{s^2}{2}\big(1-a^2(\kappa)\big)+vsV(\kappa)\right].
\end{equation}
Here $a=a(\kappa)$ is the unique solution of equation
\begin{equation}\label{Dya}
	\int_a^1\sqrt{\frac{\mu^2-a^2}{1-\mu^2}}\,\d\mu = \kappa,
\end{equation}
and $\theta_4=\theta_4(z|\,\tau)$ is the fourth Jacobi theta function
\begin{equation*}
	\theta_4(z|\,\tau) = 1+2\sum_{k=1}^{\infty}(-1)^ke^{i\pi k^2\tau}\cos(2\pi kz)
\end{equation*}
with module 
\begin{equation}\label{Dymod}
	\tau(\kappa) = 2c(\kappa)\int_{-a}^a\frac{\d\mu}{\sqrt{(a^2-\mu^2)(1-\mu^2)}}\in i\mathbb{R}_+, \ \ \ 
	c(\kappa)=\frac{i}{2}\left(\int_a^1\frac{\d\lambda}{\sqrt{(1-\lambda^2)(\lambda^2-a^2)}}\right)^{-1}.
\end{equation}
Furthermore,
\begin{equation}\label{Dyfreq}
	V=V(\kappa) = -\frac{1}{\pi}\int_{-a}^{a}\sqrt{\frac{a^2-\mu^2}{1-\mu^2}}\,\d\mu
\end{equation}
and $C(v)$ is a constant, independent of $s$. Dyson analyzed $\det(I-\gamma K_s)$ by expanding formally the LHS of \eqref{Dcorr} to obtain the formula
\begin{equation*}
	\det(I-\gamma K_s) = \frac{Z(v,s)}{Z(0,s)}
\end{equation*}
where
\begin{equation}\label{part}
	Z(v,s) = \sum\exp\left[-(W+2v\sharp)\right]
\end{equation}
is the partition function of the above Coulomb gas $\{x_j\}$ with external potential $\mathcal{V}(x)=2v$ for $x\in(-\frac{s}{\pi},\frac{s}{\pi})$ and zero otherwise. Here $\sharp=\sharp(x)$ is the number of particles from the configuration $x=\{x_j\}$ that lie in $\left(-\frac{s}{\pi},\frac{s}{\pi}\right)$. In \eqref{part}, the sum is taken over all the Coulombic interactions
\begin{equation*}
	W=-\sum_{j,k}\ln|x_j-x_k|
\end{equation*}
of the particles $\{x_j\}$ in the configuration. In \cite{Dy1}, Dyson did not attempt to give \eqref{part} a rigorous meaning, but in analogy with standard thermodynamic/statistical mechanical calculations, he observed that for large $s$ (or $v$), $Z(v,s)$ may be approximated by $\exp(-\Phi)$ where $\Phi$ is the free energy
\begin{equation}\label{Dapprox}
	\Phi=\inf(W+2v\sharp)
\end{equation}
and the infimum is taken over all the configurations of the gas. He then showed how to solve (a suitably regularized version of) \eqref{Dapprox} explicitly, and this provided the basis for his calculations of the asymptotics of $\det(I-\gamma K_s)$. We will comment further on Dyson's calculation following Remark \ref{Dycorr} below.\smallskip

It is not at all clear how to make Dyson's calculations rigorous and we will take a different approach expressing $\det(I-\gamma K_s)$ in terms of a Riemann-Hilbert problem (RHP). This approach uses the fact that the operator $K_s$ is {\it integrable}, cf. \cite{IIKS} or \cite{DIZ}. We write
\begin{equation}\label{IIKS1}
	\gamma K_s(\lambda,\mu) = \frac{f^t(\lambda)h(\mu)}{\lambda-\mu},\hspace{0.5cm}f(\lambda)=\sqrt{\frac{\gamma}{2\pi i}}\binom{e^{is\lambda}}{e^{-is\lambda}},\ \ \ h(\lambda) = \sqrt{\frac{\gamma}{2\pi i}}\binom{e^{-is\lambda}}{-e^{is\lambda}}
\end{equation}
and note by general theory that the resolvent operator $(I-\gamma K_s)^{-1}$, is again of integrable type. Moreover, by general theory, the kernel of $(I-\gamma K_s)^{-1}$ can be constructed in terms of the solution of the following RHP.
\begin{problem}[Master RHP]\label{def1}
Determine the $2\times 2$ matrix valued, piecewise analytic function $Y=Y(\lambda;s,\gamma)$ such that
\begin{itemize}
	\item $Y(\lambda)$ is analytic for $\lambda\in\mathbb{C}\backslash[-1,1]$
	\item Along the line segment $[-1,1]$, oriented from left to right,
	\begin{equation}\label{Yjump}
		Y_+(\lambda) = Y_-(\lambda)\begin{pmatrix}
		1-\gamma & \gamma e^{2i\lambda s}\\
		-\gamma e^{-2i\lambda s} & 1+\gamma\\
		\end{pmatrix},\ \ \lambda\in(-1,1).
	\end{equation}
	\item $Y(\lambda)$ has at most logarithmic singularities at $\lambda=\pm 1$. More precisely, 
	\begin{equation}\label{Ysing}
		Y(\lambda) = \check{Y}(\lambda)\left[I+\frac{\gamma}{2\pi i}\begin{pmatrix}
		-1 & 1\\
		-1 & 1\\
		\end{pmatrix}\ln\left(\frac{\lambda-1}{\lambda+1}\right)\right]e^{-is\lambda\sigma_3},\hspace{0.5cm}\lambda\rightarrow\pm 1,\ \ \ \sigma_3=\begin{pmatrix}
		1 & 0\\
		0 & -1
		\end{pmatrix},
	\end{equation}
	where $\check{Y}(\lambda)$ is analytic at $\lambda=\pm 1$, and we choose the principal branch of the logarithm.
	\item As $\lambda\rightarrow\infty$,
	\begin{equation*}
		Y(\lambda)= I+\mathcal{O}\big(\lambda^{-1}\big),\ \ \lambda\rightarrow\infty.
	\end{equation*}
\end{itemize}
\end{problem}
This RHP is uniquely solvable if and only if $(I-\gamma K_s)^{-1}$ exists (see again \cite{IIKS,DIZ}), and its solution can be written in terms of the Cauchy integral
\begin{equation}\label{Yintegraleq}
	Y(\lambda) = I-\int_{-1}^1\big(Y(w)f(w)\big)h^t(w)\frac{dw}{w-\lambda} = I+\frac{Y_1}{\lambda}+\mathcal{O}\left(\lambda^{-2}\right),\hspace{0.5cm}\lambda\rightarrow\infty,\ \ \ Y_1=\big(Y_1^{jk}\big).
\end{equation}
Let 
$a=a(\kappa),\tau=\tau(\kappa),V=V(\kappa)$ be as in \eqref{Dya},\eqref{Dymod},\eqref{Dyfreq} and $\theta_3(z|\,\tau) = \theta_4(z-\frac{1}{2}|\,\tau)$ be the third Jacobi theta function. In addition, let
\begin{equation}\label{Mformula}
	M(x,\kappa) = \frac{\Xi_0(x,\kappa)\Theta_0(x,\kappa)+6a(\kappa)\,\Xi_2(z,\kappa)}{48a(\kappa)(1+a(\kappa))}+\frac{i}{4\pi}\frac{\kappa}{\theta_3(x|\tau)}\frac{\partial^2}{\partial y^2}\theta_3(y|\tau)\big|_{y=x}\frac{\partial\tau}{\partial\kappa}
\end{equation}
where $\Xi_j(x,\kappa),j=0,2$ and $\Theta_0(x,\kappa)$ are certain explicit functions of the Jacobi theta functions $\theta_j(x|\,\tau),j=0,1,2,3$ (see Appendix \ref{thetaid}) given in \eqref{abbrev1} and \eqref{ell1} below. In our notation, $\theta_0(x|\tau)\equiv \theta_4(x|\tau)$. 
 We obtain the following.
\begin{theo} \label{theo1} For any fixed $\delta\in(0,1)$, there exist positive constants $s_0=s_0(\delta)$, $c=c(\delta)$ and $C_0=C_0(\delta)$ such that
\begin{eqnarray}\label{theo1res}
	\ln\det(I-\gamma K_s)&=&-\frac{s^2}{2}\big(1-a^2(\kappa)\big)+vsV(\kappa)+\ln\theta_3\big(sV(\kappa)|\,\tau(\kappa)\big)+A(v)\\
	&&+\int_s^{\infty}\!\!M\Big(tV\left(vt^{-1}\right),vt^{-1}\Big)\frac{\d t}{t}+J(s,v),\nonumber
\end{eqnarray}
where
\begin{equation}\label{Avexp}
	A(v) = 2\ln b(v)-\frac{v^2}{\pi^2}\left(3+2\ln\left(\frac{\pi}{v}\right)\right),
\end{equation}
with $b(v)$ as in \eqref{BWBB2} and
\begin{equation*}
	\left|\int_s^{\infty}\!\!M\Big(tV\left(vt^{-1}\right),vt^{-1}\Big)\frac{\d t}{t}\right|\leq C_0,\hspace{0.75cm}\big|J(s,v)\big|\leq cs^{-\frac{1}{4}}\ln s,
\end{equation*}
for $s\geq s_0$ and $0<v\leq s(1-\delta)$. 
\end{theo}
\begin{rem}\label{seq} The integrability of $\frac{1}{t}M(tV(vt^{-1}),vt^{-1})$ for fixed $v>0$ as $t\rightarrow\infty$, is proved in Proposition \ref{integrab} below. In the sequel \cite{BDIK} to the current paper we will show that in fact,
\begin{equation*}
	\int_s^{\infty}\!\!M\left(tV\big(vt^{-1}\big),vt^{-1}\right)\frac{\d t}{t}=A_0\big(vs^{-1}\big)+\mathcal{O}\left(s^{-1}\right)
\end{equation*}
uniformly for $s\geq s_0$ and $0<v\leq s(1-\delta)$. Here,
\begin{equation}\label{Masymp2}
	A_0\big(vs^{-1}\big)=\int_0^{\frac{v}{s}}a_0(u)\frac{\d u}{u}
\end{equation}
and $a_0(u)$ is an average of $M(x,u)$ over the ``fast" variable $x$, i.e.
\begin{equation*}
	a_0(u)=\int_0^1M(x,u)\d x.
\end{equation*}
It is also shown that $a_0(u)=\mathcal{O}(u)$ as $0\leq u\leq1-\delta$, and hence the integral in \eqref{Masymp2} exists. Observe that $A_0(vs^{-1})$ depends on the ``slow" variable $\kappa=\frac{v}{s}$, and $A_0(vs^{-1})\rightarrow 0$ if $v=o(s)$.
\end{rem}
\begin{rem} The explicit form \eqref{Avexp} for the ``constant" term $A(v)$ is derived in Section \ref{A:int} below using only the known asymptotics \eqref{BWBB1} for fixed $v$.
\end{rem}
\begin{rem} Theorem \ref{theo1} fully describes the transition of the asymptotics from the case of fixed $v$ given by \eqref{BWBB1} to the case $v\leq s(1-\delta)$. Indeed, using the small $\kappa$ expansions for $a(\kappa), V(\kappa), \tau(\kappa)$ (see \eqref{aasy},\eqref{c:1},\eqref{c:2} below) and the estimate \eqref{integrab112} for $M$, we recover \eqref{BWBB1} from \eqref{theo1res} (with a worse estimate for the error term).
\end{rem}
\begin{rem}\label{oscidie} For $0<\delta\leq\kappa=\frac{v}{s}\leq 1-\delta$, $V(\kappa)$ and $-i\tau(\kappa)$ vary over compact intervals in $(0,\infty)$, and we see from \eqref{theo1res} that $\det(I-\gamma K_s)$ experiences non-trivial modulated oscillations. But as $\kappa\downarrow 0,V\downarrow-\frac{2}{\pi}$ and $-i\tau\rightarrow+\infty$ (see Corollary \ref{cor1}), and so the oscillations become strictly periodic, but the amplitude of the oscillations die out exponentially. On the other hand, as $\kappa\uparrow 1$, $V\uparrow 0$ and $-i\tau\downarrow 0$ (see again Corollary \ref{cor1}), and so the oscillations have longer and longer periods. Thus in both cases, the oscillatory behavior of $\det(I-\gamma K_s)$ is lost, but the mechanisms are different.
\end{rem}
\begin{rem} In the calculations below we will often use the notation $\theta(y|\tau)\equiv \theta_3(y|\tau)$ without further comment.
\end{rem}
\begin{rem}\label{Dycorr} Note that we obtain the third Jacobi theta function in \eqref{theo1res} as opposed to $\theta_4(z)$ in \eqref{Dyres}. In his calculations, Dyson arrives at the oscillatory term $\theta_4(sV+c')$, where the constant $c'$ is to be determined. Recalling the relation between $\theta_3(z)$ and $\theta_4(z)$ (see \eqref{conn}), we see that if $c'=\frac{1}{2}$, then Dyson's $\theta_4$ oscillation in \eqref{Dyres} would become our $\theta_3$ oscillation. Dyson determines $c'$ by comparing two asymptotic expansions for a quantity $B_0$ (expansion $(3.38)$ versus $(3.39)$ in \cite{Dy1}) and finds $c'=0$. But a careful examination of Dyson's argument shows that he neglects certain logarithmically growing terms and so, in particular, lower order terms such as the constant $c'$ cannot be determined by this argument. We believe that this is the source of the discrepancy between Theorem \ref{theo1} and \eqref{Dyres}.
\end{rem}
Dyson solves the minimization problem \eqref{Dapprox} by replacing the Coulomb gas $\{x_j\}$ by a continuum of charges, a so-called ``Coulombic fluid", with density $\rho(x)$: The continuum minimization problem can then be solved explicitly using elementary methods from complex variables. The Coulombic fluid approximation only yields the leading asymptotics for $\det(I-\gamma K_s)$ as $s$ and/or $v\rightarrow\infty$. The oscillations in \eqref{Dyres},\eqref{theo1res} occur at lower order and are obtained by Dyson by ``weaving together" the fluid minimizer together with certain identities of Mehta from random matrix theory. Regarding the oscillations Dyson makes the following observation (see p. $16$ in \cite{Dy1})
\begin{quote}
	``The oscillations are a manifestation of the discrete nature of the Coulomb gas. They do not appear in the Coulomb fluid approximation. We may picture the Coulomb gas responding as $t\, (=\frac{s}{\pi}$ in our scaling) increases with a discrete bump each time a single charge moves across the gap from outside to inside the interval $[-t,t]$."\smallskip
\end{quote} 
One can make this observation more concrete in the following way: We have for any $s$, with $\mathcal{V}=2v\chi_{(-\frac{s}{\pi},\frac{s}{\pi})}$
\begin{eqnarray}
	\det(I-\gamma K_s) = \mathbb{E}\left(e^{-\sum_j\mathcal{V}(x_j)}\,\right)&=&\sum_{n=0}^{\infty}\,\mathbb{E}\left\{e^{-\mathcal{V}}:\ \textnormal{number of particles in}\ \left(-\frac{s}{\pi},\frac{s}{\pi}\right)=n\right\}\nonumber\\
	&=&\sum_{n=0}^{\infty}e^{-2nv}p_n(s)\label{eq:1}
\end{eqnarray}
where 
\begin{equation*}
	p_n(s) = \textnormal{Prob}\,\left\{ n\ \textnormal{particles in}\ \left(-\frac{s}{\pi},\frac{s}{\pi}\right)\right\} = \textnormal{Prob}\,\left\{\sharp(x)=n\right\}
\end{equation*}
and again $\sharp(x)$ is the number of particles in the configuration $x=(x_1,x_2,\ldots)$ that lie in $(-\frac{s}{\pi},\frac{s}{\pi})$. As noted by Dyson in \cite{Dy1} the relation \eqref{eq:1} is just the Taylor series expansion of $\det(I-\gamma K_s)$ about $\gamma=1$ in powers of $\gamma-1=-e^{-2v}$, using the standard fact (see e.g. \cite{Me}) that
\begin{equation}\label{pd}
	p_n(s) = \frac{(-1)^n}{n!}\frac{\partial^n}{\partial\gamma^n}\det(I-\gamma K_s)\Big|_{\gamma=1}.
\end{equation}
Differentiating $\mathbb{E}_s(e^{-\mathcal{V}}) = \mathbb{E}_s(e^{-2v\sharp})$ with respect to $v$ in \eqref{eq:1} at $v=0$ yields
\begin{equation*}
	-2\mathbb{E}_s(\sharp) = -2\sum_{n=0}^{\infty}np_n(s),
\end{equation*}
and so,
\begin{equation}\label{eq:2}
	\frac{2s}{\pi} = \mathbb{E}_s(\sharp) = \sum_{n=0}^{\infty}np_n(s).
\end{equation}
Now if we increase $s$ by an amount $\Delta s=\frac{\pi}{2}$, on the average, one particle enters $(-\frac{s}{\pi},\frac{s}{\pi})$. Thus we make the reasonable ansatz that
\begin{equation}\label{eq:3}
	p_{n+1}\left(s+\frac{\pi}{2}\right) = p_n(s).
\end{equation}
Substituting this relation into \eqref{eq:2} yields
\begin{eqnarray*}
	\frac{2s}{\pi}+1&=&\sum_{n=0}^{\infty}np_n\left(s+\frac{\pi}{2}\right) = \sum_{n=1}^{\infty}np_n\left(s+\frac{\pi}{2}\right) = \sum_{n=0}^{\infty}(n+1)p_{n+1}\left(s+\frac{\pi}{2}\right)\\
	&=&\sum_{n=0}^{\infty}(n+1)p_n(s) = \frac{2s}{\pi}+\sum_{n=0}^{\infty}p_n(s) = \frac{2s}{\pi}+1,
\end{eqnarray*}
which shows that \eqref{eq:3} is consistent, at this level. Now substituting \eqref{eq:3} into \eqref{eq:1}, we obtain
\begin{eqnarray*}
	\mathbb{E}_{s+\frac{\pi}{2}}\left(e^{-\mathcal{V}}\right) &=& \sum_{n=0}^{\infty}e^{-2nv}p_n\left(s+\frac{\pi}{2}\right) = p_0\left(s+\frac{\pi}{2}\right)+\sum_{n=1}^{\infty}e^{-2nv}p_n\left(s+\frac{\pi}{2}\right)\\
	&=&p_0\left(s+\frac{\pi}{2}\right)+e^{-2v}\sum_{n=0}^{\infty}e^{-2nv}p_n(s)= p_0\left(s+\frac{\pi}{2}\right)+e^{-2v}\,\mathbb{E}_s\left(e^{-\mathcal{V}}\,\right)
\end{eqnarray*}
and therefore
\begin{equation*}
	e^{\frac{4v}{\pi}(s+\frac{\pi}{2})}\mathbb{E}_{s+\frac{\pi}{2}}\left(e^{-\mathcal{V}}\,\right) = e^{\frac{4vs}{\pi}}\mathbb{E}_s\left(e^{-\mathcal{V}}\,\right)+e^{\frac{4v}{\pi}(s+\frac{\pi}{2})}p_0\left(s+\frac{\pi}{2}\right).
\end{equation*}
However for a fixed $v$, by \eqref{DW} and \eqref{pd},
\begin{equation}\label{eq:4}
	e^{\frac{4v}{\pi}(s+\frac{\pi}{2})}p_0\left(s+\frac{\pi}{2}\right)=\mathcal{O}\left(e^{-\frac{s^2}{2}+\frac{4vs}{\pi}}\right)
\end{equation}
as $s\rightarrow\infty$. Hence to leading order, we see that
\begin{equation}\label{eq:5}
	Q_v(s) = e^{\frac{4vs}{\pi}}\det(I-\gamma K_s)
\end{equation}
is periodic with period $\frac{\pi}{2}$. Now as $V\sim -\frac{2}{\pi}$ for $\kappa\downarrow 0$ (see Corollary \ref{cor1} below), this matches with the fact that $sV\sim-\frac{2s}{\pi}$ and $\theta_3(x|\,\tau)$ has period one! Also, using \eqref{aasy} below,
\begin{equation*}
	-\frac{s^2}{2}\left(1-a^2(\kappa)\right)+vsV(\kappa)\sim -\frac{4sv}{\pi},\ \ \ \kappa\downarrow 0
\end{equation*}
which completes the matching of \eqref{eq:5} with \eqref{Dyres} as $s\rightarrow\infty$ in the small $\kappa$ region.\smallskip

Notice that the error term in \eqref{eq:4} is small only if $\frac{4vs}{\pi} = \frac{4\kappa s^2}{\pi}$ is small compared to $s^2$. Certainly, as $\kappa\uparrow 1$, this fails. Furthermore, as noted above (see Remark \ref{oscidie}), the frequency in $\theta_3(sV|\tau)$ goes to zero as $\kappa\uparrow 1$.\bigskip
\begin{rem}\label{bohiPa} Bohigas and Pato in \cite{BP} interpret \eqref{eq:1} as the gap probability of a random particle system obtained from the Coulomb gas by dropping at random a fraction $1-\gamma$ of levels. As $\gamma\downarrow 0$, the gap probability becomes Poissonian, and so this particle system interpolates between Poisson statistics ($\gamma\downarrow 0$) and random matrix statistics ($\gamma\uparrow 1$). More precisely, for this new system,
\begin{align*}
	\textnormal{Prob}\,\left\{\textnormal{no particles in}\,\left(-\frac{s}{\pi},\frac{s}{\pi}\right)\right\}=&\sum_{n=0}^{\infty}\,\bigg[\textnormal{Prob}\,\left\{n\,\,\textnormal{Coulomb particles in}\,\left(-\frac{s}{\pi},\frac{s}{\pi}\right)\right\}\\
	&\hspace{1cm}\times\,\textnormal{Prob}\,\left\{\textnormal{these}\,\,n\,\,\textnormal{particles are removed at random}\right\}\bigg]\\
	=&\sum_{n=0}^{\infty}p_n(s)(1-\gamma)^n=\det(I-\gamma K_s).
\end{align*}
In order to keep to unity the average level density of the remaining fraction $\gamma$ of the particles, the scale is contracted accordingly. Thus Bohigas and Pato consider
\begin{equation*}
	P(\gamma,s)=\textnormal{Prob}\,\left\{\textnormal{no particles in}\,\left(-\frac{s}{\gamma\pi},\frac{s}{\gamma\pi}\right)\right\} = \det\Big(I-\gamma K_{\frac{s}{\gamma}}\Big).
\end{equation*}
It follows from \eqref{BBrefined} that for $s>s_0$, at leading order as $\gamma=1-e^{-2v}\downarrow 0$,
\begin{equation}\label{B.1}
	P(\gamma,s)\sim \exp\left[-\frac{4v}{\pi}\frac{s}{\gamma}\right]\sim\exp\left[-\frac{2s}{\pi}\right]
\end{equation}
(in \cite{BP}, Bohigas and Pato use the fixed $v$ result \eqref{BWBB1},\eqref{BWBB2} to derive \eqref{B.1}, but as $\gamma$ is varying, a rigorous analysis requires \eqref{BBrefined}). But \eqref{B.1} is precisely the gap probability for the particle system obtained from a Poisson process, again by dropping a fraction $1-\gamma$ of levels. Indeed, for such a system the gap probability is given by
\begin{equation*}
	\sum_{k=0}^{\infty}\left[\frac{1}{k!}\left(\frac{2s}{\pi\gamma}\right)^ke^{-\frac{2s}{\pi\gamma}}\right](1-\gamma)^k = e^{-\frac{2s}{\pi\gamma}}e^{\frac{2s}{\pi\gamma}(1-\gamma)} = e^{-\frac{2s}{\pi}}
\end{equation*}
(in \cite{BP}, $\frac{2s}{\pi}\leftrightarrow s$). Thus the new system interpolates between a Poisson particle system ($\gamma\downarrow 0$) and random matrix theory where at $\gamma=1$, $P(\gamma,s)\sim e^{-\frac{s^2}{2}}$.
\end{rem}

The proof of \eqref{theo1res} involves an analysis of $\det(I-\gamma K_s)$ for values of $(s,v),s\geq s_0$, in two complementary regions
\begin{enumerate}[leftmargin=4cm]
	\item[(i)]\ \ $s^{1-\epsilon}\leq v\leq (1-\delta)s$
\end{enumerate}
and
\begin{enumerate}[leftmargin=4cm]
	\item[(ii)]\ \ $0<v<s^{\frac{1}{3}}$
\end{enumerate}
where $0<\epsilon<1$. In this paper, Part I, we analyze $\det(I-\gamma K_s)$ in region (i) by applying the steepest-descent method to the above $Y$-RHP. In the sequel to this paper, Part II, already mentioned above (see \cite{BDIK}), we analyze $\det(I-\gamma K_s)$ in region (ii) in two different ways:
\begin{enumerate}
	\item[(ii)(a)] by applying the steepest-descent method to the above $Y$-RHP, but now the steepest-descent deformation of the RHP is very different from the deformation in this paper	\item[(ii)(b)] by expressing $\det(I-\gamma K_s)$ as a continuum limit of a Toeplitz determinant with Fisher-Hartwig singularities (see e.g.\,\cite{DIII}), and controlling the limit appropriately using RH techniques (cf.\,\cite{CK}) for $(s,v)$ satisfying $0<v<s^{\frac{1}{3}}$.
\end{enumerate}
These two approaches, (ii)(a) and (ii)(b), yield slightly different estimates and provide contrasting views on the asymptotic behavior of $\det(I-\gamma K_s)$. The reader will note that in Section \ref{error} below (see also Remark \ref{Busref}), we use estimates from Part II for $(s,v)$ in region (ii) to obtain the overall bound $|J(s,v)|\leq cs^{-\frac{1}{4}}\ln s,v\leq(1-\delta)s$, in Theorem \ref{theo1}. However, although the analysis of region (i) and the analysis of region (ii) both use steepest-descent/RH methods, they are very different, and it is natural to split the paper into two parts, I and II.\smallskip

In contrast with the region $v\leq(1-\delta)s$, we wish to point out that for $(s,v)$ above, on and slightly below the major diagonal in Figure \ref{Deiftdia}, $\det(I-\gamma K_s)$ can be analyzed using general functional analytical tools together with the known asymptotics \eqref{DW}, but without any further Riemann-Hilbert analysis.\bigskip

Indeed, the derivation of \eqref{a3} shows that for any positive definite, trace class operator $B$, say,
\begin{equation*}
	1\leq \det(I+B)=\exp\left[\int_0^1\textnormal{tr}\,\left((I+\mu B)^{-1}B\right)\,d \mu\right]\leq e^{\textnormal{tr}\,B}\leq 1+(\textnormal{tr}\, B)e^{\textnormal{tr}\,B},
\end{equation*}
i.e.
\begin{equation}\label{Lids1}
	\big|\det(I+B)-1\big|\leq (\textnormal{tr}\,B)e^{\textnormal{tr}\,B}\leq \|B\|_1\exp(\|B\|_1).
\end{equation}
Now let $s\geq s_0,v\geq v_0$ be such that $\kappa\geq 1-\chi(\frac{\ln s}{s})$  where $\chi$ is a (real) constant. Choose an integer $q=q(\chi)$ such that $q=1$ if $\chi< \frac{1}{4}$ and $q>2\chi+\frac{1}{2}$ if $\chi\geq\frac{1}{4}$, and write 
\begin{equation*}
	K_q = K_s\cdot P_s^q,
\end{equation*}
where $P_s^q$ is the projection onto the space of the eigenvectors of $K_s$ with associated eigenvalues $\{\lambda_j:\ j\geq q\}$.  We have
\begin{equation*}
	\frac{\det(I-\gamma K_s)}{\det(I-K_s)} =\det\left(I+e^{-2v}K_s(I-K_s)^{-1}\right)
\end{equation*}
and by Lidskii's theorem
\begin{eqnarray*}
	\det\left(I+e^{-2v}K_s(I-K_s)^{-1}\right)&=&\prod_{j=0}^{\infty}\left(1+e^{-2v}\frac{\lambda_j}{1-\lambda_j}\right) \\
	&=& \prod_{j=0}^{q-1}\left(1+e^{-2v}\frac{\lambda_j}{1-\lambda_j}\right)\det\left(I+e^{-2v}K_q(I-K_q)^{-1}\right)
\end{eqnarray*}
Now by \eqref{Lids1}
\begin{equation*}
	\big|\det\left(I+e^{-2v}K_q(I-K_q)^{-1}\right)-1\big|\leq\|e^{-2v}K_q(I-K_q)^{-1}\|_1\exp\left[\|e^{-2v}K_q(I-K_q)^{-1}\|_1\right]
\end{equation*}
and by \eqref{a1},
\begin{equation*}
	\|e^{-2v}K_q(I-K_q)^{-1}\|_1\leq e^{-2v}\|K_q\|_1\|(I-K_q)^{-1}\|\leq e^{-2v}\frac{\|K_q\|_1}{1-\lambda_q}\leq C_q e^{-2(v-s)}s^{-q+\frac{1}{2}}\big(1+\mathcal{O}\left(s^{-1}\right)\big).
\end{equation*}
Since by assumption
\begin{equation*}
	e^{-2(v-s)}s^{-q+\frac{1}{2}} \leq s^{-(q-2\chi-\frac{1}{2})}=o(1),\  
\end{equation*}
we obtain
\begin{equation*}
	\det(I-\gamma K_s)=\prod_{j=0}^{q-1}\left(1+e^{-2v}\frac{\lambda_j}{1-\lambda_j}\right)\Big(1+\mathcal{O}\left(s^{-(q-2\chi-\frac{1}{2})}\right)\Big)\det(I-K_s).
\end{equation*}
The point is that in the last equation, the finite product for $s\geq s_0$ is known explicitly from \eqref{a1},
\begin{equation*}
	1+e^{-2v}\frac{\lambda_j}{1-\lambda_j}=1+\frac{j!}{\sqrt{\pi}}2^{-3j-2}e^{-2(v-s)}s^{-j-\frac{1}{2}}\left(1+\mathcal{O}\left(s^{-1}\right)\right),\ \ \ 0\leq j\leq q-1
\end{equation*}
and it depends explicitly on our precise assumptions on $v-s=-s(1-\kappa)$. We have therefore proved our second main result:
\begin{theo}\label{theo2} There exist positive constants $s_0=s_0(\chi),v_0=v_0(\chi)$ such that
\begin{equation*}
	\det(I-\gamma K_s)=\exp\left[-\frac{s^2}{2}\right]s^{-\frac{1}{4}}c_0\prod_{j=0}^{q-1}\left(1+e^{-2v}\frac{\lambda_j}{1-\lambda_j}\right)\Big(1+\mathcal{O}\left(\max\left\{s^{-(q-2\chi-\frac{1}{2})},s^{-1}\right\}\right)\Big)
\end{equation*}
uniformly for $s\geq s_0,v\geq v_0$ and $\kappa\geq 1-\chi(\frac{\ln s}{s})$ for any $\chi\in\mathbb{R}$. The universal constant term $c_0$ is given in \eqref{DWconstant} and $q=q(\chi)\in\mathbb{Z}_{\geq 1}$ was introduced above.
\end{theo}
As a Corollary to the latter Theorem, we state that for $\chi<\frac{1}{4}$ the asymptotic behavior of $\det(I-\gamma K_s)$ is unchanged to leading order from the case $v=\infty$ given in \eqref{DW}
\begin{cor}\label{cor1} There exist positive constants $s_0=s_0(\chi),v_0=v_0(\chi)$ such that
\begin{equation*}
	\det(I-\gamma K_s) = \exp\left[-\frac{s^2}{2}\right]s^{-\frac{1}{4}}c_0\Big(1+\mathcal{O}\left(\max\left\{s^{-(\frac{1}{2}-2\chi)},s^{-1}\right\}\right)\Big)
\end{equation*}
uniformly for $s\geq s_0,v\geq v_0$ and $\kappa\geq 1-\chi(\frac{\ln s}{s})$ for any $\chi<\frac{1}{4}$.
\end{cor}
In Theorem \ref{theo2}, suppose $\chi>0$ and $\kappa=\frac{v}{s} = 1-\chi(\frac{\ln s}{s})$. Let $q\geq 1$ such that
\begin{equation}\label{Lids2}
	2\chi+\frac{3}{2}\geq q>2\chi+\frac{1}{2}.
\end{equation}
Then, as $0<q-2\chi-\frac{1}{2}\leq 1$, it follows from Theorem \ref{theo2}, and also \eqref{a1} for $n=q-1$, that as $s\rightarrow\infty$,
\begin{equation}\label{Lids3}
	\det(I-\gamma K_s) = \exp\left[-\frac{s^2}{2}\right]s^{-\frac{1}{4}}c_0\prod_{j=0}^{q-2}\left(1+e^{-2v}\frac{\lambda_j}{1-\lambda_j}\right)\left(1+\mathcal{O}\left(s^{-(q-2\chi-\frac{1}{2})}\right)\right).
\end{equation}
Here $\prod_{j=0}^{q-2}(\ldots)\equiv 1$ if $q=1$. For $0\leq j\leq q-2$, $e^{-2v}(1-\lambda_j)^{-1}\sim e^{(2\chi-j-\frac{1}{2})\ln s}\geq e^{(2\chi-q+\frac{3}{2})\ln s}$ and so by \eqref{Lids2}, all the terms in the product contribute to $\det(I-\gamma K_s)$ as $s\rightarrow\infty$. We see that the curves
\begin{equation}\label{Lids4}
	v=s-\chi \ln s,\ \ \ \chi=\frac{q}{2}-\frac{1}{4},\ \ q=1,2,\ldots
\end{equation}
are "Stokes lines" for $\det(I-\gamma K_s)$: The asymptotics of $\det(I-\gamma K_s)$ changes each time one of these lines is crossed, see Figure \ref{Deiftdia2}. The ``Stokes" structure of the asymptotics indicated in Theorem \ref{theo2} is a very interesting phenomenon which needs to be more deeply understood. The only similar type of asymptotic behavior known to the authors is the appearance of the ``soliton cascade" in the long-time asymptotic solution of the Cauchy problem for the KdV equation with step-like initial date (E. Khruslov, $1975$, see \cite{KK} for more details).
\begin{figure}[tbh]
\begin{center}
\includegraphics[width=0.43\textwidth]{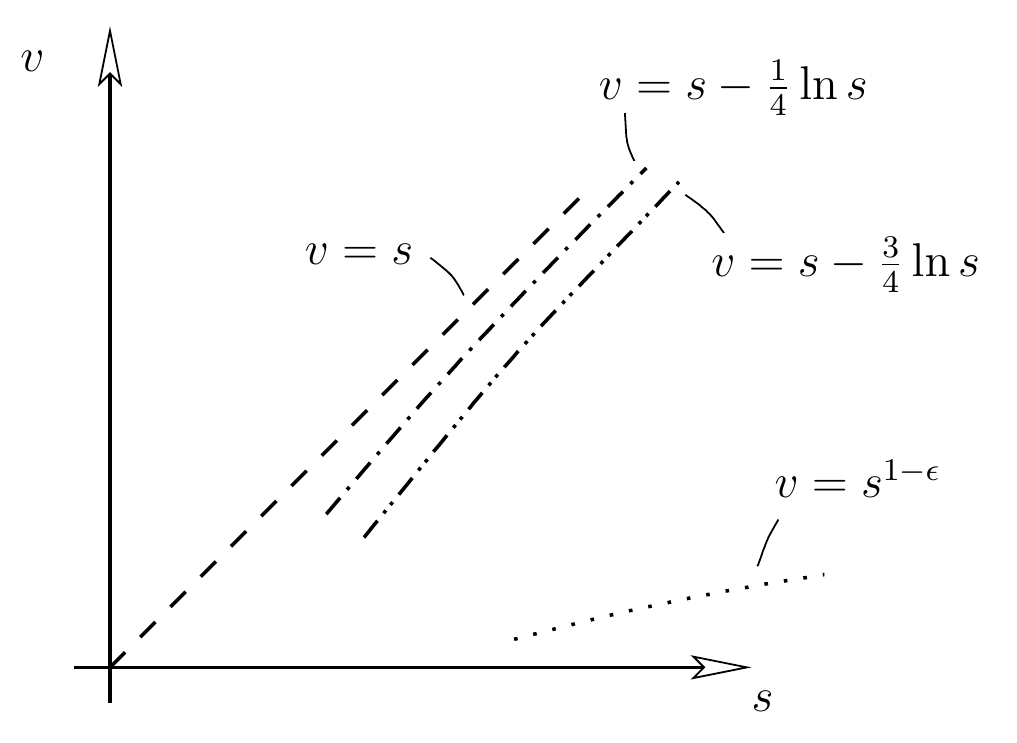}
\end{center}
\caption{"Stokes lines" displayed in the double scaling diagram}
\label{Deiftdia2}
\end{figure}

We conclude this introduction with a short outline for the remainder of the article, Part I. Our major technical tool in the asymptotic analysis of $\det(I-\gamma K_s)$ in the region (i), $s^{1-\epsilon}\leq v\leq(1-\delta)s$, is the nonlinear steepest decent analysis (see \cite{DZ,DVZ}) of the Riemann-Hilbert problem stated in Definition \ref{def1} (an analogous steepest-descent analysis also plays a key role in Part II). Here, a very important part is played by a $g$-function transformation which makes it possible to convert oscillatory jumps into exponentially decaying ones. Once this is achieved in Section \ref{sec1}, we turn our attention to the construction of model solutions to the initial master RHP. These solutions are given in terms of Jacobi theta functions on an appropriate elliptic Riemann surface as well as Airy and Bessel functions. However, as opposed to the more standard cases in \cite{DKMVZ,DIZ}, the latter parametrices involve certain "corrections": see Section \ref{sec2} for all details. The initial master RHP is then solved iteratively. We then use the crucial fact that the logarithmic derivative
\begin{equation*}
	\frac{\partial}{\partial s}\ln\det(I-\gamma K_s)
\end{equation*}
with $v$ fixed can be expressed in terms of the solution of the RHP \ref{def1} (see \eqref{IIKS2} below). In Section \ref{sec5} we substitute the main asymptotic terms of $Y(z)$ into this expression. This gives $\frac{\partial}{\partial s}\ln\det(I-\gamma K_s)$ under the constraint $\delta s\leq v\leq(1-\delta)s$. The integrability of $\frac{1}{t}M(tV(vt^{-1}),vt^{-1})$ for fixed $v>0,t\rightarrow\infty$ (see Remark \ref{seq}) is also established. In (the short) Section \ref{A:int}, formula \eqref{Avexp} for the integration constant $A(v)$ is derived, using (only) the fixed $v$-asymptotics \eqref{BWBB1}. In Section \ref{error} the estimate on $\frac{\partial}{\partial s}\ln\det(I-\gamma K_s)$ is extended to the full region (i),\,$s^{1-\epsilon}\leq v\leq (1-\delta)s$ for any $0<\epsilon<1$. Combined with the estimate on $\det(I-\gamma K_s)$ for $(s,v)$ in region (ii),\,$0<v<s^{\frac{1}{3}}$, we eventually derive the overall bound on $J$, $|J(s,v)|\leq c s^{-\frac{1}{4}}\ln s$, which allows us to complete the proof of Theorem \ref{theo1} by integrating $\frac{\partial}{\partial s}\ln\det(I-\gamma K_s)$ over the $s$-axis. Appendix \ref{thetaid} contains some identities for theta functions.\smallskip

Finally we note that it is an open, and very challenging, problem to determine the asymptotics of $\det(I-\gamma K_s)$ in the gap between the regions covered by Theorems \ref{theo1} and \ref{theo2}. Together with the estimates in region (ii), and the bound on the integral of $M$ in Remark \ref{seq}, we plan to address this problem in Part II.
%
%
%
%
%
%
%
%
%
%
%
%
%
\section{Riemann-Hilbert analysis - preliminary steps}\label{sec1}
The leading $s$ dependent terms in the statement of Theorem \ref{theo1} will be derived by integrating the well-known local identity (cf. \cite{DIZ})
\begin{equation}\label{IIKS2}
	\frac{\partial}{\partial s}\ln\det(I-\gamma K_s)=-\textnormal{tr}\,\left((I-\gamma K_s)^{-1}\frac{\partial}{\partial s}(\gamma K_s)\right)=-i\big(Y_1^{11}-Y_1^{22}\big)
\end{equation}
where $v$ is fixed and $Y_1=(Y_1^{jk})$ is the residue of $Y(\lambda)$ at $\lambda=\infty$ as stated in \eqref{Yintegraleq}.\smallskip

For certain (large) values $s_0,v_0$, we will, from now on, always assume that $s\geq s_0,v\geq v_0$ and in addition
\begin{equation}\label{dbscale}
	0<\delta\leq\kappa\leq1-\delta<1
\end{equation}
where $0<\delta<\frac{1}{2}$ remains fixed. Furthermore we fix the following notation. Set $J=(-1,-a)\cup(a,1)$ with $0<a<1$ given implicitly via \eqref{Dya}, i.e. 
\begin{equation*}
	\kappa = \frac{v}{s} = \int_a^1\sqrt{\frac{\mu^2-a^2}{1-\mu^2}}\,\d\mu.
\end{equation*}
As we shall see shortly, the latter equation defines a one-to-one correspondence between $\kappa\in[\delta,1-\delta]$ and $a\in[\delta',1-\delta']$. Now consider
the elliptic curve
\begin{equation*}
	\Gamma = \big\{(z,w):\ w^2=p(z)\big\},\hspace{0.5cm} p(z)=\left(z^2-1\right)\left(z^2-a^2\right),
\end{equation*}
of genus $g=1$. We use the representation of $\Gamma$ as two-sheeted covering of the Riemann sphere $\mathbb{CP}^1$ by gluing two copies of the slit plane $\mathbb{C}\backslash J$ together along $J$, see \cite{FK}. We let $\sqrt{p(z)}\sim z^2$ as $z\rightarrow\infty$ on the {\it first} sheet, and $\sqrt{p(z)}\sim-z^2$ as $z\rightarrow\infty$ on the {\it second} sheet.  Let $\{A_i\}_{i=0}^1$ and $\{B_1\}$ denote the cycles depicted in Figure \ref{figure1}, the cycles $A_i$ lie on the first sheet and the cycle $B_1$ passes from the first sheet through the slit $(-1,-a)$ to the second sheet, and back again through $(a,1)$. The cycles $\{A_1,B_1\}$ form a canonical homology basis for $\Gamma$.
\begin{figure}[tbh]
\begin{center}
\includegraphics[width=0.75\textwidth]{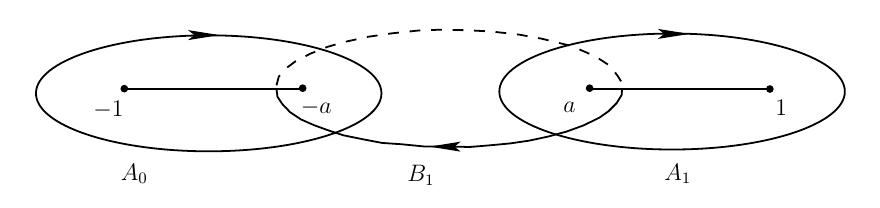}
\end{center}
\caption{Standard homology basis for $\Gamma$}
\label{figure1}
\end{figure}

One of the major ingredients in the solution of the $Y$-RHP is the use of a $g$-function transformation. Let
\begin{equation}\label{gfunc}
	g(z) = \int_1^z\sqrt{\frac{\mu^2-a^2}{\mu^2-1}}\,\d\mu,\hspace{0.5cm}z\in\mathbb{C}\backslash[-1,1]
\end{equation}
with the path of integration chosen in the simply connected domain $\mathbb{C}\mathbb{P}^1\backslash[-1,1]$ and 
\begin{equation*}
	-\pi<\textnormal{arg}\left(\frac{\mu^2-a^2}{\mu^2-1}\right)\leq\pi\ \ \textnormal{with}\ \ \sqrt{\frac{\mu^2-a^2}{\mu^2-1}}>0\ \ \textnormal{for}\ \ \mu>1.\footnote{Alternatively, we could define $g(z)$ on the first sheet $\mathbb{C}\backslash J$ of $\Gamma$.}
\end{equation*}
This implies that $g(z)$ is single-valued and analytic on $\mathbb{C}\mathbb{P}^1\backslash[-1,1]$ with the asymptotics
\begin{equation}\label{gfuncasy}
	g(z) = z+\ell-\frac{1-a^2}{2z}+\mathcal{O}\big(z^{-2}\big),\hspace{0.5cm}z\rightarrow\infty,\hspace{0.5cm} \ell=\int_1^{\infty}\left(\sqrt{\frac{\mu^2-a^2}{\mu^2-1}}-1\right)\d\mu-1.
\end{equation}
Let us now discuss the correspondence between the branch points $z=\pm a$ and the double scaling parameter $\kappa$ as defined in \eqref{Dya} 
\begin{prop} In the double-scaling limit $s\rightarrow\infty,\gamma\uparrow 1$ with $\kappa\in(0,1)$, the branch point $z=a$ is uniquely defined via \eqref{Dya}. In particular
\begin{equation}\label{aasy}
	a=1-\frac{2\kappa}{\pi}-\frac{\kappa^2}{\pi^2}+\mathcal{O}\big(\kappa^3\big),\hspace{0.5cm}\kappa\downarrow 0
\end{equation}
and
\begin{equation}\label{aasy2}
	a = 2\sqrt{\frac{1-\kappa}{|\ln(1-\kappa)|}}\left(1+\frac{\ln|\ln(1-\kappa)|}{2\ln(1-\kappa)}+\frac{1-\omega}{\ln(1-\kappa)}+\mathcal{O}\left(\frac{\ln^2|\ln(1-\kappa)|}{\ln^2(1-\kappa)}\right)\right),\hspace{0.5cm}\kappa\uparrow 1
\end{equation}	
with
\begin{equation*}
	\omega=\sum_{k=2}^{\infty}\binom{\frac{1}{2}}{k}\frac{(-1)^k}{k-1}=\frac{1}{2}-\ln 2 \approx -0.193147<0.
\end{equation*}
\end{prop}
\begin{proof} We consider the function
\begin{equation}\label{imp}
	F(a,\kappa) = \kappa-\int_a^1\sqrt{\frac{\mu^2-a^2}{1-\mu^2}}\,\d\mu,\hspace{1cm}a\in[0,1]
\end{equation}
which is real analytic in a neighborhood of the point $\left(a\in(0,1],\kappa\in[0,\infty)\right)$. Since
\begin{eqnarray}
	\int_a^1\sqrt{\frac{\mu^2-a^2}{1-\mu^2}}\,\d\mu &=& \sum_{k=0}^2\binom{\frac{1}{2}}{k}(-1)^k(1-a)^k\int_a^1\sqrt{\frac{\mu-a}{1-\mu}}\frac{\d\mu}{(1+\mu)^k}+\mathcal{O}\left((1-a)^4\right)\nonumber\\
	&=&\frac{\pi}{2}(1-a)\left(1-\frac{1-a}{4}-\frac{(1-a)^2}{16}+\mathcal{O}\big((1-a)^3\big)\right),\hspace{0.5cm} a\uparrow 1\label{aexp}\\
	\int_a^1\sqrt{\frac{\mu^2-a^2}{1-\mu^2}}\,\d\mu &=& \sum_{k=0}^{\infty}\binom{\frac{1}{2}}{k}(-1)^ka^{2k}\int_a^1\frac{\mu^{1-2k}}{\sqrt{1-\mu^2}}\,\d\mu\nonumber\\
	&=&1-\frac{a^2}{2}\left|\ln a\right|-\frac{a^2}{2}\left(1-\omega+\ln 2\right)+\frac{a^4}{16}\left|\ln a\right|+\mathcal{O}\big(a^4\big),\hspace{0.5cm} a\downarrow 0\label{aexp2}
\end{eqnarray}
we have 
\begin{equation*}
	\lim_{a\downarrow 0}F(a,\kappa) = \kappa-1,\hspace{0.5cm} \lim_{a\uparrow 1}F(a,\kappa) = \kappa>0
\end{equation*}
and the partial derivatives of $F$ are
\begin{equation*}
	F_a(a,\kappa) = a\int_a^1\frac{\d\mu}{\sqrt{(\mu^2-a^2)(1-\mu^2)}}>0,\hspace{1cm}F_{\kappa}(a,\kappa)=1.
\end{equation*}
Thus in the proposed double scaling limit, there is a real analytic solution $a(\kappa)$ of the equation $F(a,\kappa)=0$ in a neighborhood of the point $\left(a,\kappa\right)$. The expansions \eqref{aasy} and \eqref{aasy2} follow now directly from \eqref{imp} and \eqref{aexp}, \eqref{aexp2}.
\end{proof}
We also state expansions of the frequency $V=V(\kappa)$ \eqref{Dyfreq} as well as the module $\tau=\tau(\kappa)$ and the normalization constant $c=c(\kappa)$ \eqref{Dymod} which follow from \eqref{aasy}, \eqref{aasy2} and standard formulae for elliptic integrals. These expansions will be used later on.
\begin{cor}\label{cor1} As $\kappa\downarrow 0$, we have
\begin{eqnarray}
	V=V(\kappa)&=& -\frac{2}{\pi}\left(1+\frac{\kappa}{\pi}\ln\kappa -\frac{\kappa}{\pi}(1+\ln 4\pi)+\mathcal{O}\big(\kappa^2\ln\kappa\big)\right),\label{c:1}\\
	\tau=\tau(\kappa)&=&-\frac{2i}{\pi}\ln\frac{\kappa}{4\pi}+o(1),\label{c:2}\\
	c=c(\kappa)&=&\frac{i}{\pi}\left(1-\frac{\kappa}{\pi}+\mathcal{O}\left(\kappa^2\right)\right).\nonumber
\end{eqnarray}
On the other hand if $\kappa\uparrow 1$, then
\begin{eqnarray*}
	V=V(\kappa)&=&\frac{2(1-\kappa)}{\ln(1-\kappa)}\left(1+\frac{\ln|\ln(1-\kappa)|}{\ln(1-\kappa)}+\frac{2(1-\omega)}{\ln(1-\kappa)}+\mathcal{O}\left(\frac{\ln^2|\ln(1-\kappa)|}{\ln^2(1-\kappa)}\right)\right),\\
	\tau=\tau(\kappa)&=&-\frac{2\pi i}{\ln(1-\kappa)}\left(1+\frac{\ln|\ln(1-\kappa)|}{\ln(1-\kappa)}+\frac{\omega'}{\ln(1-\kappa)}+\mathcal{O}\left(\frac{\ln^2|\ln(1-\kappa)|}{\ln^2(1-\kappa)}\right)\right),\\
	c=c(\kappa)&=&-\frac{i}{\ln(1-\kappa)}\left(1+\frac{\ln|\ln(1-\kappa)|}{\ln(1-\kappa)}+\frac{\omega'}{\ln(1-\kappa)}+\mathcal{O}\left(\frac{\ln^2|\ln(1-\kappa)|}{\ln^2(1-\kappa)}\right)\right),
\end{eqnarray*}
with
\begin{equation*}
	\omega'=\sum_{k=1}^{\infty}\binom{-\frac{1}{2}}{k}\frac{(-1)^k}{k}=2\ln 2\approx 1.386294>0.
\end{equation*}
\end{cor}

Let us continue with the discussion of the properties of the $g$-function \eqref{gfunc}. Set
\begin{equation*}
	\Omega(z) = g_+(z)+g_-(z),\hspace{0.35cm}z\in \mathbb{R};\hspace{0.5cm}\textnormal{where}\ \ \ \ \ g_{\pm}(z)=\lim_{\epsilon\downarrow 0}g(z\pm i\epsilon).
\end{equation*}
Then $\Omega(z)$ is real and constant on the connected components of $J$, indeed
\begin{eqnarray*}
	\Omega(z) &=& 0,\hspace{0.5cm}z\in(a,1)\\
	\Omega(z) &=&-2\int_{-a}^a\sqrt{\frac{a^2-\mu^2}{1-\mu^2}}\,\d\mu,\hspace{0.5cm}z\in(-1,-a).
\end{eqnarray*}
Moreover, we have the following Proposition.
\begin{prop}\label{prop2} Let
\begin{equation*}
	\Pi(z) = i\left(g_+(z)-g_-(z)\right),\hspace{0.5cm}z\in\mathbb{R}.
\end{equation*}
Then $\Pi(z)$ is real, with
\begin{equation*}
	\Pi(z)<0,\hspace{0.25cm}z\in J,\hspace{1cm}\Pi(z)=0,\hspace{0.25cm}z\in\mathbb{R}\backslash[-1,1]\hspace{1cm}\textnormal{and}\hspace{1cm}\Pi(z) = -2\kappa,\hspace{0.25cm}z\in(-a,a).
\end{equation*}
\end{prop}
\begin{proof} First $g_+(z)-g_-(z)=0$ for $z\in(-\infty,-1)\cup(1,\infty)$. Second,
\begin{eqnarray*}
	\Pi(z)&=&-2\int_z^1\sqrt{\frac{\mu^2-a^2}{1-\mu^2}}\,\d\mu<0,\hspace{0.5cm}z\in(a,1)\\
	\Pi(z)&=&-2\int_a^1\sqrt{\frac{\mu^2-a^2}{1-\mu^2}}\,\d\mu+2\int_z^{-a}\sqrt{\frac{\mu^2-a^2}{1-\mu^2}}\,\d\mu=
	-2\int_{-z}^1\sqrt{\frac{\mu^2-a^2}{1-\mu^2}}\,\d\mu<0,\hspace{0.15cm}z\in(-1,-a)
\end{eqnarray*}
and finally for the gap $(-a,a)$ by \eqref{Dya}
\begin{equation*}
	\Pi(z) = -2\int_a^1\sqrt{\frac{\mu^2-a^2}{1-\mu^2}}\,\d\mu = -2\kappa,\ \ \ z\in(-a,a).
\end{equation*}
\end{proof}
We now employ the standard $g$-function transformation.
\subsection{First transformation of the RHP} Define
\begin{equation}\label{XRHP}
	X(\lambda) = e^{is\ell\sigma_3}Y(\lambda)e^{-is(g(\lambda)-\lambda)\sigma_3},\hspace{0.5cm}\lambda\in\mathbb{C}\backslash[-1,1]
\end{equation}
which transforms the initial $Y$-RHP of Section \ref{sec1} to the following $X$-RHP
\begin{itemize}
	\item $X(\lambda)$ is analytic for $\lambda\in\mathbb{C}\backslash[-1,1]$
	\item Along the line segment $[-1,1]$ oriented from left to right
	\begin{equation}\label{Xjump}
		X_+(\lambda) = X_-(\lambda)\begin{pmatrix}
		(1-\gamma)e^{-s\Pi(\lambda)} & \gamma e^{is\Omega(\lambda)}\\
		-\gamma e^{-is\Omega(\lambda)} & (1+\gamma)e^{s\Pi(\lambda)}\\
		\end{pmatrix},\hspace{0.5cm} \lambda\in[-1,1].
	\end{equation}
	\item $X(\lambda)$ has at most logarithmic singularities at $\lambda = \pm 1$, more precisely from \eqref{Ysing},
	\begin{equation}\label{Xsing}
		e^{-is\ell\sigma_3}X(\lambda)e^{isg(\lambda)\sigma_3} = \check{X}(\lambda)\left[I+\frac{\gamma}{2\pi i}\begin{pmatrix}
		-1 & 1\\
		-1 & 1\\
		\end{pmatrix}\ln\left(\frac{\lambda-1}{\lambda+1}\right)\right]
	\end{equation}
	where $\check{X}(\lambda)$ is analytic at $\lambda=\pm 1$ and the branch of the logarithm is fixed by the condition $-\pi<\textnormal{arg}\,\frac{\lambda-1}{\lambda+1}<\pi$.
	\item From \eqref{gfuncasy} as $\lambda\rightarrow\infty$, 
	\begin{equation*}
		X(\lambda)=I+\mathcal{O}\big(\lambda^{-1}\big).
	\end{equation*}
\end{itemize}
Let us take a closer look at the jump matrix $G_X(\lambda)$ in \eqref{Xjump}. For the right slit, we can write
\begin{equation}\label{jright}
	G_X(\lambda) = \begin{pmatrix}
	e^{-s(2\kappa+\Pi(\lambda))}& \gamma \\
	-\gamma  & (1+\gamma)e^{s\Pi(\lambda)}\\
	\end{pmatrix},\hspace{0.5cm}\lambda\in(a,1)
\end{equation}
with $\kappa$ as in \eqref{Dya}. Since for $\lambda\in(a,1)$
\begin{equation}\label{e01}
	2\kappa+\Pi(\lambda) = 2\left(\int_a^1\sqrt{\frac{\mu^2-a^2}{1-\mu^2}}\,\d\mu-\int_{\lambda}^1\sqrt{\frac{\mu^2-a^2}{1-\mu^2}}\,\d\mu\right) = 2\int_a^{\lambda}\sqrt{\frac{\mu^2-a^2}{1-\mu^2}}\,\d\mu>0,
\end{equation}
and by Proposition \ref{prop2}, $\Pi(\lambda)<0$ in the right slit, we have
\begin{equation}\label{jrightl}
	G_X(\lambda)\begin{pmatrix}
	0 & 1\\
	-1&0\\
	\end{pmatrix}^{-1}\rightarrow I,\hspace{1cm}\textnormal{as}\hspace{0.25cm} s\rightarrow\infty,\gamma\uparrow 1:\,\, \kappa\in[\delta,1-\delta]
\end{equation}
uniformly on any compact subset of the right slit $(a,1)$. For the left slit we argue in a similar way:
\begin{equation}\label{jleft}
	G_X(\lambda) = \begin{pmatrix}
	e^{-s(2\kappa+\Pi(\lambda))}& \gamma e^{is\Omega(\lambda)}\\
	-\gamma e^{-is\Omega(\lambda)} & (1+\gamma)e^{s\Pi(\lambda)}\\
	\end{pmatrix},\hspace{0.5cm}\lambda\in(-1,-a),
\end{equation}
and this time (compare Proposition \ref{prop2})
\begin{equation}\label{e02}
	2\kappa+\Pi(\lambda) = 2\int_{\lambda}^{-a}\sqrt{\frac{\mu^2-a^2}{1-\mu^2}}\,\d\mu>0,\hspace{0.5cm}\lambda\in(-1,-a).
\end{equation}
Hence
\begin{equation}\label{jleftl}
	G_X(\lambda)\begin{pmatrix}
	0 & e^{is\Omega(\lambda)}\\
	-e^{-is\Omega(\lambda)} & 0\\
	\end{pmatrix}^{-1}\rightarrow I,\hspace{1cm}\textnormal{as}\hspace{0.25cm} s\rightarrow\infty,\gamma\uparrow 1:\,\, \kappa\in[\delta,1-\delta],
\end{equation}
uniformly on any compact subset of the left slit $(-1,-a)$. In the gap $(-a,a)$, we recall Proposition \ref{prop2} once more:
\begin{equation}\label{mjump}
	G_X(\lambda) = \begin{pmatrix}
	1 & \gamma e^{is\Omega(\lambda)}\\
	-\gamma e^{-is\Omega(\lambda)} & 1-\gamma^2\\
	\end{pmatrix} = \begin{pmatrix}
	1 & 0\\
	-\gamma e^{-is\Omega(\lambda)} & 1\\
	\end{pmatrix}\begin{pmatrix}
	1 & \gamma e^{is\Omega(\lambda)}\\
	0 & 1\\
	\end{pmatrix} = S_L(\lambda)S_U(\lambda).
\end{equation}
\begin{prop}\label{prop3} Define for $z\in(-a,a)$
\begin{equation*}
	H_1(z) = i\Omega(z),\hspace{0.5cm} H_2(z)=-i\Omega(z).
\end{equation*}
Then $H_1(z)$ admits analytic continuation into a neighborhood of the line-segment $(-a,a)$ in the upper half-plane such that
\begin{equation*}
	\Re\,\big(H_1(z)\big)<0,\hspace{0.75cm}\Im\, z>0,\ \Re\, z\in(-a,a).
\end{equation*}
Similarly, $H_2(z)$ admits local analytic continuation into the lower half-plane such that
\begin{equation*}
	\Re\,\big(H_2(z)\big)<0,\hspace{0.75cm}\Im\, z<0,\ \Re\, z\in(-a,a).
\end{equation*}
\end{prop}
\begin{proof} We notice that
\begin{equation*}
	i\Omega(z) = -2i\int_z^a\sqrt{\frac{a^2-\mu^2}{1-\mu^2}}\,\d\mu,\hspace{0.5cm}z\in(-a,a)
\end{equation*}
and therefore
\begin{eqnarray*}
	\frac{\d}{\d y}H_1(x+iy)\Big|_{y=0}&=&-2\sqrt{\frac{a^2-x^2}{1-x^2}}<0,\hspace{0.5cm}x\in(-a,a)\\
	\frac{\d}{\d y}H_2(x-iy)\Big|_{y=0}&=&-2\sqrt{\frac{a^2-x^2}{1-x^2}}<0,\hspace{0.5cm}x\in(-a,a).
\end{eqnarray*}
Hence we can locally continue $H_1(z)$ into the upper half-plane with $\Re\left(H_1(z)\right)<0$ and $H_2(z)$ into the lower half-plane with $\Re\left(H_2(z)\right)<0$.
\end{proof}
The latter Proposition enables us to open lens.
\subsection{Second transformation of the RHP - opening of lens}\label{sec3}
Let $\mathcal{L}^{\pm}$ denote the {\it upper (lower) lens}, shown in Figure \ref{figure2}, which is bounded by the contour $\gamma^{\pm}$. Define
\begin{equation}\label{SRHP}
	S(\lambda) = \left\{
                                   \begin{array}{ll}
                                     X(\lambda)S_U^{-1}(\lambda), & \hbox{$\lambda\in\mathcal{L}^+$,} \\
                                     X(\lambda)S_L(\lambda), & \hbox{$\lambda\in\mathcal{L}^-$,} \\ 
                                     X(\lambda), & \hbox{otherwise},
                                   \end{array}
                                 \right.
\end{equation}
then $S(\lambda)$ solves the following RHP
\begin{itemize}
	\item $S(\lambda)$ is analytic for $\lambda\in\mathbb{C}\backslash\big([-1,-a]\cup\gamma^+\cup\gamma^-\cup[a,1]\big)$
	\item The following jumps hold, with orientation fixed as in Figure \ref{figure2}
	\begin{equation}\label{Sjump}
		S_+(\lambda) = S_-(\lambda)\left\{
                                   \begin{array}{ll}
                                     G_X(\lambda), & \hbox{$\lambda\in[-1,-a]\cup[a,1]$,} \\
                                     S_U(\lambda), & \hbox{$\lambda\in\gamma^+$,} \\ 
                                     S_L(\lambda), & \hbox{$\lambda\in\gamma^-$,}
                                   \end{array}
                                 \right.
  \end{equation}
  where $G_X(\lambda)$ is given by \eqref{jright},\eqref{jleft}, and $S_U(\lambda)$ and $S_L(\lambda)$ denote the analytic continuations of the corresponding matrices in \eqref{mjump} to the upper and lower half-plane.
  \item $S(\lambda)$ has at most logarithmic endpoint singularities at the branch points $\lambda=\pm 1$, see \eqref{Xsing}.
  \item As $\lambda\rightarrow\infty$
  \begin{equation*}
  	S(\lambda)=I+\mathcal{O}\big(\lambda^{-1}\big).
  \end{equation*} 
\end{itemize}
\begin{figure}[tbh]
\begin{center}
\includegraphics[width=0.55\textwidth]{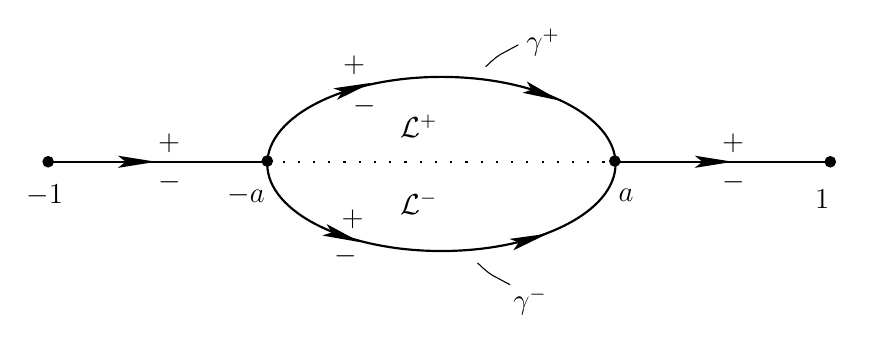}
\end{center}
\caption{Opening of lenses}
\label{figure2}
\end{figure}

Recalling Proposition \ref{prop3}, we see that in the considered double scaling limit \eqref{dbscale}
\begin{equation}\label{Sesti}
	S_U(\lambda)\rightarrow I,\hspace{0.5cm}\lambda\in\gamma^+;\hspace{1cm}S_L(\lambda)\rightarrow I,\hspace{0.5cm}\lambda\in\gamma^-
\end{equation}
away from the branch points $\lambda=\pm a$ and the stated convergence is in fact exponentially fast. Hence the jumps on the lens boundaries $\gamma^+\cup\gamma^-$ decay exponentially fast and recalling \eqref{jrightl} and \eqref{jleftl}, we therefore expect that the major contribution to the asymptotical solution of the $X$-RHP will arise from the slits $(-1,-a)\cup(a,1)$ as well as the neighborhoods of $\lambda=\pm 1,\pm a$. This will be proven rigorously with the help of explicit model functions/parametrices in the relevant neighborhoods. 

\section{Riemann-Hilbert analysis for $\kappa\in[\delta,1-\delta]$ - model problems}\label{sec2}
\subsection{The outer parametrix}
We first consider the following problem. Find a piecewise analytic $2\times 2$ matrix valued function $M(\lambda)$ such that
\begin{itemize}
	\item $M(\lambda)$ is analytic for $\lambda\in\mathbb{C}\backslash\overline{J}$
	\item Along the left slit $(-1,-a)$ as oriented in Figure \ref{figure2}, 
	\begin{equation*}
		M_+(\lambda) = M_-(\lambda)\begin{pmatrix}
		0 & e^{is\Omega(\lambda)}\\
	-e^{-is\Omega(\lambda)} & 0\\
	\end{pmatrix};
	\end{equation*}
	along the right slit $(a,1)$,
	\begin{equation*}
		M_+(\lambda) = M_-(\lambda)\begin{pmatrix}
		0 & 1\\
		-1 & 0\\
		\end{pmatrix}
	\end{equation*}
	\item $M(\lambda)$ has $L^2$-boundary values on $J$
	\item As $\lambda\rightarrow\infty$,
	\begin{equation*}
		M(\lambda)\rightarrow I.
	\end{equation*}
\end{itemize}
Such model problems appeared, for instance, in \cite{DIZ} or \cite{DKMVZ} and can be explicitly solved in terms of theta-functions on the Riemann surface $\Gamma$. Recall that the space of holomorphic one-forms on $\Gamma$ is one-dimensional, and denote 
\begin{equation}\label{Anorm}
	\omega = \frac{c\,\d\lambda}{\sqrt{p(\lambda)}},\hspace{0.5cm} c^{-1} = -2i\int_a^1\frac{\d\lambda}{\sqrt{(1-\lambda^2)(\lambda^2-a^2)}},
\end{equation}
the unique holomorphic differential with $A$-period normalization
\begin{equation*}
	\oint_{A_1}\omega = 1
\end{equation*}
and $B$-period
\begin{equation}\label{Bper}
	\tau = \oint_{B_1}\omega = 2c\int_{-a}^a\frac{\d\lambda}{\sqrt{(a^2-\lambda^2)(1-\lambda^2)}}.
\end{equation}
As we see from \eqref{Anorm}, $\tau$ is pure imaginary and $-i\tau$ is positive for $\kappa\in[\delta,1-\delta]$, thus the associated theta function
\begin{equation}\label{thetaf}
	\theta(z)\equiv \theta_3(z|\tau) = \sum_{k\in\mathbb{Z}}\exp\left\{\pi ik^2\tau+2\pi ikz\right\}=1+2\sum_{k=1}^{\infty}e^{\pi ik^2\tau}\cos\left(2\pi kz\right),\hspace{0.5cm}z\in\mathbb{C}
\end{equation}
is a well-defined entire function, which is even and quasi-periodic
\begin{equation}\label{thetaprop}
	\theta(z+n+m\tau) = \theta(z)e^{-\pi i\tau m^2-2\pi izm},\hspace{0.5cm} n,m\in\mathbb{Z},\ z\in\mathbb{C}.
\end{equation}
Moreover, note for future purposes that, the following ratio of theta functions
\begin{equation*}
	f(z) = \frac{\theta(z+c_1+c_2)}{\theta(z+c_2)},\hspace{0.5cm}c_i,z\in\mathbb{C}
\end{equation*}
formally satisfies the quasi-periodicity relation
\begin{equation}\label{ratioj}
	f(z+n+m\tau) = f(z)e^{-2\pi imc_1}.
\end{equation}
Next we set
\begin{equation}\label{udef}
	u(z) = \int_1^z\omega,\hspace{0.5cm}z\in\mathbb{C}\backslash[-1,1],
\end{equation}
where the integral is taken along any path from $1$ to $z$ in the upper (lower) half plane if  $\Im\,z>0\,(<0)$, see for instance Figure \ref{figure3} below.
\begin{figure}[tbh]
\begin{center}
\includegraphics[width=0.55\textwidth]{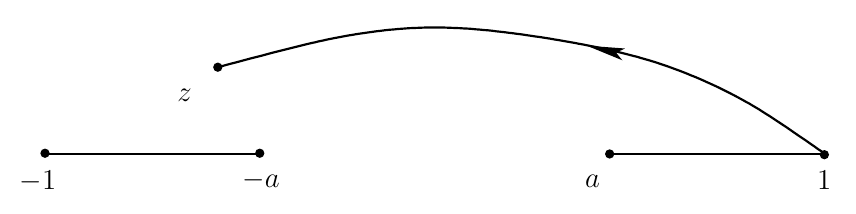}
\end{center}
\caption{Definition of the Abelian integral $u(z)$}
\label{figure3}
\end{figure}

Since we also need to consider integration on the Riemann surface $\Gamma$, we extend the definition
\begin{equation*}
	u(P)=\int_1^P\omega,\hspace{0.5cm}P\in\Gamma
\end{equation*}
and let $\Lambda = \mathbb{Z}+\tau\mathbb{Z}$ be the period lattice. The analytical properties of $u(z)$ with respect to $\Lambda$ are summarized in the following Proposition.
\begin{prop}\label{prop4} The abelian integral $u(z)$ is analytic on the first sheet, with a cut along $(-1,1)$. We have
\begin{equation*}
	u_+(z)-u_-(z) = \left\{
                                   \begin{array}{ll}
                                     0, & \hbox{$z\in(-\infty,-1)\cup(1,\infty)$,} \\
                                     -1, & \hbox{$z\in(-a,a)$,} 
                                   \end{array}
                                 \right.
\end{equation*}
and
\begin{equation*}
	u_+(z) +u_-(z) = \left\{
                                   \begin{array}{ll}
                                     0, & \hbox{$z\in(a,1)$,} \\
                                     \tau, & \hbox{$z\in(-1,-a)$,} 
                                   \end{array}
                                 \right.
\end{equation*}
and for $z\rightarrow\infty$,
\begin{equation*}
	u(z) = u_{\infty}-\frac{c}{z}+\mathcal{O}\big(z^{-2}\big),\hspace{0.5cm} u_{\infty} = \int_1^{P_1(\infty)}\omega,
\end{equation*}
where $P_1(\infty)$ denotes the pre-image of $\infty$ on the first sheet of $\Gamma$.
\end{prop}
\begin{proof} The stated jump properties can be verified easily, we only note that
\begin{equation*}
	u_+(z)-u_-(z) = 2ic\int_a^1\frac{\d\lambda}{\sqrt{(\lambda^2-a^2)(1-\lambda^2)}} = -\int_{A_1}\omega,\hspace{0.5cm}z\in(-a,a)
\end{equation*}
and
\begin{equation*}
	u_+(z)+u_-(z) = 2c\int_{-a}^a\frac{\d\lambda}{\sqrt{(a^2-\lambda^2)(1-\lambda^2)}} = \int_{B_1}\omega,\hspace{0.5cm}z\in(-1,-a).
\end{equation*}
\end{proof}
Let us assume temporarily that $\theta\left(u(\lambda)\pm d\right)\not\equiv 0$, where $d\in\mathbb{C}$ is a constant which will be fixed later on. Introduce another constant
\begin{equation}\label{Vconst}
	V  \equiv -\frac{1}{\pi}\int_{-a}^a\sqrt{\frac{a^2-\mu^2}{1-\mu^2}}\,\d\mu= \frac{1}{2\pi}\Omega(\lambda),\ \ \ \textnormal{if}\ \ \lambda\in(-1,-a)
\end{equation}
and set
\begin{equation}\label{Ndef}
	N(\lambda,\pm d) = \left(\frac{\theta\left(u(\lambda)+sV\pm d\right)}{\theta\left(u(\lambda)\pm d\right)},\frac{\theta\left(-u(\lambda)+sV\pm d\right)}{\theta\left(-u(\lambda)\pm d\right)}\right) = \left(N_1(\lambda,\pm d),N_2(\lambda,\pm d)\right),
\end{equation}
where $u(\lambda)$ is given by \eqref{udef} with the integral on the first sheet of $\Gamma$.
\begin{prop}\label{prop5} The functions $N(z,d)$ and $N(z,-d)$ are single-valued and meromorphic in $\mathbb{C}\backslash J$ with
\begin{eqnarray*}
	N_+(z,\pm d)&=&N_-(z,\pm d)\begin{pmatrix}
	0 & 1\\
	1 & 0\\
	\end{pmatrix},\hspace{0.5cm}z\in(a,1)\\
	N_+(z,\pm d) &=& N_-(z,\pm d)\begin{pmatrix}
	0 & e^{2\pi isV}\\
	e^{-2\pi isV} & 0\\
	\end{pmatrix},\hspace{0.5cm}z\in(-1,-a).
\end{eqnarray*}
\end{prop}
\begin{proof} All values of the multivalued function $u(z)$ differ by an integer since we choose integration on the first sheet of $\Gamma$. But $\theta(z)$ is periodic with respect to a shift by an integer, hence $N(z,\pm d)$ are single-valued and since $\theta(u(z)-d) = \theta(-u(z)+d)$ is not identically zero, the functions $N(z,\pm d)$ are in fact meromorphic on $\mathbb{C}\backslash J$. The remaining jump properties can be easily derived from Proposition \ref{prop4} and equation \eqref{ratioj}.
\end{proof}
In order to correct the jumps along the slits, we consider
\begin{equation}\label{bdef}
	\beta(z) = \left(\frac{(z+a)(z-1)}{(z+1)(z-a)}\right)^{\frac{1}{4}},\hspace{0.5cm}z\in\mathbb{C}\backslash\overline{J}
\end{equation}
with the branch fixed by the condition $\beta(z)>0$ for $z>1$.
\begin{prop}\label{prop6} The function $\beta(z)-\left(\beta(z)\right)^{-1}$ has only one zero $z_0=0$. This root lies in the gap of $J$. The function $\beta(z)+\left(\beta(z)\right)^{-1}$ has no zeros in $\mathbb{C}\backslash \overline{J}$.
\end{prop}
\begin{proof} Since
\begin{equation*}
	\beta(z)\pm\left(\beta(z)\right)^{-1}=0\ \ \Leftrightarrow\ \ \left(\beta(z)\right)^2\pm 1 = 0\ \Rightarrow\ \ \beta^4(z) = 1\ \ \Leftrightarrow\ \  2z(a-1)=0
\end{equation*}
we have only one zero $z_0\in(-a,a)$. On the other hand $\left(\beta(z_0)\right)^2=\pm 1$, and $\left(\beta(z)\right)^2>0$ for $z\in(-a,a)$, which proves the Proposition. 
\end{proof}
As our next step, we define for $z\in\mathbb{C}\backslash\overline{J}$
\begin{equation}\label{hdef}
	\phi(z) = \frac{1}{2}\left(\beta(z)+\left(\beta(z)\right)^{-1}\right),\hspace{0.5cm}\hat{\phi}(z) = \frac{1}{2i}\left(\beta(z)-\left(\beta(z)\right)^{-1}\right),
\end{equation}
which satisfy
\begin{equation*}
	\phi_+(z) = -\hat{\phi}_-(z),\hspace{0.5cm}\hat{\phi}_+(z)=\phi_-(z),\hspace{1cm}z\in J,
\end{equation*}
and as $z\rightarrow\infty$
\begin{equation*}
	\phi(z) = 1+\frac{(1-a)^2}{8z^2}+\mathcal{O}\big(z^{-3}\big),\hspace{1cm}\hat{\phi}(z) = -\frac{1-a}{2iz}\left(1+\mathcal{O}\left(z^{-2}\right)\right).
\end{equation*}
By Proposition \ref{prop6}, the function $\hat{\phi}(z)$ has only one zero $z_0=0$ in $\mathbb{C}\backslash\overline{J}$. We will now choose the constant $d$ in \eqref{Ndef} so that $\theta\left(u(z)-d\right)$ also vanishes at $z_0$. First, we note the following.
\begin{prop}\label{prop7} We have
\begin{equation}\label{Ab1}
	u_{\infty}-\frac{\tau}{4}=0
\end{equation}
\end{prop}
\begin{proof} Observe that
\begin{equation*}
	u_{\infty}=\int_1^{P_1(\infty)}\omega = c\int_1^{\infty}\frac{\d\lambda}{\sqrt{(\lambda^2-1)(\lambda^2-a^2)}}=c\int_{-a}^0\frac{\d w}{\sqrt{(a^2-w^2)(1-w^2)}}=\frac{\tau}{4}
\end{equation*}
and thus the stated identity follows.
\end{proof}
We now choose
\begin{equation}\label{dchoice}
	d = -\frac{\tau}{4}
\end{equation}
so that
\begin{equation}\label{dconse}
	u_{\infty}+d=0
\end{equation}
and
\begin{equation*}
	u(z_0)-d =\int_1^{P_1(0)}\omega-d\equiv \frac{1}{2}\left(1+\tau\right)\mod\Lambda.
\end{equation*}
\begin{prop}\label{prop8} The multi-valued functions $\theta\left(u(P)\pm d\right)$ with $d$ as in \eqref{dchoice} are not identically zero on $\Gamma$. We have
\begin{equation*}
	\theta\left(u(P)+d\right) = 0\hspace{0.15cm}\Leftrightarrow\hspace{0.15cm}P=P_2(z_0)\hspace{0.5cm}\textnormal{and}\hspace{0.5cm}\theta\left(u(P)-d\right)=0\hspace{0.15cm}\Leftrightarrow\hspace{0.15cm}P=P_1(z_0)
\end{equation*}
and
\begin{equation}\label{uzero}
	\theta(sV)\neq 0\hspace{0.5cm}\textnormal{for all}\ s\in\mathbb{R}.
\end{equation}
\end{prop}
\begin{proof} We first rewrite \eqref{dchoice} as follows
\begin{equation*}
	d = \int_{-a}^{P_1(0)}\omega = -K-\int_{-1}^{P_2(0)}\omega-\frac{\tau}{2}
\end{equation*}
where we introduced the Riemann constant (see \cite{FK})
\begin{equation}\label{RK}
	K=\int_{-1}^a\omega=-\frac{1}{2}-\frac{\tau}{2}.
\end{equation}
Hence
\begin{equation*}
	u(P)+d = \int_1^P\omega+d = \int_{-1}^P\omega +\frac{\tau}{2}+d \equiv \int_{-1}^P\omega -K-\int_{-1}^{P_2(0)}\omega \mod \Lambda.
\end{equation*}
Since the divisor $P_2(0)$ is non-special, it follows by the Riemann-Roch theorem that (e.g. \cite{FK})
\begin{equation*}
	\theta\left(\int_{-1}^P\omega-K-\int_{-1}^{P_2(0)}\omega\right)
\end{equation*}
is not identically zero on $\Gamma$ and has roots precisely at $P=P_2(0)$. On the other hand
\begin{equation*}
	-d = K+\int_{-1}^{P_2(0)}\omega + \frac{\tau}{2} \equiv -K-\int_{-1}^{P_1(0)}\omega+\frac{\tau}{2}\mod\Lambda,
\end{equation*}
which follows from \eqref{RK}. The divisor $P_1(0)$ is also non-special, hence we have by similar reasoning as above that $\theta\left(u(P)-d\right)$ does not vanish identically and has zeros precisely at $P=P_1(0)$. The remaining statement \eqref{uzero} follows from Proposition \ref{prop7} and \eqref{dchoice} as well as the fact that the zeros of the theta function $\theta(z)$ are given by the points $z\equiv\frac{1}{2}+\frac{\tau}{2}\mod\Lambda$.
\end{proof}
Combining all the above, we define
\begin{equation}\label{Qdef}
	Q(\lambda) = \begin{pmatrix}
	N_1(\lambda,d)\phi(\lambda) & N_2(\lambda,d)\hat{\phi}(\lambda)\\
	-N_1(\lambda,-d)\hat{\phi}(\lambda) & N_2(\lambda,-d)\phi(\lambda)\\
	\end{pmatrix},\hspace{0.5cm}\lambda\in\mathbb{C}\backslash\overline{J}.
\end{equation}
Then $Q(\lambda)$ is analytic in $\mathbb{C}\backslash\overline{J}$ with jumps
\begin{eqnarray*}
	Q_+(\lambda)&=&Q_-(\lambda)\begin{pmatrix}
	0 & 1\\
	-1 & 0\\
	\end{pmatrix},\hspace{0.5cm}\lambda\in(a,1)\\
	Q_+(\lambda)&=&Q_-(\lambda)\begin{pmatrix}
	0 & e^{is\Omega(\lambda)}\\
	-e^{-is\Omega(\lambda)} & 0\\
	\end{pmatrix},\hspace{0.5cm}\lambda\in(-1,-a)
\end{eqnarray*}
and normalization (recall \eqref{dconse} and the fact that $\theta(z)$ is even)
\begin{equation*}
	Q(\lambda)= Q(\infty)\Bigg[I+\frac{1}{\lambda}\begin{pmatrix}
	-c\,\frac{\theta'(sV)}{\theta(sV)} & -\frac{\theta(0)}{\theta(sV)}\frac{\theta(u_{\infty}-sV-d)}{\theta(u_{\infty}-d)}\frac{1-a}{2i}\smallskip\\
	\frac{\theta(0)}{\theta(sV)}\frac{\theta(u_{\infty}+sV-d)}{\theta(u_{\infty}-d)}\frac{1-a}{2i} & c\,\frac{\theta'(sV)}{\theta(sV)}\\
	\end{pmatrix}+\mathcal{O}\big(\lambda^{-2}\big)\Bigg]
\end{equation*}
as $\lambda\rightarrow\infty$, where $\theta'(x)=\frac{\d}{\d x}\theta(x)$ and with
\begin{equation*}
	Q(\infty) = \frac{\theta(sV)}{\theta(0)}I.
\end{equation*}
Now set
\begin{equation}\label{Mdef}
	M(\lambda) = \left(Q(\infty)\right)^{-1}Q(\lambda),\hspace{0.5cm}\lambda\in\mathbb{C}\backslash\overline{J}
\end{equation}
and we are immediately lead to the following Proposition, which gives the solution of the model problem.
\begin{prop}\label{prop9} The matrix valued function $M(\lambda)$ given by \eqref{Mdef} is analytic in $\mathbb{C}\backslash\overline{J}$ with jumps
\begin{eqnarray*}
	M_+(\lambda) &=& M_-(\lambda)\begin{pmatrix}
	0 & 1\\
	-1 & 0\\
	\end{pmatrix},\hspace{0.5cm}\lambda\in(a,1)\\
	M_+(\lambda) &=&M_-(\lambda)\begin{pmatrix}
	0 & e^{is\Omega(\lambda)}\\
	-e^{-is\Omega(\lambda)} & 0\\
	\end{pmatrix},\hspace{0.5cm}\lambda\in(-1,-a),
\end{eqnarray*}
has square integrable singularities at the branchpoints $\lambda=\pm a,\pm 1$, and 
\begin{equation}\label{Masy}
	M(\lambda) = I+\frac{M_1}{\lambda}+\mathcal{O}\big(\lambda^{-2}\big),\hspace{0.5cm} M_1=\big(M_1^{jk}\big),\ \ \ \ \textnormal{as}\ \ \lambda\rightarrow\infty.
\end{equation}
\end{prop}

\subsection{Approximate parametrix near the branch point $\lambda=a$}
As opposed to ``standard'' cases in which we can construct a local model function that precisely models the jump properties of the initial problem, we are facing here a novelty in the parametrix analysis. Let us fix a small neighborhood $\mathcal{U}$ of radius less than $r=\min\left\{\frac{a}{2},\frac{1}{2}(1-a)\right\}$ of the point $a$ as shown in Figure \ref{figure4}.
\begin{figure}[tbh]
\begin{center}
\includegraphics[width=0.63\textwidth]{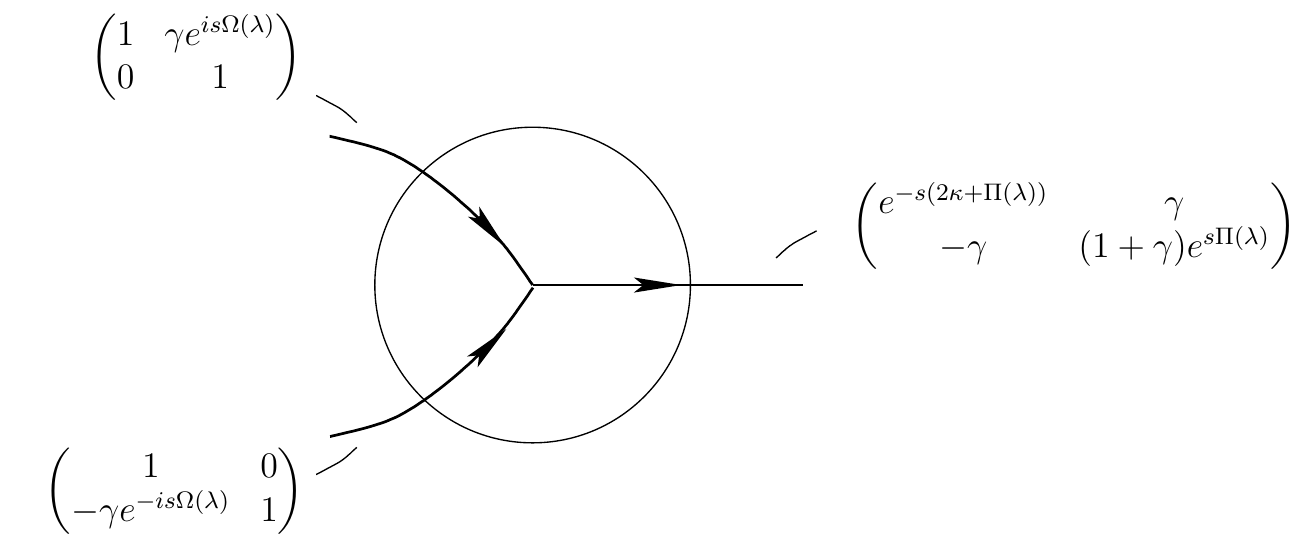}
\end{center}
\caption{A neighborhood $\mathcal{U}$ of $\lambda=a$ with jumps for $S(\lambda)$}
\label{figure4}
\end{figure}

Observe that by Proposition \ref{prop2},
\begin{equation}\label{app:1}
	\Pi(\lambda) = -2\int_{\lambda}^1\sqrt{\frac{\mu^2-a^2}{1-\mu^2}}\,\d\mu = -2\kappa+c_0\left(\lambda-a\right)^{\frac{3}{2}}+\mathcal{O}\left((\lambda-a)^{\frac{5}{2}}\right),\hspace{0.5cm} c_0=\frac{4}{3}\sqrt{\frac{2a}{1-a^2}}\,>0
\end{equation}
for $\lambda\in\mathcal{U}\cap(a,1)$, where the function $(z-a)^{\frac{3}{2}}$ is defined for $z\in\mathbb{C}\backslash[a,\infty)$ and the branch of the root is fixed by the condition
\begin{equation*}
	\textnormal{arg}\,\left(z-a\right)=\pi\ \ \textnormal{if}\ \ z<a.
\end{equation*}
Furthermore,
\begin{equation*}
	i\Omega(\lambda) = -2i\int_{\lambda}^a\sqrt{\frac{a^2-\mu^2}{1-\mu^2}}\,\d\mu = c_0\left(\lambda-a\right)^{\frac{3}{2}}+\mathcal{O}\left((\lambda-a)^{\frac{5}{2}}\right)
\end{equation*}
for $\lambda\in\mathcal{U}\cap\gamma^+$ and
\begin{equation*}
	-i\Omega(\lambda) = -c_0\left(\lambda-a\right)^{\frac{3}{2}}+\mathcal{O}\left((\lambda-a)^{\frac{5}{2}}\right),\hspace{0.25cm}\lambda\in\mathcal{U}\cap\gamma^-.
\end{equation*}
These behaviors suggest to use the following local variable
\begin{eqnarray}\label{change1}
	\zeta(\lambda) &=&\left(\frac{3s}{4}\right)^{\frac{2}{3}}\left(2i\int_a^{\lambda}\sqrt{\frac{\mu^2-a^2}{\mu^2-1}}\,\d\mu\right)^{\frac{2}{3}}\\
	&=&\left(s\sqrt{\frac{2a}{1-a^2}}\,\right)^{\frac{2}{3}}(\lambda-a)\left(1+\frac{1+3a^2}{10a(1-a^2)}(\lambda-a)+\mathcal{O}\big((\lambda-a)^2\big)\right),\hspace{0.5cm} |\lambda-a|<r.\nonumber
\end{eqnarray}
With this variable, the exact jump matrices depicted in Figure \ref{figure4} can be written as
\begin{eqnarray}
	\begin{pmatrix}
	1 &  \gamma e^{is\Omega(\lambda)}\\
	0& 1\\
	\end{pmatrix} &=& e^{\frac{2}{3}\zeta^{\frac{3}{2}}(\lambda)\sigma_3}\begin{pmatrix}
	1 & \gamma\\
	0& 1\\
	\end{pmatrix}e^{-\frac{2}{3}\zeta^{\frac{3}{2}}(\lambda)\sigma_3},\hspace{0.25cm}\lambda\in\mathcal{U}\cap \gamma^+\label{jp1}\\
	\begin{pmatrix}
	1 & 0\\
	-\gamma e^{-is\Omega(\lambda)} & 1\\
	\end{pmatrix} &=& e^{\frac{2}{3}\zeta^{\frac{3}{2}}(\lambda)\sigma_3}\begin{pmatrix}
	1 & 0\\
	-\gamma& 1\\
	\end{pmatrix}e^{-\frac{2}{3}\zeta^{\frac{3}{2}}(\lambda)\sigma_3},\hspace{0.25cm}\lambda\in\mathcal{U}\cap\gamma^-\label{jp2}\\
	\begin{pmatrix}
	e^{-s(2\kappa+\Pi(\lambda))} & \gamma\\
	-\gamma & (1+\gamma)e^{s\Pi(\lambda)}\\
	\end{pmatrix} &=&e^{\frac{2}{3}\zeta_-^{\frac{3}{2}}(\lambda)\sigma_3}\begin{pmatrix}
	1 & \gamma\\
	-\gamma& 1-\gamma^2\\
	\end{pmatrix}e^{-\frac{2}{3}\zeta_+^{\frac{3}{2}}(\lambda)\sigma_3}.\label{jp3}
\end{eqnarray}
with $\lambda\in\mathcal{U}\cap (a,1)$ in the latter equality. In the neigbhorhood $\mathcal{U}$ of the branch point $\lambda=a$, the model function $M(\lambda)$ in \eqref{Mdef} can be written as follows
\begin{equation}\label{MMhat}
	M(\lambda) = \widehat{M}(\lambda)\delta(\lambda)^{-\sigma_3}\frac{i}{2}\begin{pmatrix}
	1 & -i\\
	-1 & -i\\
	\end{pmatrix},\hspace{1cm} \delta(\lambda) = \left(\frac{\lambda-a}{\lambda-1}\right)^{\frac{1}{4}}\rightarrow 1,\hspace{0.5cm}\lambda\rightarrow\infty
\end{equation}
where $\widehat{M}(\lambda)$ is analytic in $\mathcal{U}$. Since $\delta(\lambda)/\zeta^{\frac{1}{4}}(\lambda)$ is locally analytic we are thus looking for a solution to the following model problem
\begin{itemize}
	\item $A^{RH}(\zeta)$ is analytic for $\zeta\in\mathbb{C}\backslash\Gamma_{\zeta}$ where the contour $\Gamma_{\zeta}$ is depicted in Figure \ref{figure5}
	\begin{figure}[tbh]
\begin{center}
\includegraphics[width=0.32\textwidth]{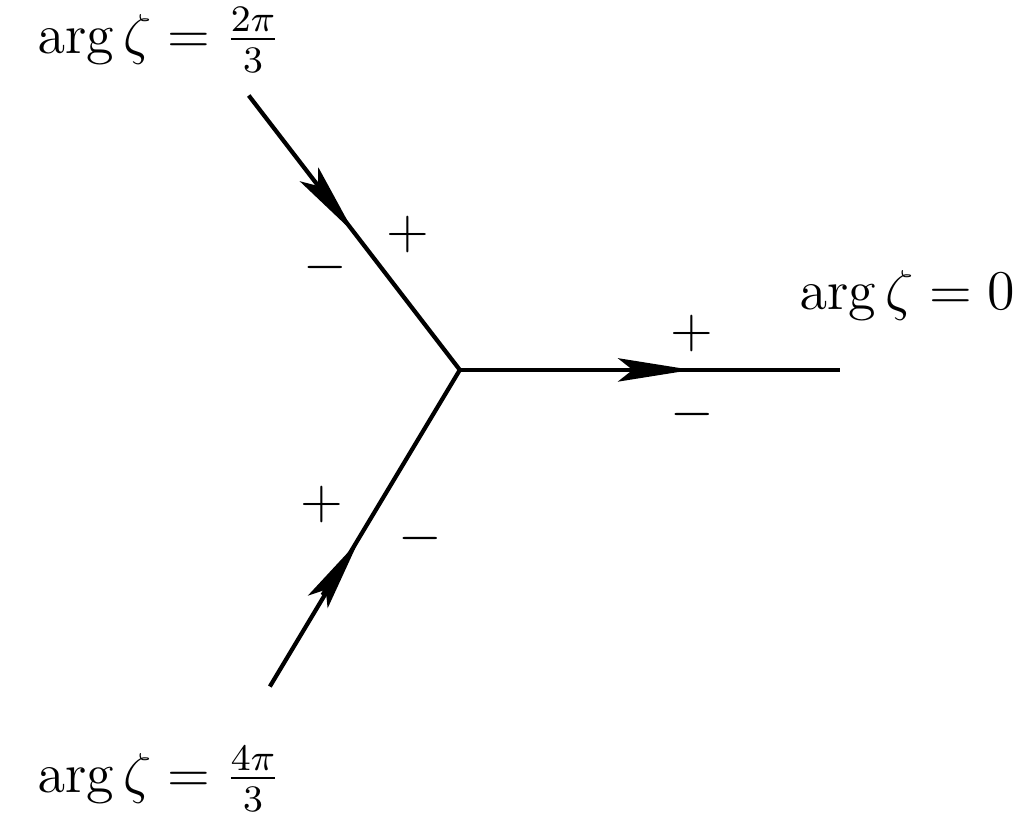}
\end{center}
\caption{Jump contour $\Gamma_{\zeta}$ for the model function $A^{RH}(\zeta)$}
\label{figure5}
\end{figure}
	\item $A_+^{RH}(\zeta) = A^{RH}_-(\zeta)L(\gamma)$ as $\zeta\in\Gamma_{\zeta}$ with a piecewise constant jump matrix
	\begin{equation}\label{ejaj}
		L = \begin{pmatrix}
		1 & \gamma\\
		0 & 1\\
		\end{pmatrix},\ \textnormal{arg}\,\zeta=\frac{2\pi}{3},\hspace{0.5cm} L = \begin{pmatrix}
		1 & 0\\
		-\gamma & 1\\
		\end{pmatrix},\ \textnormal{arg}\,\zeta=\frac{4\pi}{3},\hspace{0.5cm}L = \begin{pmatrix}
		1 & \gamma\\
		-\gamma & 1-\gamma^2\\
		\end{pmatrix},\ \textnormal{arg}\ \zeta=0.
	\end{equation}
	\item As $\zeta\rightarrow\infty$
	\begin{equation*}
		A^{RH}(\zeta) = \zeta^{-\frac{1}{4}\sigma_3}\frac{i}{2}\begin{pmatrix}
		1 & -i\\
		-1 & -i\\
		\end{pmatrix}\Big(I+\mathcal{O}\left(\zeta^{-1}\right)\Big)e^{\frac{2}{3}\zeta^{\frac{3}{2}}\sigma_3}.
	\end{equation*}
\end{itemize}
In this model problem, the jump matrix $L(\gamma)$ is $\zeta$-independent, hence the function
\begin{equation*}
	A(\zeta) = \frac{\d A^{RH}}{\d\zeta}(\zeta)\left(A^{RH}(\zeta)\right)^{-1}
\end{equation*}
is entire, and from its behavior at infinity we have
\begin{equation*}
	A(\zeta) = \begin{pmatrix}
	0 & -1\\
	-\zeta+d(\gamma) & 0\\
	\end{pmatrix} +\mathcal{O}\left(\zeta^{-\frac{1}{2}}\right),\hspace{0.5cm}\zeta\rightarrow\infty.
\end{equation*}
Here $d=d(\gamma)$ is independent of $\zeta$, and therefore by analyticity
\begin{equation}\label{sys1}
	\frac{\d A^{RH}}{\d\zeta}(\zeta)= \begin{pmatrix}
	0 & -1\\
	-\zeta+d(\gamma) & 0\\
	\end{pmatrix}A^{RH}(\zeta).
\end{equation}
This equation yields a matrix form for the Airy equation, indeed let
\begin{equation*}
	y(\zeta) = \left(A^{RH}(\zeta)\right)_{1j},\ \ y'(\zeta)=-\left(A^{RH}(\zeta)\right)_{2j},\ \ j=1,2,
\end{equation*}
then \eqref{sys1} gives us
\begin{equation*}
	y'' = (\zeta-d)y.
\end{equation*}
So we can expect that an ``exact'' parametrix near $\lambda=a$ is constructed out of Airy-functions. However, the Stokes multipliers of Airy functions are constant, i.e. in particular, $\gamma$-independent. On the other hand, the jumps in \eqref{ejaj} are $\gamma$-dependent. Instead of an ``exact'' parametrix, we will construct an ``approximate'' one which models the behavior with $\gamma=1$ in \eqref{ejaj}. More precisely, we choose the jump behavior near $\lambda=a$ as shown in Figure \ref{figure41}.
\begin{figure}[tbh]
\begin{center}
\includegraphics[width=0.45\textwidth]{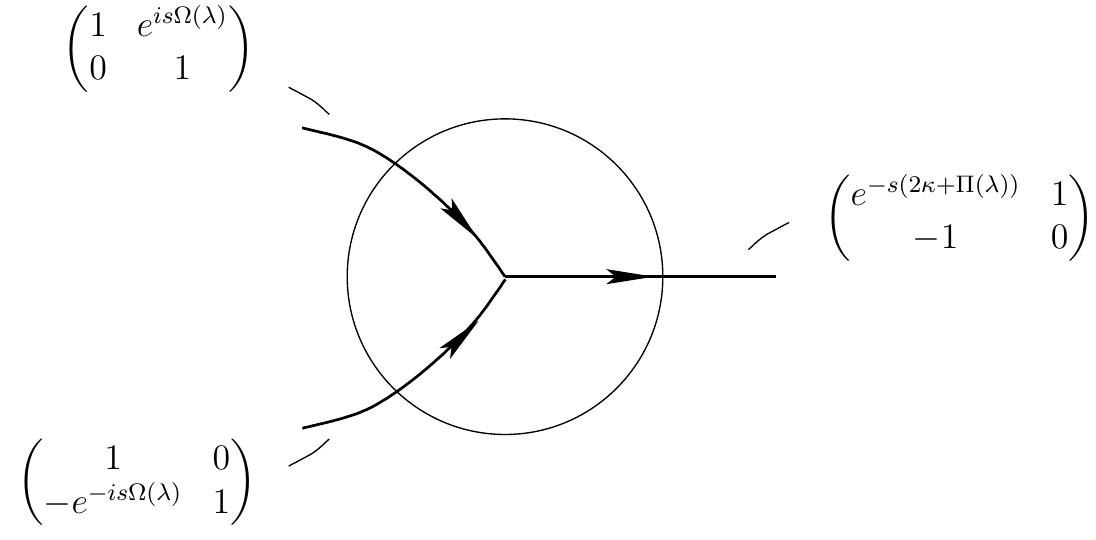}
\end{center}
\caption{A neighborhood $\mathcal{U}$ of $\lambda=a$ with approximate jumps}
\label{figure41}
\end{figure}

We use $\textnormal{Ai}(\zeta)$ for our construction, the solution to Airy's equation
\begin{equation*}
	w''=zw
\end{equation*}
uniquely determined by its asymptotics as $\zeta\rightarrow\infty$ and $-\pi<\textnormal{arg}\ \zeta<\pi$
\begin{equation*}
	\textnormal{Ai}(\zeta) = \frac{\zeta^{-\frac{1}{4}}}{2\sqrt{\pi}}e^{-\frac{2}{3}\zeta^{\frac{3}{2}}}\left(1-\frac{5}{48}\zeta^{-\frac{3}{2}}+\frac{385}{4608}\zeta^{-\frac{6}{2}}+\mathcal{O}\left(\zeta^{-\frac{9}{2}}\right)\right),
\end{equation*}
with the roots of $\zeta$ chosen with zero argument for $\zeta>0$. For $\pi<\textnormal{arg}\,\zeta<\frac{5\pi}{3}$ (see e.g. \cite{BE,NIST})
\begin{eqnarray}\label{p0}
	\textnormal{Ai}(\zeta) &=& \frac{\zeta^{-\frac{1}{4}}}{2\sqrt{\pi}}e^{-\frac{2}{3}\zeta^{\frac{3}{2}}}\left(1-\frac{5}{48}\zeta^{-\frac{3}{2}}+\frac{385}{4608}\zeta^{-\frac{6}{2}}+\mathcal{O}\left(\zeta^{-\frac{9}{2}}\right)\right)\\
	&&+\frac{i\zeta^{-\frac{1}{4}}}{2\sqrt{\pi}}e^{\frac{2}{3}\zeta^{\frac{3}{2}}}\left(1+\frac{5}{48}\zeta^{-\frac{3}{2}}+\frac{385}{4608}\zeta^{-\frac{6}{2}}+\mathcal{O}\left(\zeta^{-\frac{9}{2}}\right)\right),\hspace{0.5cm}\zeta\rightarrow\infty\nonumber,
\end{eqnarray}
and the following monodromy relation holds on the entire universal covering of the complex plane
\begin{equation}\label{p00}
	\textnormal{Ai}(\zeta)+e^{i\frac{2\pi}{3}}\textnormal{Ai}\left(\zeta e^{i\frac{2\pi}{3}}\right)+e^{-i\frac{2\pi}{3}}\textnormal{Ai}\left(\zeta e^{-i\frac{2\pi}{3}}\right) = 0.
\end{equation}
Now introduce for $\zeta\in\mathbb{C}$ the following unimodular function
\begin{equation}\label{p1}
	A_0(\zeta) = i\sqrt{\pi}e^{-i\frac{\pi}{6}\sigma_3}\begin{pmatrix}
	e^{-i\frac{\pi}{6}} & 0\\
	0 & i\\
	\end{pmatrix}\begin{pmatrix}
	\textnormal{Ai}\big(e^{-i\frac{2\pi}{3}}\zeta\big)&\textnormal{Ai}(\zeta)\smallskip\\
	e^{-i\frac{2\pi}{3}}\textnormal{Ai}'\big(e^{-i\frac{2\pi}{3}}\zeta\big) & \textnormal{Ai}'(\zeta)\\
	\end{pmatrix}e^{i\frac{\pi}{6}\sigma_3},\hspace{0.5cm}\zeta\in\mathbb{C},
\end{equation}
where $\textnormal{Ai}'(x)=\frac{\d}{\d x}\textnormal{Ai}(x)$ and which implies
\begin{equation*}
	A_0(\zeta) = \zeta^{-\frac{1}{4}\sigma_3}\frac{i}{2}\begin{pmatrix}
	1 & -i\\
	-1 & -i\\
	\end{pmatrix}\left[I+\frac{1}{48\zeta^{\frac{3}{2}}}\begin{pmatrix}-1 & 6i\\
	6i & 1\\
	\end{pmatrix}+\mathcal{O}\big(\zeta^{-\frac{6}{2}}\big)\right]e^{\frac{2}{3}\zeta^{\frac{3}{2}}\sigma_3}
\end{equation*}
as $\zeta\rightarrow\infty$ for $0<\textnormal{arg}\,\zeta<\pi$. For $\pi<\textnormal{arg}\,\zeta<\frac{5\pi}{3}$ we use \eqref{p0} and derive as $\zeta\rightarrow\infty$ in the latter sector
\begin{equation}\label{p12}
	A_0(\zeta) = \zeta^{-\frac{1}{4}\sigma_3}\frac{i}{2}\begin{pmatrix}
	1 & -i\\
	-1 & -i\\
	\end{pmatrix}\left[I+\frac{1}{48\zeta^{\frac{3}{2}}}\begin{pmatrix}-1 & 6i\\
	6i & 1\\
	\end{pmatrix}+\mathcal{O}\big(\zeta^{-\frac{6}{2}}\big)\right]e^{\frac{2}{3}\zeta^{\frac{3}{2}}\sigma_3}\begin{pmatrix}
	1 & 1\\
	0 & 1\\
	\end{pmatrix}.
\end{equation}
Finally in case $\frac{5\pi}{3}<\textnormal{arg}\,\zeta<2\pi$ we use \eqref{p00} and obtain
\begin{equation}\label{p13}
	A_0(\zeta) = \zeta^{-\frac{1}{4}\sigma_3}\frac{i}{2}\begin{pmatrix}
	1 & -i\\
	-1 & -i\\
	\end{pmatrix}\left[I+\frac{1}{48\zeta^{\frac{3}{2}}}\begin{pmatrix}-1 & 6i\\
	6i & 1\\
	\end{pmatrix}+\mathcal{O}\big(\zeta^{-\frac{6}{2}}\big)\right]e^{\frac{2}{3}\zeta^{\frac{3}{2}}\sigma_3}\begin{pmatrix}
	1 & 0\\
	-1 & 1\\
	\end{pmatrix}\begin{pmatrix}
	1 & 1\\
	0 & 1\\
	\end{pmatrix}
\end{equation}
as $\zeta\rightarrow\infty$. Keeping these properties in mind we obtain the solution to the $A^{RH}$-RHP for $\gamma=1$
\begin{equation}\label{p2}
	A^{RH}(\zeta) = \left\{
                                   \begin{array}{ll}
                                     A_0(\zeta), & \hbox{$\textnormal{arg}\,\zeta\in(0,\frac{2\pi}{3})$,}\smallskip \\
                                     A_0(\zeta)\begin{pmatrix}
                                     1 & -1\\
                                     0 & 1\\
                                     \end{pmatrix}, & \hbox{$\textnormal{arg}\,\zeta\in(\frac{2\pi}{3},\frac{4\pi}{3})$,}\smallskip \\
                                     A_0(\zeta)\begin{pmatrix}
                                     1 & -1\\
                                     0 & 1\\
                                     \end{pmatrix}\begin{pmatrix}
                                     1 & 0\\
                                     1 & 1\\
                                     \end{pmatrix}, & \hbox{$\textnormal{arg}\,\zeta\in(\frac{4\pi}{3},2\pi)$.}
                                   \end{array}
                                 \right.
\end{equation}
$A^{RH}(\zeta)$ has the following properties:
\begin{itemize}
	\item $A^{RH}(\zeta)$ is analytic for $\zeta\in\mathbb{C}\backslash \Gamma_{\zeta}$ where $\Gamma_{\zeta}$ is depicted in Figure \ref{figure5}.
	\item The following jump relations hold, with orientiation as in Figure \ref{figure5}:
	\begin{eqnarray*}
		A_+^{RH}(\zeta) &=&A_-^{RH}(\zeta)\begin{pmatrix}
		1 & 1\\
		0 & 1\\
		\end{pmatrix},\hspace{0.5cm}\textnormal{arg}\,\zeta=\frac{2\pi}{3},\\
		A_+^{RH}(\zeta)&=&A_-^{RH}(\zeta)\begin{pmatrix}
		1 & 0\\
		-1 & 1\\
		\end{pmatrix},\hspace{0.5cm}\textnormal{arg}\,\zeta=\frac{4\pi}{3},\\
		A_+^{RH}(\zeta)&=&A_-^{RH}(\zeta)\begin{pmatrix}
		1 & 1\\
		-1 & 0\\
		\end{pmatrix},\hspace{0.5cm}\textnormal{arg}\,\zeta=0.
	\end{eqnarray*}
	\item As $\zeta\rightarrow\infty$, our previous estimations imply directly for $\textnormal{arg}\,\zeta\in(0,\frac{2\pi}{3})$ 
	\begin{equation}\label{p3}
		A^{RH}(\zeta) = \zeta^{-\frac{1}{4}\sigma_3}\frac{i}{2}\begin{pmatrix}
	1 & -i\\
	-1 & -i\\
	\end{pmatrix}\left[I+\frac{1}{48\zeta^{\frac{3}{2}}}\begin{pmatrix}-1 & 6i\\
	6i & 1\\
	\end{pmatrix}+\mathcal{O}\big(\zeta^{-\frac{6}{2}}\big)\right]e^{\frac{2}{3}\zeta^{\frac{3}{2}}\sigma_3}.
	\end{equation}
	However, by construction, the latter asymptotics are in fact valid in a full neighborhood of infinity. If $\textnormal{arg}\,\zeta\in(\frac{2\pi}{3},\pi)$ we notice that
	\begin{equation*}
		e^{\frac{2}{3}\zeta^{\frac{3}{2}}\sigma_3}\begin{pmatrix}
		1 & -1 \\
		0 & 1\\
		\end{pmatrix}e^{-\frac{2}{3}\zeta^{\frac{3}{2}}\sigma_3} = \begin{pmatrix}
		1 & -e^{\frac{4}{3}\zeta^{\frac{3}{2}}}\\
		0 & 1\\
		\end{pmatrix}
	\end{equation*}
	but $\textnormal{Re}\big(\zeta^{\frac{3}{2}}\big)<0$ for $\textnormal{arg}\,\zeta\in(\frac{2\pi}{3},\pi)$, hence the latter matrix product approaches the identitiy matrix exponentially fast and we therefore restore \eqref{p3} in the sector under consideration. For the remaining sectors we use \eqref{p12} and \eqref{p13} to deduce \eqref{p3}.
\end{itemize}
The model function $A^{RH}(\zeta)$ will now be used in the construction of the approximate parametrix to the original $S$-RHP near $\lambda=a$. Set
\begin{equation}\label{p4}
	U(\lambda) = B_{r_1}(\lambda)A^{RH}\big(\zeta(\lambda)\big)e^{-\frac{2}{3}\zeta^{\frac{3}{2}}(\lambda)\sigma_3},\hspace{0.5cm}|\lambda-a|<r,
\end{equation}
where $\zeta=\zeta(\lambda)$ is given in \eqref{change1} and we introduced the matrix multiplier
\begin{equation*}
	B_{r_1}(\lambda) = M(\lambda)\begin{pmatrix}
	-i & i\\
	1 & 1\\
	\end{pmatrix}\delta(\lambda)^{\sigma_3}\left(\zeta(\lambda)\frac{\lambda-1}{\lambda-a}\right)^{\frac{1}{4}\sigma_3},\hspace{0.5cm}\delta(\lambda)=\left(\frac{\lambda-a}{\lambda-1}\right)^{\frac{1}{4}}\rightarrow 1,\hspace{0.5cm}\lambda\rightarrow\infty,
\end{equation*}
which involves the model function $M(\lambda)$ given in \eqref{Mdef}. Notice first, that $B_{r_1}(\lambda)$ is analytic in a full neighborhood of $\lambda=a$: For $\lambda\in(a,a+r)$
\begin{equation*}
	\left(B_{r_1}(\lambda)\right)_+ = M_-(\lambda)\begin{pmatrix}
	0 & 1\\
	-1 & 0\\
	\end{pmatrix}\begin{pmatrix}
	-i & i\\
	1 & 1\\
	\end{pmatrix}\delta_-(\lambda)^{\sigma_3}\begin{pmatrix}
	-i & 0\\
	0 & i\\
	\end{pmatrix}\left(\zeta(\lambda)\frac{\lambda-1}{\lambda-a}\right)^{\frac{1}{4}\sigma_3}_- = \left(B_{r_1}(\lambda)\right)_-,
\end{equation*}
hence $B_{r_1}(\lambda)$ is analytic across $(a,a+r)$ with a possible singularity at $\lambda=a$.  But since this singularity is at worst of square root type, it follows that the singularity at $\lambda=a$ is removable, and hence analyticity of $B_{r_1}(\lambda)$ in the whole disk around $\lambda=a$ follows. Alternatively, this follows from \eqref{MMhat}. Secondly, the parametrix $U(\lambda)$ has jumps along the curves as depicted in Figure \ref{figure4}. We assume that the contour in the $S$-RHP coincides with $\Gamma_{\zeta}$ in $\mathcal{U}$. This matching is achieved by a standard deformation argument.\smallskip

The role of the left multiplier $B_{r_1}$ in \eqref{p4} is to ensure a matching-relation between the local model function $U(\lambda)$ and the outer parametrix $M(\lambda)$: Observe that
\begin{equation*}
	B_{r_1}(\lambda)\zeta^{-\frac{1}{4}\sigma_3}(\lambda)\frac{i}{2}\begin{pmatrix}
	1 & -i\\
	-1 & -i\\
	\end{pmatrix} = M(\lambda),
\end{equation*}
which, in turn, implies with \eqref{p3} that
\begin{eqnarray}
	U(\lambda)&=&M(\lambda)\left[I+\frac{1}{48\zeta^{\frac{3}{2}}}\begin{pmatrix}
	-1 & 6i\\
	6i & 1\\
	\end{pmatrix}+\mathcal{O}\big(\zeta^{-\frac{6}{2}}\big)\right]\nonumber\\
	&=&\left[I+\big(Q(\infty)\big)^{-1}Q(\lambda)\left\{\frac{1}{48\zeta^{\frac{3}{2}}}\begin{pmatrix}
	-1 & 6i\\
	6i & 1\\
	\end{pmatrix}+\mathcal{O}\left(\zeta^{-\frac{6}{2}}\right)\right\}\big(Q(\lambda)\big)^{-1}Q(\infty)\right]M(\lambda)\label{p5}
\end{eqnarray}
as $s\rightarrow\infty$ and $0<r_1\leq|\lambda-a|\leq r_2<r=\min\left\{\frac{a}{2},\frac{1}{2}(1-a)\right\}$ (hence $|\zeta|\rightarrow\infty$). Since the function $\zeta(\lambda)$ is of order $\mathcal{O}\big(s^{\frac{2}{3}}\big)$ on the latter annulus and $Q(\lambda)$ is bounded in $s$, equation \eqref{p5} yields 
\begin{equation*}
	U(\lambda)=\big(I+o(1)\big)M(\lambda),\hspace{0.5cm}s\rightarrow\infty,\ 0<r_1\leq|\lambda-a|\leq r_2<\frac{a}{2},
\end{equation*}
which is very important for our subsequent steps.
\subsection{Approximate parametrix near the branch point $\lambda=-a$}

We again construct an approximate parametrix in a small neigbhorhood $\mathcal{V}$ of radius less than $r=\min\left\{\frac{a}{2},\frac{1}{2}(1-a)\right\}$ as shown in Figure \ref{figure6}. First note that by Proposition \ref{prop2}
\begin{equation*}
	\Pi(\lambda) = -2\int_{-\lambda}^1\sqrt{\frac{\mu^2-a^2}{1-\mu^2}}\,\d\mu = -2\kappa +ic_0\left(\lambda+a\right)^{\frac{3}{2}}+\mathcal{O}\left((\lambda+a)^{\frac{5}{2}}\right)
\end{equation*}
for $\lambda\in\mathcal{V}\cap(-1,-a)$, where $c_0$ is given in \eqref{app:1} and the function $(z+a)^{\frac{3}{2}}$ is defined for $z\in\mathbb{C}\backslash(-\infty,-a]$ with its branch fixed by the condition
\begin{equation*}
	(z+a)^{\frac{3}{2}}>0\hspace{0.5cm}\textnormal{if}\ z>-a.
\end{equation*}
\begin{figure}[tbh]
\begin{center}
\includegraphics[width=0.5\textwidth]{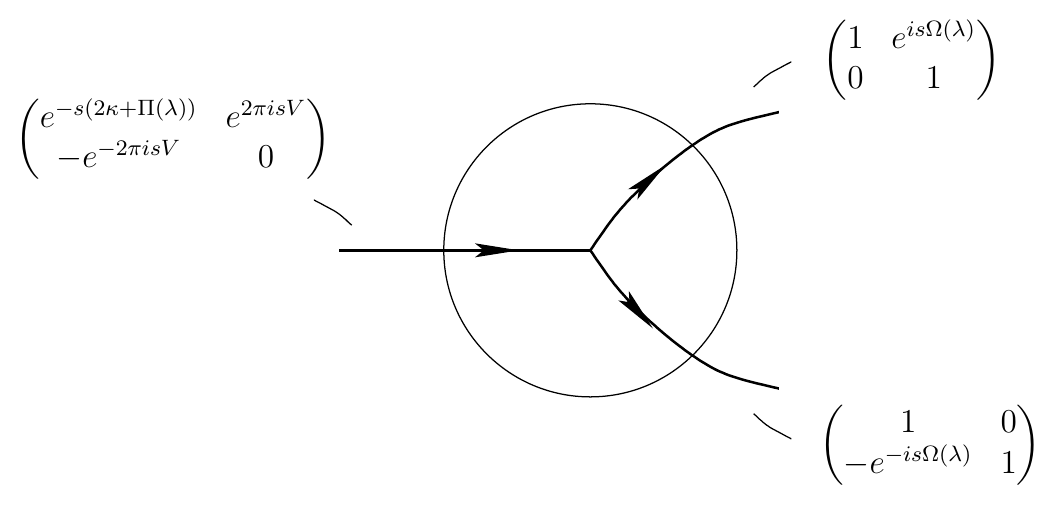}
\end{center}
\caption{A neighborhood $\mathcal{V}$ of $\lambda=-a$ with approximate jumps}
\label{figure6}
\end{figure}

Furthermore,
\begin{equation*}
	i\Omega(\lambda) = -2i\int_{\lambda}^a\sqrt{\frac{a^2-\mu^2}{1-\mu^2}}\,\d\mu = 2\pi i V+ic_0(\lambda+a)^{\frac{3}{2}}+\mathcal{O}\left((\lambda+a)^{\frac{5}{2}}\right),\hspace{0.5cm}\lambda\in\mathcal{V}\cap \gamma^+
\end{equation*}
and
\begin{equation*}
	-i\Omega(\lambda) = -2\pi iV-ic_0\left(\lambda+a\right)^{\frac{3}{2}}+\mathcal{O}\left((\lambda+a)^{\frac{5}{2}}\right),\hspace{0.5cm}\lambda\in\mathcal{V}\cap\gamma^-.
\end{equation*}
These behaviors suggest to consider the following change of variable
\begin{eqnarray}\label{p6}
	\zeta(\lambda) &=& \left(\frac{3s}{4}\right)^{\frac{2}{3}}\left(2\int_{-a}^{\lambda}\sqrt{\frac{\mu^2-a^2}{\mu^2-1}}\,\d\mu\right)^{\frac{2}{3}}\\
	&=&\left(s\sqrt{\frac{2a}{1-a^2}}\,\right)^{\frac{2}{3}}(\lambda+a)\left(1-\frac{1+3a^2}{10a(1-a^2)}(\lambda+a)+\mathcal{O}\big((\lambda+a)^2\big)\right),\hspace{0.5cm}|\lambda+a|<r\nonumber
\end{eqnarray}
which allows us to factorize the jump matrices depicted in Figure \ref{figure6}:
\begin{eqnarray}
	\begin{pmatrix}
	1 & e^{is\Omega(\lambda)}\\
	0 & 1\\
	\end{pmatrix} &=&e^{\frac{2}{3}i\zeta^{\frac{3}{2}}(\lambda)\sigma_3}e^{i\pi sV\sigma_3}\begin{pmatrix}
	1 & 1\\
	0 & 1\\
	\end{pmatrix}e^{-i\pi sV\sigma_3}e^{-\frac{2}{3}i\zeta^{\frac{3}{2}}(\lambda)\sigma_3},\hspace{0.2cm}\lambda\in\mathcal{V}\cap\gamma^+\label{fac1}\\
	\begin{pmatrix}
	1 & 0\\
	-e^{-is\Omega(\lambda)} & 1\\
	\end{pmatrix}&=&e^{\frac{2}{3}i\zeta^{\frac{3}{2}}(\lambda)\sigma_3}e^{i\pi sV\sigma_3}\begin{pmatrix}
	1 & 0\\
	-1 & 1\\
	\end{pmatrix}e^{-i\pi sV\sigma_3}e^{-\frac{2}{3}i\zeta^{\frac{3}{2}}(\lambda)\sigma_3},\ \lambda\in\mathcal{V}\cap\gamma^-\ \ \ \label{fac2}\\
	\begin{pmatrix}
	e^{-s(2\kappa+\Pi(\lambda))}& e^{2\pi isV}\\
	-e^{-2\pi isV} & 0\\
	\end{pmatrix}&=&e^{\frac{2}{3}i\zeta_-^{\frac{3}{2}}(\lambda)\sigma_3}e^{i\pi sV\sigma_3}\begin{pmatrix}
	1 & 1\\
	-1 & 0\\
	\end{pmatrix}e^{-i\pi sV\sigma_3}e^{-\frac{2}{3}i\zeta_+^{\frac{3}{2}}(\lambda)\sigma_3},\label{fac3}
\end{eqnarray} 
where we choose $\lambda\in\mathcal{V}\cap(-1,-a)$ in the latter identity. Similarly to \eqref{p1}, define for $\zeta\in\mathbb{C}$
\begin{equation}\label{p7}
	\tilde{A}_0(\zeta) = i\sqrt{\pi}e^{i\pi sV\sigma_3}e^{-i\frac{\pi}{3}\sigma_3}\begin{pmatrix}
	e^{-i\frac{\pi}{4}}e^{-2\pi isV} & 0\\
	0 & e^{-i\frac{5\pi}{12}}\\
	\end{pmatrix}\begin{pmatrix}
	e^{i\pi}\textnormal{Ai}'\left(e^{i\pi}\zeta\right) & -e^{i\frac{\pi}{3}}\textnormal{Ai}'\left(e^{i\frac{\pi}{3}}\zeta\right)\smallskip\\
	-\textnormal{Ai}\left(e^{i\pi}\zeta\right) & \textnormal{Ai}\left(e^{i\frac{\pi}{3}}\zeta\right)\\
	\end{pmatrix}e^{i\frac{\pi}{3}\sigma_3}e^{-i\pi sV\sigma_3}
\end{equation}
which implies that
\begin{equation*}
	\tilde{A}_0(\zeta) = \zeta^{\frac{1}{4}\sigma_3}\frac{i}{2}e^{-2\pi isV}\begin{pmatrix}
	1 & -ie^{2\pi isV}\\
	-1 & -ie^{2\pi isV}\\
	\end{pmatrix}\left[I+\frac{i}{48\zeta^{\frac{3}{2}}}\begin{pmatrix}
	1 & 6i e^{2\pi isV}\\
	6i e^{-2\pi isV} & -1\\
	\end{pmatrix}+\mathcal{O}\left(\zeta^{-\frac{6}{2}}\right)\right]e^{\frac{2}{3}i\zeta^{\frac{3}{2}}\sigma_3}
\end{equation*}
as $\zeta\rightarrow\infty$ in the sector $-\pi<\textnormal{arg}\,\zeta<0$. In the other sectors, similar formulae hold: 
\begin{eqnarray*}
	\tilde{A}_0(\zeta) &=&\zeta^{\frac{1}{4}\sigma_3}\frac{i}{2}e^{-2\pi isV}\begin{pmatrix}
	1 & -ie^{2\pi isV}\\
	-1 & -ie^{2\pi isV}\\
	\end{pmatrix}\left[I+\frac{i}{48\zeta^{\frac{3}{2}}}\begin{pmatrix}
	1 & 6i e^{2\pi isV}\\
	6i e^{-2\pi isV} & -1\\
	\end{pmatrix}+\mathcal{O}\left(\zeta^{-\frac{6}{2}}\right)\right]\\
	&&\times e^{\frac{2}{3}i\zeta^{\frac{3}{2}}\sigma_3} \begin{pmatrix}
	1 & 0\\
	e^{-2\pi isV} & 1\\
	\end{pmatrix},\hspace{0.5cm}\zeta\rightarrow\infty,\ \ 0<\textnormal{arg}\,\zeta<\frac{2\pi}{3}
\end{eqnarray*}
and 
\begin{eqnarray*}
	\tilde{A}_0(\zeta) &=&\zeta^{\frac{1}{4}\sigma_3}\frac{i}{2}e^{-2\pi isV}\begin{pmatrix}
	1 & -ie^{2\pi isV}\\
	-1 & -ie^{2\pi isV}\\
	\end{pmatrix}\left[I+\frac{i}{48\zeta^{\frac{3}{2}}}\begin{pmatrix}
	1 & 6i e^{2\pi isV}\\
	6i e^{-2\pi isV} & -1\\
	\end{pmatrix}+\mathcal{O}\left(\zeta^{-\frac{6}{2}}\right)\right]\\
	&&\times e^{\frac{2}{3}i\zeta^{\frac{3}{2}}\sigma_3}\begin{pmatrix}
	1 & -e^{2\pi isV}\\
	0 & 1\\
	\end{pmatrix} \begin{pmatrix}
	1 & 0\\
	e^{-2\pi isV} & 1\\
	\end{pmatrix},\hspace{0.5cm}\zeta\rightarrow\infty,\ \ \frac{2\pi}{3}<\textnormal{arg}\,\zeta<\pi.
\end{eqnarray*}
These properties suggest the definition (compare \eqref{p2}):
\begin{equation}\label{p8}
	\tilde{A}^{RH}(\zeta) = \left\{
                                   \begin{array}{ll}
                                     \tilde{A}_0(\zeta), & \hbox{$\textnormal{arg}\,\zeta\in(-\pi,-\frac{\pi}{3})$,}\smallskip \\
                                     \tilde{A}_0(\zeta)\begin{pmatrix}
                                     1 & 0\\
                                     -e^{-2\pi isV} & 1\\
                                     \end{pmatrix}, & \hbox{$\textnormal{arg}\,\zeta\in(-\frac{\pi}{3},\frac{\pi}{3})$,}\smallskip \\
                                     \tilde{A}_0(\zeta)\begin{pmatrix}
                                     1 & 0\\
                                     -e^{-2\pi isV} & 1\\
                                     \end{pmatrix}\begin{pmatrix}
                                     1 & e^{2\pi isV}\\
                                     0 & 1\\
                                     \end{pmatrix}, & \hbox{$\textnormal{arg}\,\zeta\in(\frac{\pi}{3},\pi)$,}
                                   \end{array}
                                 \right.
\end{equation}
The model function $\tilde{A}^{RH}(\zeta)$ satisfies the following RHP:
\begin{itemize}
	\item $\tilde{A}^{RH}(\zeta)$ is analytic for $\zeta\in\mathbb{C}\backslash\tilde{\Gamma}_{\zeta}$, where $\tilde{\Gamma}_{\zeta}$ is depicted in Figure \ref{figure7}
	\item The function $\tilde{A}^{RH}(\zeta)$ satifies the following jump relations
	\begin{eqnarray*}
		\tilde{A}^{RH}_+(\zeta)&=&\tilde{A}^{RH}_-(\zeta)\begin{pmatrix}
		1 & 0\\
		-e^{-2\pi isV} & 1\\
		\end{pmatrix},\hspace{0.5cm}\textnormal{arg}\,\zeta=-\frac{\pi}{3}\\
		\tilde{A}^{RH}_+(\zeta)&=&\tilde{A}^{RH}_-(\zeta)\begin{pmatrix}
		1 & e^{2\pi isV}\\
		0 & 1\\
		\end{pmatrix},\hspace{0.5cm}\textnormal{arg}\,\zeta=\frac{\pi}{3}\\
		\tilde{A}^{RH}_+(\zeta)&=&\tilde{A}^{RH}_-(\zeta)\begin{pmatrix}
		1 & e^{2\pi isV}\\
		-e^{-2\pi isV} & 0\\
		\end{pmatrix},\hspace{0.5cm}\textnormal{arg}\,\zeta=\pi\\
	\end{eqnarray*}
	\item As $\zeta\rightarrow\infty$, a similar argument as in the construction of \eqref{p2} leads us to the following asymptotics
	\begin{eqnarray*}
		\tilde{A}^{RH}(\zeta) &=& \zeta^{\frac{1}{4}\sigma_3}\frac{i}{2}e^{-2\pi isV}\begin{pmatrix}
	1 & -ie^{2\pi isV}\\
	-1 & -ie^{2\pi isV}\\
	\end{pmatrix}\bigg[I+\frac{i}{48\zeta^{\frac{3}{2}}}e^{i\pi sV\sigma_3}\begin{pmatrix}
	1 & 6i\\
	6i & -1\\
	\end{pmatrix}e^{-i\pi sV\sigma_3}\\
	&&+\mathcal{O}\left(\zeta^{-\frac{6}{2}}\right)\bigg]e^{\frac{2}{3}i\zeta^{\frac{3}{2}}\sigma_3}
	\end{eqnarray*}
\end{itemize}
\begin{figure}[tbh]
\begin{center}
\includegraphics[width=0.33\textwidth]{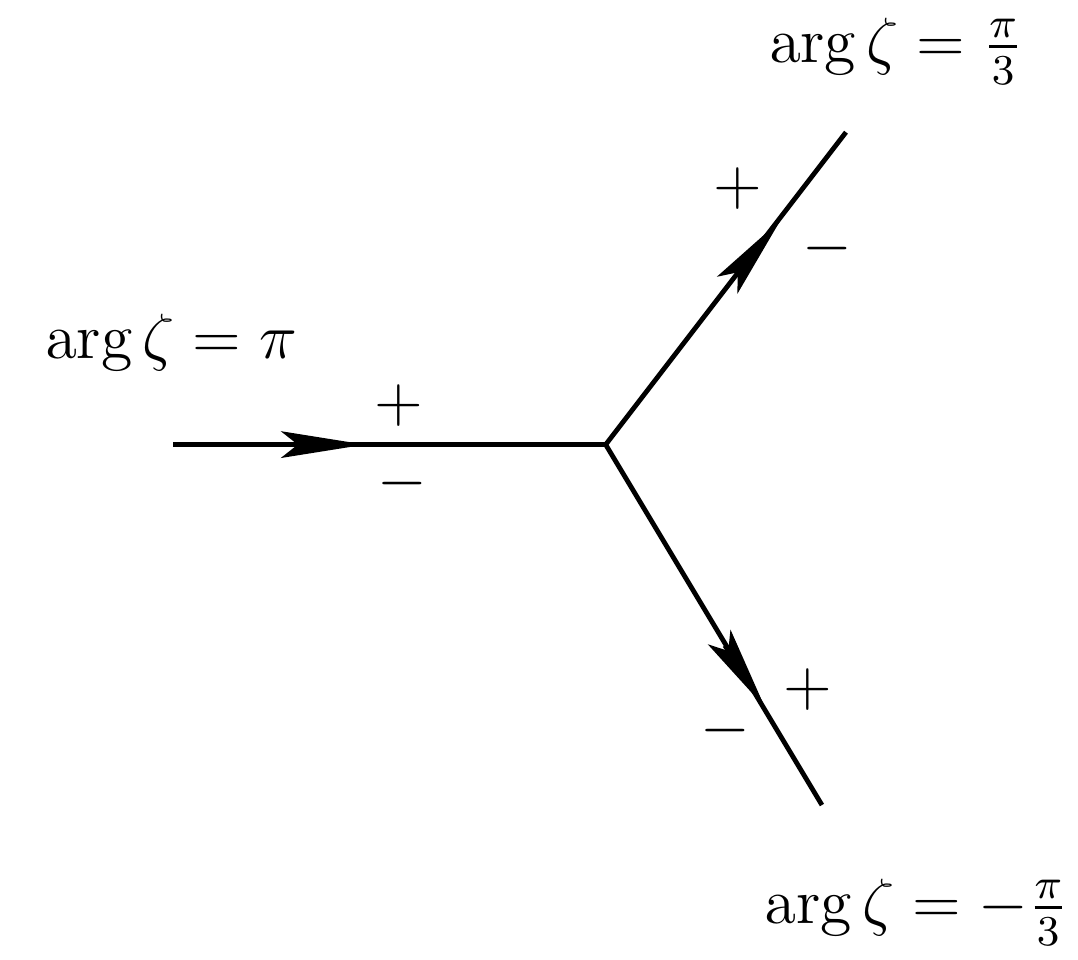}
\end{center}
\caption{Jump contour $\tilde{\Gamma}_{\zeta}$ for the model function $\tilde{A}^{RH}(\zeta)$}
\label{figure7}
\end{figure}

We now define the approximate parametrix to the original $S$-RHP near $\lambda=-a$. Set
\begin{equation}\label{p9}
	V(\lambda)=B_{l_1}(\lambda) \tilde{A}^{RH}\big(\zeta(\lambda)\big)e^{-\frac{2}{3}i\zeta^{\frac{3}{2}}(\lambda)\sigma_3},\hspace{0.5cm}|\lambda+a|<r
\end{equation}
where $\zeta=\zeta(\lambda)$ is given in \eqref{p6} and we use the multiplier
\begin{equation*}
	B_{l_1}(\lambda) = M(\lambda)\begin{pmatrix}
	-ie^{2\pi isV} & ie^{2\pi isV}\\
	1 & 1\\
	\end{pmatrix}\tilde{\delta}(\lambda)^{-\sigma_3}\left(\zeta(\lambda)\frac{\lambda+1}{\lambda+a}\right)^{-\frac{1}{4}\sigma_3},\hspace{0.25cm}\tilde{\delta}(\lambda) = \left(\frac{\lambda+a}{\lambda+1}\right)^{\frac{1}{4}}\rightarrow 1,\ \ \lambda\rightarrow\infty.
\end{equation*}
Here $B_{l_1}(\lambda)$ is analytic in a full neighborhood of $\lambda=-a$: Indeed for $\lambda\in(-a-r,-a)$
\begin{eqnarray*}
	\left(B_{l_1}(\lambda)\right)_+ &=& M_-(\lambda)\begin{pmatrix}
	0 & e^{2\pi isV}\\
	-e^{-2\pi isV} & 0\\
	\end{pmatrix}\begin{pmatrix}
	-i e^{2\pi isV} & ie^{2\pi isV}\\
	1 & 1\\
	\end{pmatrix}\begin{pmatrix}
	-i & 0\\
	0 & i\\
	\end{pmatrix}\\
	&&\times\tilde{\delta}_-(\lambda)^{-\sigma_3}\left(\zeta(\lambda)\frac{\lambda+1}{\lambda+a}\right)^{-\frac{1}{4}\sigma_3} = \left(B_{l_1}(\lambda)\right)_-
\end{eqnarray*}
and since the singularity of $B_{l_1}(\lambda)$ at $\lambda=-a$ is at worst of square root type, it is in fact a removable singularity, hence analyticity follows. Next, the parametrix $V(\lambda)$ has jumps along the curves depcited in Figure \ref{figure7}, and we can again locally match these jump contours with the ones in the original $S$-RHP near $\lambda=-a$. The jump matrices coincide with those of the $S$-RHP for $\gamma=1$ in $\mathcal{V}$.\smallskip

The exact form of the left multiplier $B_{l_1}(\lambda)$ follows here also from the asymptotical matching relation between the model functions $V(\lambda)$ and $M(\lambda)$. Since
\begin{equation*}
	B_{l_1}(\lambda)\zeta^{\frac{1}{4}\sigma_3}(\lambda)\frac{i}{2}e^{-2\pi isV}\begin{pmatrix}
	1 & -ie^{2\pi isV}\\
	-1 & -ie^{2\pi isV}\\
	\end{pmatrix}=M(\lambda)
\end{equation*}
we obtain
\begin{align}
	&V(\lambda)=M(\lambda)\left[I+\frac{i}{48\zeta^{\frac{3}{2}}}e^{i\pi sV\sigma_3}\begin{pmatrix}
	1 & 6i\\
	6i & -1\\
	\end{pmatrix}e^{-i\pi sV\sigma_3}+\mathcal{O}\left(\zeta^{-\frac{6}{2}}\right)\right]\nonumber\\
	&=\bigg[I+\big(Q(\infty)\big)^{-1}Q(\lambda)\left\{\frac{i}{48\zeta^{\frac{3}{2}}}e^{i\pi sV\sigma_3}\begin{pmatrix}
	1 & 6i\\
	6i & -1\\
	\end{pmatrix}e^{-i\pi sV\sigma_3}+\mathcal{O}\left(\zeta^{-\frac{6}{2}}\right)\right\}\big(Q(\lambda)\big)^{-1}Q(\infty)\bigg]M(\lambda)\label{p10}
\end{align}
as $s\rightarrow\infty$ and $0<r_1\leq|\lambda+a|\leq r_2<r=\min\left\{\frac{a}{2},\frac{1}{2}(1-a)\right\}$ (i.e. $|\zeta|\rightarrow\infty$). Since $\zeta(\lambda)$ is of order $\mathcal{O}\big(s^{\frac{2}{3}}\big)$ on the latter annulus and $Q(\lambda)$ is bounded in $s$, equation \eqref{p10} implies
\begin{equation*}
	V(\lambda) = \big(I+o(1)\big)M(\lambda),\hspace{0.5cm}s\rightarrow\infty,\ \ 0<r_1\leq|\lambda+a|\leq r_2<r.
\end{equation*}
\subsection{Exact parametrix at the branch point $\lambda=1$}\label{rendp}
We will manage to construct exact model functions near the endpoints $\lambda=\pm 1$, however also here, compared to the ``standard'' cases as in \cite{DIZ}, novel features in the parametrix analysis appear. We will use Bessel functions as in \cite{DIZ}, but they will have to be combined with certain logarithmic ``corrections'':\smallskip

Fix a neighborhood $\hat{\mathcal{U}}$ of radius less than $\frac{1}{2}(1-a)$ as shown in Figure \ref{figure8}. We want to construct a model function which has the depicted jump condition for the $S$-RHP and which matches $M(\lambda)$ to leading order as $s\rightarrow\infty,\gamma\uparrow 1$ on the boundary $\partial\hat{\mathcal{U}}$. To this end notice that
\begin{equation*}
	\Pi(\lambda) = -2\int_{\lambda}^1\sqrt{\frac{\mu^2-a^2}{1-\mu^2}}\,\d\mu = -d_0\left|1-\lambda\right|^{\frac{1}{2}}+\mathcal{O}\left(|\lambda-1|^{\frac{3}{2}}\right),\hspace{0.5cm}d_0 = \sqrt{8(1-a^2)}>0
\end{equation*}
for $\lambda\rightarrow 1,\lambda\in\hat{\mathcal{U}}\cap(a,1)$.
\begin{figure}[tbh]
\begin{center}
\includegraphics[width=0.45\textwidth]{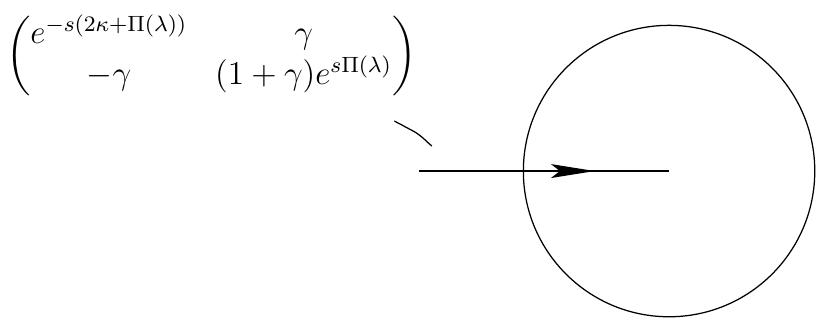}
\end{center}
\caption{A neighborhood $\hat{\mathcal{U}}$ of $\lambda=1$}
\label{figure8}
\end{figure}

This observation combined with the fact that 
\begin{equation*}
	S(\lambda) = \mathcal{O}\big(\left|\ln|\lambda-1|\right|),\hspace{0.5cm}\lambda\rightarrow 1
\end{equation*}
suggests to use the Bessel functions $H_{0}^{(1)}(\zeta)$ and $H_{0}^{(2)}(\zeta)$ for our construction. This idea can be justified rigorously as follows. Recall that the stated Hankel functions are unique linearly independent solutions to Bessel's equation
\begin{equation*}
	zw''+w'+zw=0,
\end{equation*}
satisfying the following asymptotics as $\zeta\rightarrow\infty$ in the sector $-\pi<\textnormal{arg}\,\zeta<\pi$ (cf. \cite{BE})
\begin{eqnarray*}
	H_{0}^{(1)}(\zeta)&\sim&\sqrt{\frac{2}{\pi\zeta}}e^{i(\zeta-\frac{\pi}{4})}\left(1-\frac{i}{8\zeta}-\frac{9}{128\zeta^2}
	+\mathcal{O}\big(\zeta^{-3}\big)\right)\\
	H_{0}^{(2)}(\zeta)&\sim&\sqrt{\frac{2}{\pi\zeta}}e^{-i(\zeta-\frac{\pi}{4})}\left(1+\frac{i}{8\zeta}-\frac{9}{128\zeta^2}
	+\mathcal{O}\big(\zeta^{-3}\big)\right),
\end{eqnarray*}
where $\sqrt{\zeta}$ is taken positive for $\zeta>0$, and defined with $-\pi<\textnormal{arg}\,\zeta<\pi$. We also need the following monodromy relations, valid on the entire universal covering of the punctured plane
\begin{eqnarray*}
	H_{0}^{(1)}\big(\zeta e^{i\pi}\big) &=& -H_{0}^{(2)}(\zeta),\hspace{0.5cm} H_{0}^{(2)}\big(\zeta e^{-i\pi}\big)=-H_{0}^{(1)}(\zeta),\\
	 H_{0}^{(2)}\big(\zeta e^{i\pi}\big)&=& H_{0}^{(1)}(\zeta)+2H_{0}^{(2)}(\zeta),
\end{eqnarray*}
and the following expansions at the origin (compare to \eqref{Xsing})
\begin{equation}\label{Hankelorigin}
	H_0^{(1)}(\zeta) = a_0+a_1\ln\zeta+a_2\zeta^2+a_3\zeta^2\ln\zeta +\mathcal{O}\big(\zeta^4\ln\zeta),\ \zeta\rightarrow 0,
\end{equation}
with coefficients $a_i$ given as
\begin{equation}\label{excoeff}
	a_0=1+\frac{2i\gamma_E}{\pi}-\frac{2i}{\pi}\ln 2,\ \ a_1=\frac{2i}{\pi},\ \ a_2=\frac{i}{2\pi}(1-\gamma_E)-\frac{1}{4}+\frac{i}{2\pi}\ln 2,\ \ a_3=-\frac{i}{2\pi}
\end{equation}
where $\gamma_E$ is Euler's constant. The expansion for $H_0^{(2)}(\zeta)$ is with the replacement $a_i\mapsto \bar{a}_i$, of the coefficient $a_i$ by its complex conjugate, identical to \eqref{Hankelorigin}.\footnote{In more detail, we have $H_0^{(1)}(\zeta) = h_1(\zeta)+h_2(\zeta)\ln \zeta,\zeta\rightarrow 0$ where $h_j(\zeta),j=1,2$ are analytic at $\zeta=0$. A similar statement is true for $H_0^{(2)}(\zeta)$.\label{foot1}} Let
\begin{equation}\label{p15}
	P_{BE}^{RH}(\zeta) = \sqrt{\frac{\pi}{2}}\begin{pmatrix}
	H_{0}^{(1)}(\sqrt{\zeta}) & H_{0}^{(2)}(\sqrt{\zeta})\smallskip\\
	\sqrt{\zeta}\big(H_{0}^{(1)}\big)'(\sqrt{\zeta}) & \sqrt{\zeta}\big(H_{0}^{(2)}\big)'(\sqrt{\zeta})\\
	\end{pmatrix},\hspace{0.2cm}-\pi<\textnormal{arg}\,\zeta\leq\pi,
\end{equation}
on the punctured plane $\zeta\in\mathbb{C}\backslash\{0\}$ where $H'(x)=\frac{\d}{\d x}H(x)$, and observe that
\begin{eqnarray*}
	P_{BE}^{RH}(\zeta) &=&\zeta^{-\frac{1}{4}\sigma_3}\begin{pmatrix}
	1 & 1\\
	i & -i\\
	\end{pmatrix}e^{-i\frac{\pi}{4}\sigma_3}\Bigg[I+\frac{i}{8\sqrt{\zeta}}\begin{pmatrix}
	1 & 2i\\
	2i & -1\\
	\end{pmatrix}+\frac{3}{128\zeta}\begin{pmatrix}
	1 & -4i\\
	4i & 1\\
	\end{pmatrix}\\
	&&+\mathcal{O}\left(\zeta^{-\frac{3}{2}}\right)\Bigg]e^{i\sqrt{\zeta}\sigma_3},\hspace{0.25cm}\zeta\rightarrow\infty,
\end{eqnarray*}
in the sector $-\pi<\textnormal{arg}\,\zeta\leq\pi$. In order to determine the behavior of the model function $P_{BE}^{RH}(\zeta)$ on the negative real axis $\textnormal{arg}\,\zeta=\pi$ (compare Figure \ref{figure8}) we use the aforementioned monodromy relations:
\begin{eqnarray*}
	H_{0}^{(2)}\big(\sqrt{\zeta}_+\big)&=&H_{0}^{(2)}\big(\sqrt{\zeta}_-e^{i\pi}\big) = H_{0}^{(1)}\big(\sqrt{\zeta}_-\big)+2H_{0}^{(2)}\big(\sqrt{\zeta}_-\big),\\
	\big(H_{0}^{(2)}\big)'\big(\sqrt{\zeta}_+\big)&=&e^{-i\pi}\Big(\big(H_{0}^{(1)}\big)'\big(\sqrt{\zeta}_-\big)+2\big(H_{0}^{(2)}\big)'\big(\sqrt{\zeta}_-\big)\Big),
\end{eqnarray*}
and
\begin{eqnarray*}
	H_{0}^{(1)}\big(\sqrt{\zeta}_+\big) &=&H_{0}^{(1)}\big(\sqrt{\zeta}_-e^{i\pi}\big) = -H_{0}^{(2)}\big(\sqrt{\zeta}_-\big),\\
	\big(H_{0}^{(1)}\big)'\big(\sqrt{\zeta}_+\big) &=&\big(H_{0}^{(2)}\big)'\big(\sqrt{\zeta}_-\big).
\end{eqnarray*}
Therefore
\begin{equation*}
	\big(P_{BE}^{RH}(\zeta)\big)_+=\big(P_{BE}^{RH}(\zeta)\big)_-\begin{pmatrix}
	0 & 1\\
	-1 & 2\\
	\end{pmatrix},\hspace{0.5cm}\textnormal{arg}\,\zeta=\pi.
\end{equation*}

Thus $P_{BE}^{RH}(\zeta)$ has the following properties.
\begin{itemize}
	\item $P_{BE}^{RH}(\zeta)$ is analytic for $\zeta\in\mathbb{C}\backslash\left\{\textnormal{arg}\,\zeta=\pi\right\}$
	\item The following jump relation holds on the line $\textnormal{arg}\,\zeta=\pi$ (see Figure \ref{figure9}):
	\begin{equation*}
		\big(P_{BE}^{RH}(\zeta)\big)_+=\big(P_{BE}^{RH}(\zeta)\big)_-\begin{pmatrix}
	0 & 1\\
	-1 & 2\\
	\end{pmatrix}.
	\end{equation*}
	\begin{figure}[tbh]
\begin{center}
\includegraphics[width=0.33\textwidth]{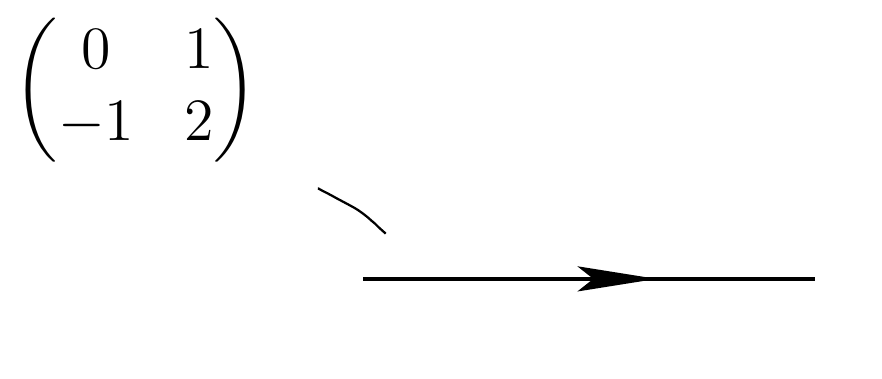}
\end{center}
\caption{A model problem near $\lambda=+1$ which can be solved explicitly using Bessel functions}
\label{figure9}
\end{figure}

	\item As $\zeta\rightarrow 0$ with $-\pi<\textnormal{arg}\,\zeta\leq \pi$
	\begin{equation}\label{Psing}
		P_{BE}^{RH}(\zeta) =\sqrt{\frac{\pi}{2}}\Bigg[\begin{pmatrix}
		a_0 & \bar{a}_0\\
		a_1 & \bar{a}_1\\
		\end{pmatrix}+\ln\zeta\begin{pmatrix}
		\frac{a_1}{2} & \frac{\bar{a}_1}{2}\\
		0 & 0\\
		\end{pmatrix}
		+\mathcal{O}\big(\zeta\left|\ln\left|\zeta\right|\right|\big)\Bigg]
	\end{equation}
	\item As established above we have for $\zeta\rightarrow\infty$ in a whole neighborhood of infinity
	\begin{eqnarray*}
	P_{BE}^{RH}(\zeta) &=&\zeta^{-\frac{1}{4}\sigma_3}\begin{pmatrix}
	1 & 1\\
	i & -i\\
	\end{pmatrix}e^{-i\frac{\pi}{4}\sigma_3}\Bigg[I+\frac{i}{8\sqrt{\zeta}}\begin{pmatrix}
	1 & 2i\\
	2i & -1\\
	\end{pmatrix}+\frac{3}{128\zeta}\begin{pmatrix}
	1 & -4i\\
	4i & 1\\
	\end{pmatrix}\\
	&&+\mathcal{O}\left(\zeta^{-\frac{3}{2}}\right)\Bigg]e^{i\sqrt{\zeta}\sigma_3},\hspace{0.25cm}\zeta\rightarrow\infty.
	\end{eqnarray*}
\end{itemize}
Let us now move ahead in our construction of the exact parametrix to the solution of the original $S$-RHP near $\lambda=1$. We proceed in two steps. First introduce the change of variable
\begin{equation}\label{p16}
	\zeta(\lambda) = \left(s\int_1^{\lambda}\sqrt{\frac{\mu^2-a^2}{\mu^2-1}}\,\d\mu\right)^2 = s^2g^2(\lambda),\hspace{0.5cm}|\lambda-1|<r,\ \ -\pi<\textnormal{arg}\,\zeta<\pi,
\end{equation}
i.e.,
\begin{equation*}
	\sqrt{\zeta(\lambda)} = sg(\lambda).
\end{equation*}
This change of coordinates $\lambda\mapsto\zeta$ is locally conformal since
\begin{equation*}
	\zeta(\lambda) = 2s^2(1-a^2)(\lambda-1)\left(1+\frac{3+a^2}{6(1-a^2)}(\lambda-1)+\mathcal{O}\big((\lambda-1)^2\big)\right),\hspace{0.5cm}|\lambda-1|<r,
\end{equation*}
and it enables us to define a parametrix $E(\lambda)$ near $\lambda=+1$ as follows:
\begin{equation}\label{p17}
	E(\lambda) = B_{r_2}(\lambda)P_{BE}^{RH}\big(\zeta(\lambda)\big)\bigg[I-\frac{e^{-2s\kappa}}{2\pi i}\begin{pmatrix}
	-1 & 1\\
	-1 & 1\\
	\end{pmatrix}\ln\left(\frac{\lambda-1}{\lambda+1}\right)\bigg]e^{-i\sqrt{\zeta(\lambda)}\sigma_3},\hspace{0.5cm}|\lambda-1|<r,
\end{equation}
with $\zeta(\lambda)$ as in \eqref{p16} and the multiplier $B_{r_2}(\lambda)$ given by
\begin{equation}\label{mulright}
	B_{r_2}(\lambda) = M(\lambda)\frac{1}{2}e^{i\frac{\pi}{4}\sigma_3}\begin{pmatrix}
	1 & -i\\
	1 & i\\
	\end{pmatrix}\delta(\lambda)^{-\sigma_3}\left(\zeta(\lambda)\frac{\lambda-a}{\lambda-1}\right)^{\frac{1}{4}\sigma_3},
\end{equation}
where the function $\delta(\lambda)$ was introduced above as
\begin{equation*}
	\delta(\lambda) = \left(\frac{\lambda-a}{\lambda-1}\right)^{\frac{1}{4}}\rightarrow 1,\hspace{0.5cm}\lambda\rightarrow\infty.
\end{equation*} 
By construction, $B_{r_2}(\lambda)$ has no jump along the line segment $(1-r,1)$, indeed
\begin{eqnarray*}
	\big(B_{r_2}(\lambda)\big)_+ &=& M_-(\lambda)\begin{pmatrix}
	0 & 1\\
	-1 & 0\\
	\end{pmatrix}\frac{1}{2}e^{i\frac{\pi}{4}\sigma_3}\begin{pmatrix}
	1 & -i\\
	1 & i\\
	\end{pmatrix}\delta_-(\lambda)^{-\sigma_3}\begin{pmatrix}
	i & 0\\
	0 & -i\\
	\end{pmatrix}\\
	&&\times\left(\zeta(\lambda)\frac{\lambda-a}{\lambda-1}\right)^{\frac{1}{4}\sigma_3}_-= \big(B_{r_2}(\lambda)\big)_-,\hspace{0.5cm}\lambda\in(1-r,1).
\end{eqnarray*}
Hence we are left with a possible singularity at $\lambda=1$, which again is at worst of square root type and therefore is removable. This implies the analyticity of $B_{r_2}(\lambda)$ in a full neighborhood of $\lambda=1$. Next, the parametrix $E(\lambda)$ has a jump along the curve depicted in Figure \ref{figure8}, which is described by the same matrix as in the original $S$-RHP, indeed
\begin{eqnarray*}
	E_+(\lambda) &=& B_{r_2}(\lambda)\Big(P_{BE}^{RH}\big(\zeta(\lambda)\big)\Big)_-\begin{pmatrix}
	0 & 1\\
	-1 & 2\\
	\end{pmatrix}\bigg[I-\frac{e^{-2s\kappa}}{2\pi i}\begin{pmatrix}
	-1 & 1\\
	-1 & 1\\
	\end{pmatrix}\ln\left(\frac{\lambda-1}{\lambda+1}\right)\bigg]_-\\
	&&\times\begin{pmatrix}
	1-(\gamma-1) & \gamma-1\\
	-(\gamma-1) & 1+(\gamma-1)\\
	\end{pmatrix}e^{-i\sqrt{\zeta(\lambda)}_+\sigma_3}\\
	&=&E_-(\lambda)e^{i\sqrt{\zeta(\lambda)}_-\sigma_3}\begin{pmatrix}
	1-\gamma & \gamma\\
	-\gamma & \gamma+1\\
	\end{pmatrix}e^{-i\sqrt{\zeta(\lambda)}_+\sigma_3}\\
	&=&E_-(\lambda)\begin{pmatrix}
	e^{-s(2\kappa+\Pi(\lambda))} & \gamma\\
	-\gamma & (1+\gamma)e^{s\Pi(\lambda)}\\
	\end{pmatrix},\hspace{0.5cm}\textnormal{arg}\,(\lambda-1)=\pi.
\end{eqnarray*}
Furthermore, the behavior of $E(\lambda)$ at the endpoint $\lambda=+1$ matches that of $S(\lambda)$: First, compare \eqref{Psing} and the footnote \ref{foot1},
\begin{equation}\label{add1}
	P_{BE}^{RH}\big(\zeta(\lambda)\big) = \check{P}_{BE}^{RH}\big(\zeta(\lambda)\big)\bigg[I+\frac{1}{2\pi i}\begin{pmatrix}
	-1 & 1\\
	-1 & 1\\
	\end{pmatrix}\ln\left(\frac{\lambda-1}{\lambda+1}\right)\bigg],\hspace{0.5cm}\lambda\rightarrow 1
\end{equation} 
with a locally analytic function $\check{P}_{BE}^{RH}\big(\zeta(\lambda)\big)$. Hence from the analyticity of $B_{r_2}(\lambda)$ at $\lambda=1$ and \eqref{p16}, we obtain
\begin{equation}\label{add2}
	E(\lambda)e^{isg(\lambda)\sigma_3} =B_{r_2}(\lambda)\check{P}_{BE}^{RH}\big(\zeta(\lambda)\big)\bigg[I+\frac{\gamma}{2\pi i}\begin{pmatrix}
	-1 & 1\\
	-1 & 1\\
	\end{pmatrix}\ln\left(\frac{\lambda-1}{\lambda+1}\right)\bigg],\hspace{0.5cm}\lambda\rightarrow 1
\end{equation}
which agrees with \eqref{Xsing}. Thus the ratio $N_{r_1}(\lambda)$ of $S(\lambda)$ with $E(\lambda)$ is locally analytic, i.e.,
\begin{equation}\label{mana1}
	S(\lambda) = N_{r_1}(\lambda)E(\lambda),\hspace{0.5cm}|\lambda-1|<r<\frac{1}{2}(1-a).
\end{equation}
Let us now derive the matching relation between the model functions $E(\lambda)$ and $M(\lambda)$. First, note that
\begin{equation*}
	B_{r_2}(\lambda)\zeta^{-\frac{1}{4}\sigma_3}\begin{pmatrix}
	1 & 1\\
	i & -i\\
	\end{pmatrix}e^{-i\frac{\pi}{4}\sigma_3} = M(\lambda)
\end{equation*}
and therefore
\begin{eqnarray}
	E(\lambda)&=&M(\lambda)\Bigg[I+\frac{i}{8\sqrt{\zeta}}\begin{pmatrix}
	1 & 2i\\
	2i & -1\\
	\end{pmatrix}+\frac{3}{128\zeta}\begin{pmatrix}
	1 & -4i\\
	4i & 1\\
	\end{pmatrix}+\mathcal{O}\left(\zeta^{-\frac{3}{2}}\right)\Bigg]\nonumber\\
	&&\times\bigg[I-\frac{e^{-2s\kappa}}{2\pi i}\begin{pmatrix}
	-1 & e^{2isg(\lambda)}\\
	-e^{-2isg(\lambda)} & 1\\
	\end{pmatrix}\ln\left(\frac{\lambda-1}{\lambda+1}\right)\bigg]
	\label{p18}
\end{eqnarray}
as $s\rightarrow\infty,\gamma\uparrow 1$ for $\kappa\in[\delta,1-\delta],\delta>0$ and $0<r_1\leq |\lambda-1|\leq r_2<\frac{1}{2}(1-a)$ (hence $|\zeta|\rightarrow\infty$). For values of $\lambda$ chosen from a circle $|\lambda-1|=r$, $r_1<r<r_2$, we have
\begin{equation*}
	\Big|e^{-2s\kappa\pm 2isg(\lambda)}\Big| \leq e^{-2s\kappa}e^{\mp 2s\sqrt{r}\sin(\frac{1}{2}\textnormal{arg}(\lambda-1))c}
\end{equation*}
with some constant $c>0$ whose values is not important. Thus, for instance, the choice of 
\begin{equation*}
	\sqrt{r} < \frac{\kappa}{2c}
\end{equation*}
ensures that the third factor in equation \eqref{p18} is exponentially close to $I$ for $\kappa\in[\delta,1-\delta], |\lambda-1|=r$. Since the function $\zeta(\lambda)$ grows quadratically in $s$ on the latter circle and $Q(\lambda)$ is bounded there, equation \eqref{p18} yields in fact the desired matching relation between the model functions $E(\lambda)$ and $M(\lambda)$, for some $c>0$,
\begin{equation*}
	E(\lambda) = \big(I+o(1)\big)M(\lambda),\hspace{0.5cm}s\rightarrow\infty,\gamma\uparrow 1,\kappa\in[\delta,1-\delta],\delta>0,\ \ 0<r_1\leq|\lambda-1|\leq r_2<\left(\frac{\kappa}{2c}\right)^2.
\end{equation*}
\subsection{Exact parametrix at the branch point $\lambda=-1$}\label{lendp}
The construction of the remaining parametrix near $\lambda=-1$ is similar to that of the previous section. Fix a small neighborhood $\hat{\mathcal{V}}$ of radius $r<\frac{1}{2}(1-a)$ as shown in Figure \ref{figure10} below and observe that
\begin{equation*}
	\Pi(\lambda) = -2\int_{-\lambda}^1\sqrt{\frac{\mu^2-a^2}{1-\mu^2}}\,\d\mu = -d_0|1+\lambda|^{\frac{1}{2}}+\mathcal{O}\left(|1+\lambda|^{\frac{3}{2}}\right),\hspace{0.5cm}d_0 = \sqrt{8(1-a^2)}>0
\end{equation*}
for $\lambda\rightarrow -1, \lambda\in\hat{\mathcal{V}}\cap(-1,-a)$. 
\begin{figure}[tbh]
\begin{center}
\hspace{0.7cm}\includegraphics[width=0.45\textwidth]{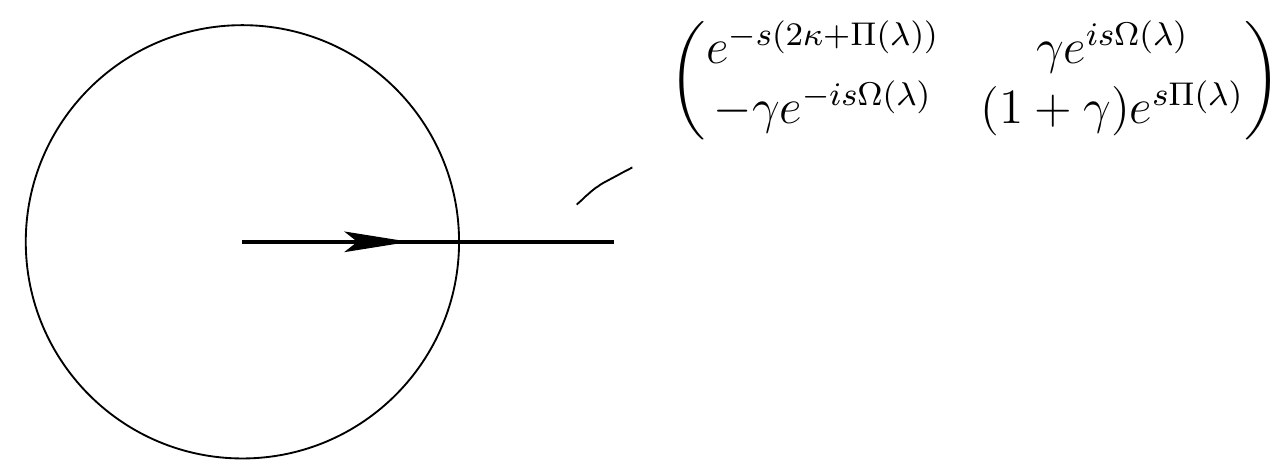}
\end{center}
\caption{A neighborhood $\hat{\mathcal{V}}$ of $\lambda=-1$}
\label{figure10}
\end{figure}

Similarly to \eqref{p15}, we let 
\begin{eqnarray}
	\tilde{P}_{BE}^{RH}(\zeta) &=& \sqrt{\frac{\pi}{2}}e^{i\pi sV\sigma_3}\begin{pmatrix}
	e^{-i\frac{\pi}{2}}\sqrt{\zeta}\big(H_{0}^{(2)}\big)'\big(e^{-i\frac{\pi}{2}}\sqrt{\zeta}\big) & e^{-i\frac{\pi}{2}}\sqrt{\zeta}\big(H_{0}^{(1)}\big)'\big(e^{-i\frac{\pi}{2}}\sqrt{\zeta}\big)\smallskip\\
	H_{0}^{(2)}\big(e^{-i\frac{\pi}{2}}\sqrt{\zeta}\big) & H_{0}^{(1)}\big(e^{-i\frac{\pi}{2}}\sqrt{\zeta}\big)\\
	\end{pmatrix}e^{-i\pi sV\sigma_3}\nonumber\\
	&=&e^{i\pi sV\sigma_3}\sigma_1P_{BE}^{RH}\big(e^{-i\pi}\zeta\big)\sigma_1 e^{-i\pi sV\sigma_3},\hspace{0.5cm}0<\textnormal{arg}\,\zeta<2\pi,\label{p19}
\end{eqnarray}
on the punctured plane $\zeta\in\mathbb{C}\backslash\{0\}$, where the branch of the root $\sqrt{\zeta}$ is chosen with $0<\textnormal{arg}\,\zeta<2\pi$. Applying similar arguments as in the previous section, we obtain that $\tilde{P}_{BE}^{RH}(\zeta)$ solves the following model problem:
\begin{figure}[tbh]
\begin{center}
\includegraphics[width=0.35\textwidth]{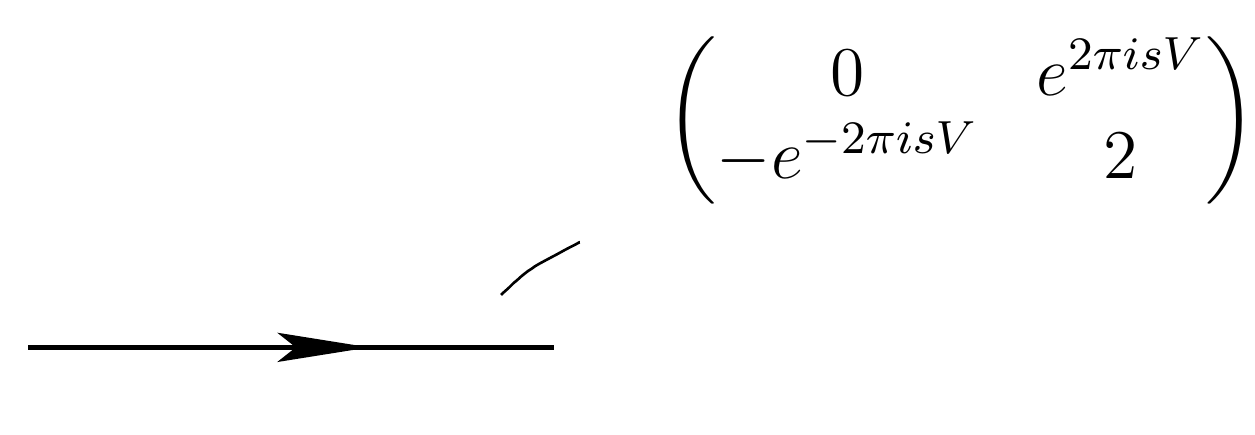}
\end{center}
\caption{A model problem near $\lambda=-1$ which can be solved explicitly using Bessel functions}
\label{figure11}
\end{figure}

\begin{itemize}
	\item $\tilde{P}_{BE}^{RH}(\zeta)$ is analytic for $\zeta\in\mathbb{C}\backslash\{\textnormal{arg}\,\zeta=0\}$.
	\item The following jump relation holds on the line $\textnormal{arg}\,\zeta=0$ (see Figure \ref{figure11}):
	\begin{eqnarray*}
		\big(\tilde{P}_{BE}^{RH}(\zeta)\big)_+ &=& \big(\tilde{P}_{BE}^{RH}(\zeta)\big)_-e^{i\pi sV\sigma_3}\begin{pmatrix}
		0 & 1\\
		-1 & 2\\
		\end{pmatrix}e^{-i\pi sV\sigma_3}\\
		&=&\big(\tilde{P}_{BE}^{RH}(\zeta)\big)_-\begin{pmatrix}
		0 & e^{2\pi isV}\\
		-e^{-2\pi isV} & 2\\
		\end{pmatrix}.
	\end{eqnarray*}
	\item For $\zeta\rightarrow 0$, with $0<\textnormal{arg}\,\zeta\leq 2\pi$,
	\begin{equation*}
		\tilde{P}_{BE}^{RH}(\zeta)=\sqrt{\frac{\pi}{2}}e^{i\pi sV\sigma_3}\bigg[\begin{pmatrix}
		\bar{a}_1 & a_1\\
		\bar{a}_0 & a_0\\
		\end{pmatrix}+\ln\left(e^{-i\pi}\zeta\right)\begin{pmatrix}
		0 & 0\\
		\frac{\bar{a}_1}{2} & \frac{a_1}{2}\\
		\end{pmatrix}+\mathcal{O}\big(\zeta\left|\ln\left|\zeta\right|\right|\big)\bigg]e^{-i\pi sV\sigma_3},
	\end{equation*}
	where $a_0,a_1$ are given in \eqref{excoeff}.
	\item As $\zeta\rightarrow\infty$, we have in a full neighborhood of infinity
	\begin{eqnarray*}
		\tilde{P}_{BE}^{RH}(\zeta)&=&\big(e^{-i\pi}\zeta\big)^{\frac{1}{4}\sigma_3}\begin{pmatrix}
		-i & ie^{2\pi isV}\\
		e^{-2\pi isV} & 1\\
		\end{pmatrix}e^{i\frac{\pi}{4}\sigma_3}\Bigg[I+\frac{ie^{i\pi sV\sigma_3}}{8\sqrt{\zeta}}\begin{pmatrix}
		-i & -2\\
		-2 & i\\
		\end{pmatrix}\\
		&&\times e^{-i\pi sV\sigma_3}+\frac{3ie^{i\pi sV\sigma_3}}{128\zeta}\begin{pmatrix}
		i & -4\\
		4 & i\\
		\end{pmatrix}e^{-i\pi sV\sigma_3}+\mathcal{O}\left(\zeta^{-\frac{3}{2}}\right)\Bigg]e^{-\sqrt{\zeta}\sigma_3}.
	\end{eqnarray*}
\end{itemize}
To finalize the construction of the parametrix near $\lambda=-1$, define
\begin{equation}\label{p20}
	\zeta(\lambda) = \left(is\int_{-1}^{\lambda}\sqrt{\frac{\mu^2-a^2}{\mu^2-1}}\,\d\mu\right)^2,\hspace{0.5cm}|\lambda+1|<r,\ \ 0<\textnormal{arg}\,\zeta< 2\pi
\end{equation}
i.e.,
\begin{equation*}
	\sqrt{\zeta(\lambda)} = -is\int_{-1}^{\lambda}\sqrt{\frac{\mu^2-a^2}{\mu^2-1}}\,\d\mu=i\pi s V-isg(\lambda).
\end{equation*}
This change of variable is locally conformal, since
\begin{equation*}
	\zeta(\lambda) = 2s^2(1-a^2)(\lambda+1)\left(1-\frac{3+a^2}{6(1-a^2)}(\lambda+1)+\mathcal{O}\big((\lambda+1)^2\big)\right),\hspace{0.5cm}|\lambda+1|<r,
\end{equation*}
and we can define the parametrix $F(\lambda)$ near $\lambda=-1$ as follows:
\begin{equation}\label{p21}
	F(\lambda) = B_{l_2}(\lambda)\tilde{P}_{BE}^{RH}\big(\zeta(\lambda)\big)e^{i\pi sV\sigma_3}\bigg[I-\frac{e^{-2s\kappa}}{2\pi i}\begin{pmatrix}
	-1 & 1\\
	-1 & 1\\
	\end{pmatrix}\ln\left(\frac{\lambda-1}{\lambda+1}\right)\bigg]e^{-i\pi sV\sigma_3}e^{\sqrt{\zeta(\lambda)}\sigma_3},\hspace{0.15cm}|\lambda+1|<r.
\end{equation}
Here $\zeta(\lambda)$ is given in \eqref{p20} and we set
\begin{equation*}
	B_{l_2}(\lambda) = M(\lambda)\frac{1}{2}e^{-i\frac{\pi}{4}\sigma_3}\begin{pmatrix}
	i & e^{2\pi isV}\\
	-ie^{-2\pi isV} & 1\\
	\end{pmatrix}\tilde{\delta}(\lambda)^{\sigma_3}\left(e^{-i\pi}\zeta(\lambda)\frac{\lambda+a}{\lambda+1}\right)^{-\frac{1}{4}\sigma_3}
\end{equation*}
with
\begin{equation*}
	\tilde{\delta}(\lambda) = \left(\frac{\lambda+a}{\lambda+1}\right)^{\frac{1}{4}}\rightarrow 1,\hspace{0.5cm}\lambda\rightarrow\infty.
\end{equation*}
Here $B_{l_2}(\lambda)$ is analytic in a full neighborhood of $\lambda=-1$, since
\begin{eqnarray*}
	\big(B_{l_2}(\lambda)\big)_+ &=& M_-(\lambda)\begin{pmatrix}
				0 & e^{2\pi isV}\\
				-e^{-2\pi isV} & 0\\
				\end{pmatrix}\frac{1}{2}e^{-i\frac{\pi}{4}\sigma_3}\begin{pmatrix}
				i & e^{2\pi isV}\\
				-ie^{-2\pi isV} & 1\\
				\end{pmatrix}\begin{pmatrix}
				i & 0\\
				0 & -i\\
				\end{pmatrix}\\
				&&\times\tilde{\delta}_-(\lambda)^{\sigma_3}\left(\zeta(\lambda)\frac{\lambda+a}{\lambda+1}\right)^{-\sigma_3/4} = \big(B_{l_2}(\lambda)\big)_-,\hspace{0.5cm}\lambda\in(-1,-1+r),
\end{eqnarray*}
and the possible singularity of $B_{l_2}(\lambda)$ at $\lambda=-1$ is at worst of square-root type, hence is removable. By construction, the parametrix $F(\lambda)$ has a jump along the curve depicted in Figure \ref{figure11} and this jump is given by
\begin{eqnarray*}
	F_+(\lambda)&=&F_-(\lambda)e^{-\sqrt{\zeta(\lambda)}_-\sigma_3}e^{i\pi sV\sigma_3}\begin{pmatrix}
	1-\gamma & \gamma\\
	-\gamma & 1+\gamma\\
	\end{pmatrix}e^{-i\pi sV\sigma_3}e^{\sqrt{\zeta(\lambda)}_+\sigma_3}\\
	&=&F_-(\lambda)\begin{pmatrix}
	e^{-s(2\kappa+\Pi(\lambda))} & \gamma e^{is\Omega(\lambda)}\\
	-\gamma e^{-is\Omega(\lambda)} & (1+\gamma)e^{s\Pi(\lambda)}\\
	\end{pmatrix},\hspace{0.5cm}\textnormal{arg}\,(\lambda+1)=0.
\end{eqnarray*}
Furthermore, by construction, the behavior of $F(\lambda)$ at the endpoint $\lambda=-1$ matches the behavior of $S(\lambda)$: We have for $\lambda\rightarrow -1$,
\begin{equation*}
	\tilde{P}_{BE}^{RH}\big(\zeta(\lambda)\big) = \check{\tilde{P}}_{BE}^{RH}\big(\zeta(\lambda)\big)\bigg[I+\frac{1}{2\pi i}\begin{pmatrix}
	-1 & 1\\
	-1 & 1\\
	\end{pmatrix}\ln\left(\frac{\lambda-1}{\lambda+1}\right)\bigg]e^{-i\pi sV\sigma_3}
\end{equation*}
with a locally analytic function $\check{\tilde{P}}_{BE}^{RH}\big(\zeta(\lambda)\big)$. Now using the analyticity of $B_{l_2}(\lambda)$, we obtain that the behavior of $F(\lambda)$ at $\lambda=-1$ matches \eqref{Xsing}:
\begin{equation*}
	F(\lambda)e^{-\sqrt{\zeta(\lambda)}\sigma_3}e^{i\pi sV\sigma_3} = B_{l_2}(\lambda)\check{\tilde{P}}_{BE}^{RH}\big(\zeta(\lambda)\big)\bigg[I+\frac{\gamma}{2\pi i}\begin{pmatrix}
	-1 & 1\\
	-1 & 1\\
	\end{pmatrix}\ln\left(\frac{\lambda-1}{\lambda+1}\right)\bigg],\hspace{0.5cm}\lambda\rightarrow -1.
\end{equation*}
Therefore the ratio $N_{l_2}(\lambda)$ of $S(\lambda)$ with $F(\lambda)$ is locally analytic, i.e.
\begin{equation}\label{mana2}
	S(\lambda) = N_{l_2}(\lambda)F(\lambda),\hspace{0.5cm}|\lambda+1|<r.
\end{equation}
To obtain a matching condition between $F(\lambda)$ and $M(\lambda)$, note that
\begin{equation*}
	B_{l_2}(\lambda)\big(e^{-i\pi}\zeta\big)^{\frac{1}{4}\sigma_3}\begin{pmatrix}
	-i & ie^{2\pi isV}\\
	e^{-2\pi isV} & 1\\
	\end{pmatrix}e^{i\frac{\pi}{4}\sigma_3} = M(\lambda)
\end{equation*}
and therefore
\begin{eqnarray}
	F(\lambda)&=&M(\lambda)\Bigg[I+\frac{ie^{i\pi sV\sigma_3}}{8\sqrt{\zeta}}\begin{pmatrix}
		-i & -2\\
		-2 & i\\
		\end{pmatrix}e^{-i\pi sV\sigma_3}+\frac{3ie^{i\pi sV\sigma_3}}{128\zeta}\begin{pmatrix}
		i & -4\\
		4 & i\\
		\end{pmatrix}e^{-i\pi sV\sigma_3}+\mathcal{O}\left(\zeta^{-\frac{3}{2}}\right)\Bigg]\nonumber\\
		&&\times\bigg[I-\frac{e^{-2s\kappa}}{2\pi i}\begin{pmatrix}
		-1 & e^{2\pi isV-2\sqrt{\zeta(\lambda)}}\\
		-e^{-2\pi isV+2\sqrt{\zeta(\lambda)}}& 1\\
		\end{pmatrix}\ln\left(\frac{\lambda-1}{\lambda+1}\right)\bigg]\label{p22}
\end{eqnarray}
for $\kappa\in[\delta,1-\delta],\delta>0$ and $0<r_1\leq |\lambda+1|\leq r_2<\frac{1}{2}(1-a)$ as $s\rightarrow\infty$ (hence $|\zeta|\rightarrow\infty$). Also here, for values of $\lambda$ on the circle $|\lambda+1|=r$, we have
\begin{equation*}
	\Big|e^{-2s\kappa\pm 2\pi isV\mp 2\sqrt{\zeta(\lambda)}}\Big|\leq e^{-2s\kappa}e^{\pm 2s\sqrt{r}\sin(\frac{1}{2}\textnormal{arg}(\lambda+1))c},\hspace{0.5cm}c>0
\end{equation*}
and therefore, with $r_1<r<r_2$ for sufficiently small $r_2$, the third factor in \eqref{p22} approaches the identity matrix in the described double scaling limit with $\kappa\in[\delta,1-\delta]$ exponentially fast in $s$. Together with the boundedness of $Q(\lambda)$ in $s$ in the latter annulus and the quadratic growth of $\zeta(\lambda)$, equation \eqref{p22} provides us therefore with the following formula:
\begin{equation*}
	F(\lambda) = \big(I+o(1)\big)M(\lambda),\hspace{0.5cm}s\rightarrow\infty,\gamma\uparrow 1,\kappa\in[\delta,1-\delta],\delta>0,\ \ 0<r_1\leq |\lambda+1|\leq r_2<\left(\frac{\kappa}{2c}\right)^2.
\end{equation*}
\subsection{Final transformation and error analysis}\label{errorana}
Set
\begin{equation}\label{it1}
	R(\lambda) = S(\lambda)\left\{
                                   \begin{array}{ll}
                                     \big(U(\lambda)\big)^{-1}, & \hbox{$|\lambda-a|<r_1$,}\smallskip \\
                                     \big(V(\lambda)\big)^{-1}, & \hbox{$|\lambda+a|<r_1$,}\smallskip \\
                                     \big(E(\lambda)\big)^{-1}, & \hbox{$|\lambda-1|<r_2$,}\smallskip \\
                                     \big(F(\lambda)\big)^{-1}, & \hbox{$|\lambda+1|<r_2$,}\smallskip \\
                                     \big(M(\lambda)\big)^{-1}, & \hbox{$|\lambda\mp a|>r_1,\,|\lambda\mp 1|>r_2$}
                                   \end{array}
                                 \right.
\end{equation}
where $0<r_1<\min\left\{\frac{a}{2},r_2\right\}$ and $0<r_2<(\frac{\kappa}{2c})^2$ remain fixed. With $C_{r_i}$ and $C_{l_i}$ denoting the clockwise oriented circles depicted in Figure \ref{figure12} below, the ratio-function $R(\lambda)$ solves the following RHP, which follows from our constructions above:
\begin{figure}[tbh]
\begin{center}
\includegraphics[width=0.85\textwidth]{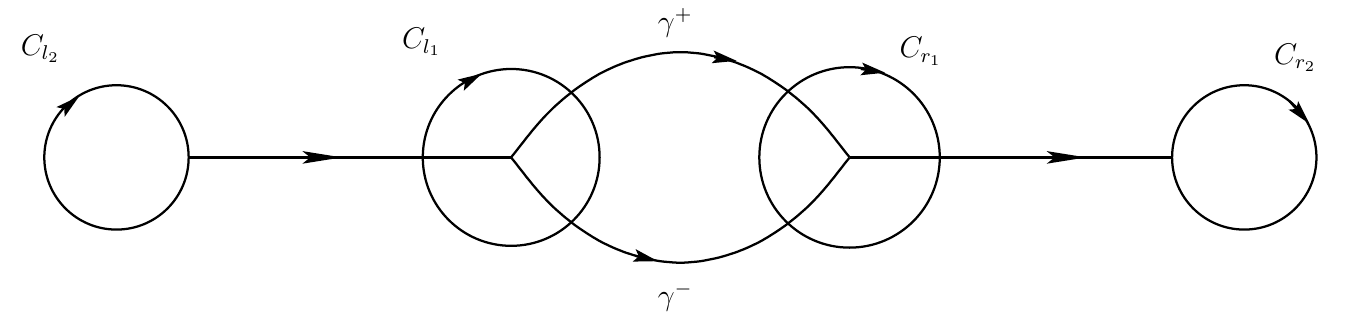}
\end{center}
\caption{The jump graph $\Sigma_R$ for the ratio-function $R(\lambda)$}
\label{figure12}
\end{figure}

\begin{itemize}
	\item $R(\lambda)$ is analytic for $\lambda\in\mathbb{C}\backslash \Sigma_{R}$
	with
	\begin{equation*}
		\Sigma_{R} = C_{r_1}\cup C_{r_2}\cup C_{l_1}\cup C_{l_2}\cup \gamma^+\cup\gamma^-\cup(-1+r_2,-a)\cup(a,1-r_2)
	\end{equation*}
	\item For the jumps, we have on the lens boundaries $\gamma^+,\gamma^-$ away from $\lambda=\pm a$,
	\begin{eqnarray*}
		R_+(\lambda)&=&R_-(\lambda)M(\lambda)\begin{pmatrix}
		1 & \gamma e^{is\Omega(\lambda)}\\
		0 & 1\\
		\end{pmatrix}\big(M(\lambda)\big)^{-1},\hspace{0.5cm}\lambda\in\gamma^+\cap\big\{|\lambda\mp a|>r_1\big\},\\
		R_+(\lambda)&=&R_-(\lambda)M(\lambda)\begin{pmatrix}
		1 & 0\\
		-\gamma e^{-is\Omega(\lambda)} & 1\\
		\end{pmatrix}\big(M(\lambda)\big)^{-1},\hspace{0.5cm}\lambda\in\gamma^-\cap\big\{|\lambda\mp a|>r_1\big\}.
	\end{eqnarray*}
	On the line segments: first, for $(-1+r_2,-a-r_1)\cup(a+r_1,1-r_2)$,
	\begin{eqnarray*}
		R_+(\lambda)&=&R_-(\lambda)M_-(\lambda)\begin{pmatrix}
		\gamma & -e^{-s(2\kappa+\Pi(\lambda)-i\Omega(\lambda))}\\
		(1+\gamma)e^{s\Pi(\lambda)-is\Omega(\lambda)} & \gamma\\
		\end{pmatrix}\big(M_-(\lambda)\big)^{-1},\\
		R_+(\lambda)&=&R_-(\lambda)M_-(\lambda)\begin{pmatrix}
		\gamma & -e^{-s(2\kappa+\Pi(\lambda))}\\
		(1+\gamma)e^{s\Pi(\lambda)} & \gamma\\
		\end{pmatrix}\big(M_-(\lambda)\big)^{-1},\hspace{0.15cm}\lambda\in(a+r_1,1-r_2).
	\end{eqnarray*}
	Next, on the circle boundaries,
	\begin{eqnarray}
		R_+(\lambda)=R_-(\lambda)U(\lambda)\big(M(\lambda)\big)^{-1},\ \ \lambda\in C_{r_1},&&R_+(\lambda)=R_-(\lambda)V(\lambda)\big(M(\lambda)\big)^{-1},\ \ \lambda\in C_{l_1},\label{circle1}\\
		R_+(\lambda)=R_-(\lambda)E(\lambda)\big(M(\lambda)\big)^{-1},\ \ \lambda\in C_{r_2},&&R_+(\lambda)=R_-(\lambda)F(\lambda)\big(M(\lambda)\big)^{-1},\ \ \lambda\in C_{l_2},\label{circle2}
	\end{eqnarray}
	and finally on the parts of the original jump contours inside $C_{l_1}$ and $C_{r_1}$: first, near $\lambda=a$,
	\begin{eqnarray}\label{Gr1}
		R_+(\lambda) &=&R_-(\lambda)U_-(\lambda)\begin{pmatrix}
		1 & -e^{-2s\kappa+is\Omega(\lambda)}\\
		0 & 1\\
		\end{pmatrix}\big(U_-(\lambda)\big)^{-1},\hspace{0.1cm} \lambda\in\gamma^+\cap\big\{|\lambda-a|<r_1\big\}\\
		R_+(\lambda) &=&R_-(\lambda)U_-(\lambda)\begin{pmatrix}
		1 & 0\\
		e^{-2s\kappa-is\Omega(\lambda)} & 1\\
		\end{pmatrix}\big(U_-(\lambda)\big)^{-1},\hspace{0.1cm}\lambda\in\gamma^-\cap\big\{|\lambda-a|<r_1\big\}\nonumber\\
		R_+(\lambda)&=&R_-(\lambda)U_-(\lambda)\begin{pmatrix}
		\gamma & -e^{-2s\kappa}e^{-s(2\kappa+\Pi(\lambda))}\\
		(1+\gamma)e^{s\Pi(\lambda)} & \gamma+(1+\gamma)e^{-2s\kappa}\\
		\end{pmatrix}\big(U_-(\lambda)\big)^{-1},\hspace{0.1cm}\lambda\in[a,a+r_1).\nonumber
	\end{eqnarray}
	And second, near $\lambda=-a$,
	\begin{eqnarray}\label{Gl1}
		R_+(\lambda) &=&R_-(\lambda)V_-(\lambda)\begin{pmatrix}
		1 & -e^{-2s\kappa+is\Omega(\lambda)}\\
		0 & 1\\
		\end{pmatrix}\big(V_-(\lambda)\big)^{-1},\hspace{0.1cm}\lambda\in\gamma^+\cap\big\{|\lambda+a|<r_1\big\}\\
		R_+(\lambda) &=&R_-(\lambda)V_-(\lambda)\begin{pmatrix}
		1 & 0\\
		e^{-2s\kappa-is\Omega(\lambda)} & 1\\
		\end{pmatrix}\big(V_-(\lambda)\big)^{-1},\hspace{0.1cm}\lambda\in\gamma^-\cap\big\{|\lambda+a|<r_1\big\}\nonumber\\
		R_+(\lambda) &=&R_-(\lambda)V_-(\lambda)\begin{pmatrix}
		\gamma & -e^{-2s\kappa}e^{-s(2\kappa+\Pi(\lambda))+is\Omega(\lambda)}\\
		(1+\gamma)e^{s\Pi(\lambda)-is\Omega(\lambda)} &\gamma+(1+\gamma)e^{-2s\kappa}\\
		\end{pmatrix}\big(V_-(\lambda)\big)^{-1}\nonumber
	\end{eqnarray}
	for $\lambda\in(-a-r_1,-a]$ in the latter jump relation.
	\item $R(\lambda)$ is analytic at $\lambda=\pm 1$. This follows from our observation that the parametrices $E(\lambda)$ and $F(\lambda)$ match the endpoint behavior of $S(\lambda)$ as stated in \eqref{Xsing}, see also \eqref{mana1} and \eqref{mana2}.
	\item As $\lambda\rightarrow\infty$, we have $R(\lambda)\rightarrow I$, valid in a full neighbhorhood of infinity.
\end{itemize}
We can solve the ratio-RHP asymptotically for $\kappa\in[\delta,1-\delta],\delta>0$, provided all its jumps are close to the identity matrix, see \cite{DZ}. By Proposition \ref{prop3}, see also our discussion at the end of section \ref{sec3}, the jump matrices corresponding to the lens boundaries $\gamma^+\cup\gamma^-$ away from the branch points $\lambda=\pm a$ are in fact exponentially close to the identity matrix
\begin{equation}\label{DZ1}
	\|M\bigl(\begin{smallmatrix}
	1 & \gamma e^{is\Omega(\cdot)}\\
	0 & 1\\
	\end{smallmatrix}\bigr)(M)^{-1}-I\|_{L^2\cap L^{\infty}(\gamma^+\cap\{|\lambda\mp a|>r_1\})}\leq c_1 e^{-c_2as},\hspace{0.5cm}\kappa\in[\delta,1-\delta],
\end{equation}
with constants $c_i>0$, and, similarly,
\begin{equation}\label{DZ2}
	\|M\bigl(\begin{smallmatrix}
	1 & 0\\
	-\gamma e^{-is\Omega(\cdot)} & 1\\
	\end{smallmatrix}\bigr)(M)^{-1}-I\|_{L^2\cap L^{\infty}(\gamma^-\cap\{|\lambda\mp a|>r_1\})}\leq c_3 e^{-c_4as},\hspace{0.5cm}\kappa\in[\delta,1-\delta].
\end{equation}
For the line segment between the circles $C_{l_i}$, i.e. for $\lambda\in(-1+r_2,-a-r_1)$ we recall \eqref{e02} and the fact that $\gamma=1-e^{-2\kappa s}$ and deduce that
\begin{equation}\label{DZ3}
	\|M_-\Bigl(\begin{smallmatrix}
	\gamma & -e^{-s(2\kappa+\Pi(\cdot))+is\Omega(\cdot)}\\
	(1+\gamma)e^{s\Pi(\cdot)-is\Omega(\cdot)} & \gamma\\
	\end{smallmatrix}\Bigr)(M_-)^{-1}-I\|_{L^2\cap L^{\infty}(-1+r_2,-a-r_1)}\leq c_5e^{-c_6a\kappa s},
\end{equation}
whereas \eqref{e01} implies that on the other line segment $(a+r_1,1-r_2)$
\begin{equation}\label{DZ4}
	\|M_-\Bigl(\begin{smallmatrix}
	\gamma & -e^{-s(2\kappa+\Pi(\cdot))}\\
	(1+\gamma)e^{s\Pi(\cdot)} & \gamma\\
	\end{smallmatrix}\Bigr)(M_-)^{-1}-I\|_{L^2\cap L^{\infty}(a+r_1,1-r_2)}\leq c_7e^{-c_8a\kappa s},\hspace{0.5cm}\kappa\in[\delta,1-\delta].
\end{equation}
The behavior of the jumps on the circle boundaries can be derived from the matching relations \eqref{p5},\eqref{p10}, \eqref{p18} and \eqref{p22}
\begin{eqnarray}
	\|UM^{-1}-I\|_{L^2\cap L^{\infty}(C_{r_1})}\leq c_9s^{-1}&&\|VM^{-1}-I\|_{L^2\cap L^{\infty}(C_{l_1})}\leq c_{10}s^{-1}\label{DZ5}\\
	\|EM^{-1}-I\|_{L^2\cap L^{\infty}(C_{r_2})}\leq c_{11}s^{-1}&&\|FM^{-1}-I\|_{L^2\cap L^{\infty}(C_{l_2})}\leq c_{12}s^{-1}\label{DZ6}
\end{eqnarray}
for $\kappa\in[\delta,1-\delta],\delta>0$. Consider now the inside of the circles $C_{l_1}$ and $C_{r_1}$ and recall that the model functions $U(\lambda)$ and $V(\lambda)$ are bounded at $\lambda=a$ and $\lambda=-a$. Hence from the previously listed jump conditions, we obtain
\begin{equation}\label{DZ61}
	\|U_-G_{r_1}\big(U_-\big)^{-1}-I\|_{L^2\cap L^{\infty}(|\lambda-a|<r_1)}\leq c_{13}e^{-c_{14}s\kappa}
\end{equation}
where the jump matrix $G_{r_1}$ of $R(\lambda)$ can be read from \eqref{Gr1}. Finally, with $G_{l_1}$ given in \eqref{Gl1},
\begin{equation}\label{DZ62}
	\|V_-G_{l_1}\big(V_-\big)^{-1}-I\|_{L^2\cap L^{\infty}(|\lambda+a|<r_1)}\leq c_{15}e^{-c_{16}s\kappa}.
\end{equation}
We can now combine the estimates \eqref{DZ1} -- \eqref{DZ62} to derive for the jump matrix $G_R(\lambda)$ of $R(\lambda)$:
\begin{equation}\label{DZ7}
	\|G_{R}-I\|_{L^2\cap L^{\infty}(\Sigma_{R})}\leq cs^{-1},\hspace{0.5cm}s\rightarrow\infty,\,\gamma\uparrow 1,\,\kappa\in[\delta,1-\delta],\delta>0.
\end{equation}
This enables us to solve the ratio problem, which is equivalent to the singular integral equation
\begin{equation*}
	R_-(\lambda) = I+\frac{1}{2\pi i}\int_{\Sigma_R}R_-(w)\big(G_R(w)-I\big)\frac{\d w}{w-\lambda_-}.
\end{equation*}
From \eqref{DZ7}, we see that the underlying integral operator is a contraction for $\kappa\in[\delta,1-\delta]$ and we can solve the latter equation iteratively in $L^2(\Sigma_R)$; its unique solution satisfies
\begin{equation}\label{DZ8}
	\|R_--I\|_{L^2(\Sigma_R)}\leq cs^{-1},\hspace{0.5cm}s\rightarrow\infty,\gamma\uparrow 1,\,\kappa\in[\delta,1-\delta],\delta>0.
\end{equation}
The solution $R(\lambda)$ at hand, we now derive the leading $s$-dependent terms in Theorem \ref{theo1}.
\section{Proof of Theorem \ref{theo1} for $\kappa\in[\delta,1-\delta],\delta>0$}\label{sec5}
\subsection{Derivation of leading terms} We use \eqref{IIKS2}
\begin{equation*}
	\frac{\partial}{\partial s}\ln\det(I-\gamma K_s)=-i\big(Y_1^{11}-Y_1^{22}\big),
\end{equation*}
where the connection to the $Y$-RHP is established through
\begin{equation*}
	Y(\lambda) = I+\frac{Y_1}{\lambda}+\mathcal{O}\big(\lambda^{-2}\big),\ \ \lambda\rightarrow\infty,\hspace{0.5cm} Y_1=(Y_1^{jk}).
\end{equation*}
First trace back the relevant transformations and obtain
\begin{eqnarray*}
	Y_1 &=& \lim_{\lambda\rightarrow\infty}\lambda\Big(Y(\lambda)-I\Big) = \lim_{\lambda\rightarrow\infty}\lambda\Big(e^{-isl\sigma_3}X(\lambda)e^{is(g(\lambda)-\lambda)\sigma_3}-I\Big)\\
	&=&\lim_{\lambda\rightarrow\infty}\lambda\Big(e^{-is\ell\sigma_3}S(\lambda)e^{is(g(\lambda)-\lambda)\sigma_3}-I\Big) = \lim_{\lambda\rightarrow\infty}\lambda\Big(e^{-is\ell\sigma_3}R(\lambda)M(\lambda)e^{is(g(\lambda)-\lambda)\sigma_3}-I\Big).
\end{eqnarray*}
Since (compare \eqref{Masy}) for $\lambda\rightarrow\infty$
\begin{equation*}
	R(\lambda) = I+\frac{i}{2\pi\lambda}\int_{\Sigma_R}R_-(w)\big(G_R(w)-I\big)\,\d w +\mathcal{O}\big(\lambda^{-2}\big),\hspace{0.5cm}M(\lambda)=I+\frac{M_1}{\lambda}+\mathcal{O}\big(\lambda^{-2}\big),
\end{equation*}
we obtain using \eqref{gfuncasy},
\begin{equation*}
	Y_1 = -\frac{is}{2}(1-a^2)\sigma_3+e^{-is\ell\sigma_3}M_1e^{is\ell\sigma_3}+\frac{i}{2\pi}e^{-is\ell\sigma_3}\int_{\Sigma_R}R_-(w)\big(G_R(w)-I\big)\d w\,e^{is\ell\sigma_3},
\end{equation*}
and therefore, using \eqref{DZ7}, we have
\begin{eqnarray}\label{sd4}
	\frac{\partial}{\partial s}\ln\det(I-\gamma K_s) &=& -s(1-a^2)-i\big(M_1^{11}-M_1^{22}\big)\\
	&&+\frac{1}{2\pi}\int_{\Sigma_R}\bigg[\big(G_R(w)-I\big)_{11}-\big(G_R(w)-I\big)_{22}\bigg]\,\d w+\mathcal{O}\big(s^{-2}\big).\nonumber
\end{eqnarray}
Let us consider the first two terms in \eqref{sd4}. Since by \eqref{Mdef} with $c$ given in \eqref{Anorm},
\begin{equation*}
	M_1^{11}=-c\big(\ln\theta(sV|\tau)\big)',\hspace{0.8cm}M_1^{22}=c\big(\ln\theta(sV|\tau)\big)'
\end{equation*}
we have
\begin{equation}\label{Meq}
	-i\big(M_1^{11}-M_1^{22}\big) = 2ic\big(\ln\theta(sV|\tau)\big)',
\end{equation}
where $(')$ always indicates differentiation with respect to the first argument of the given theta function. At this point we recall that the branchpoints $z=\pm a$ are in fact $s$-dependent by equation \eqref{Dya}. Hence we have to view all the quantities in \eqref{sd4} as functions of $s$
\begin{equation*}
	a=a(s),\ \ u_{\infty}=u_{\infty}(s),\ \ V=V(s),\ \ d=d(s).
\end{equation*}
Also, the elliptic nome $\tau=\tau(s)$ of the Jacobi theta function is $s$-dependent, hence the arguments in \eqref{Meq} have to viewed as
\begin{equation*}
	\theta(sV|\tau)=\theta\big(sV(s)| \tau(s)\big).
\end{equation*}
In this setting the following Proposition is useful.
\begin{prop}\label{prop10} There holds the following relation between the normalizing constant $c$ given in \eqref{Anorm}, the elliptic nome $\tau$ given in \eqref{Dymod} and the parameter $V$ given in \eqref{Vconst}
\begin{equation}\label{parrel}
		\pi V+i\tau\kappa = 2\pi ic.
\end{equation}
Moreover, we have the differential identity (with $v$ fixed)
\begin{equation}\label{diffrel}
	\frac{\partial}{\partial s}\big(sV(s)\big) = 2ic(s).
\end{equation}
\end{prop}
\begin{proof} Differentiating the left hand side in \eqref{diffrel} (recall that $\frac{\partial}{\partial s}$ is the partial derivative with $v=\kappa s$ fixed) gives
\begin{equation*}
	\frac{\partial}{\partial s}\left(sV(s)\right) =V+s\frac{\partial V}{\partial\kappa}\frac{\partial\kappa}{\partial s} = V-\kappa\frac{\partial V}{\partial a}\frac{\partial a}{\partial\kappa}=V+\kappa \frac{i\tau}{\pi}
\end{equation*}
where the last equality follows from \eqref{Dya}, \eqref{Vconst} and \eqref{Bper}. Let us consider the following meromorphic differential of the second kind on $\Gamma$
\begin{equation*}
	\d\widehat{\omega} = \sqrt{\frac{\lambda^2-a^2}{\lambda^2-1}}\,\d\lambda.
\end{equation*}
In terms of the local variable $\xi=\lambda^{-1}$,
\begin{equation*}
	\d\widehat{\omega}(\xi) = \mp\left(1+\frac{1}{2}(1-a^2)\xi^2+\mathcal{O}\big(\xi^4\big)\right)\frac{\d\xi}{\xi^2},\hspace{0.5cm} \omega = \mp c\left(1+\mathcal{O}\left(\xi^2\right)\right)d\xi
\end{equation*}
as $\xi\rightarrow 0$ on the first (respectively, second) sheet of $\Gamma$. Thus by Riemann's bilinear relations
\begin{equation*}
	\int_{A_1}\omega\int_{B_1}\d\widehat{\omega}-\int_{A_1}\d\widehat{\omega}\int_{B_1}\omega = 4\pi ic.
\end{equation*}
However $\int_{A_1}\omega = 1,\int_{B_1}\d\widehat{\omega}=2\pi V$ and $\int_{A_1}\d\widehat{\omega} = -2i\kappa$ and therefore
\begin{equation*}
	2\pi V+2i\kappa\tau = 4\pi ic,
\end{equation*}
which gives equation \eqref{parrel} and thus also \eqref{diffrel}.
\end{proof} 
Equation \eqref{diffrel} allows us to rewrite \eqref{Meq} as follows
\begin{equation*}
	-i\big(M_1^{11}-M_1^{22}\big) =2ic\big(\ln\theta(sV)\big)'=\frac{\partial}{\partial s}\ln\theta(sV|\tau)-\frac{\partial}{\partial\tau}\ln\theta(x|\tau)\big|_{x=sV}\frac{\partial\tau}{\partial s}
\end{equation*}
and using \eqref{de} and substituting this expression into \eqref{sd4}, we obtain
\begin{eqnarray}\label{lint0}
	\frac{\partial}{\partial s}\ln\det(I-\gamma K_s) &=& -s\big(1-a^2\big)+\frac{\partial}{\partial s}\ln\theta(sV)+\frac{i}{4\pi}\frac{1}{\theta(sV|\tau)}\frac{\partial^2}{\partial x^2}\theta(x|\tau)\big|_{x=sV}\frac{\partial\tau}{\partial s}\\
	&&+\frac{1}{2\pi}\int_{\Sigma_R}\bigg[\big(G_R(w)-I\big)_{11}-\big(G_R(w)-I\big)_{22}\bigg]\,\d w+\mathcal{O}\big(s^{-2}\big)\nonumber.
\end{eqnarray}
\subsection{Computation of $\mathcal{O}(s^{-1})$ term}\label{firstcorr}
Now, we start a somewhat involved computation of the line integral
\begin{equation}\label{lint1}
	\int_{\Sigma_R} \big(G_R(w)-I\big)\,\d w = \sum_{i=1}^2\left[\oint_{C_{r_i}}\big(G_R(w)-I\big)\,\d w+\oint_{C_{l_i}}\big(G_R(w)-I\big)\,\d w\right]+\mathcal{O}\left(e^{-c_{17}a\kappa s}\right)
\end{equation}
where the error term is the exponentially small contributions arising from the rest of the jump contour $\Sigma_R$, see \eqref{DZ1},\eqref{DZ2},\eqref{DZ3},\eqref{DZ4} and \eqref{DZ61},\eqref{DZ62}. The following computation is similar in style to the one carried out in \cite{BL}:\smallskip

On the relevant circles, we have for $\kappa\in[\delta,1-\delta]$ by \eqref{circle1},\eqref{circle2} and \eqref{p5},\eqref{p10},\eqref{p18},\eqref{p22},
\begin{eqnarray*}
	G_R(\lambda)-I&=&\frac{1}{48}\,\zeta^{-\frac{3}{2}}(\lambda)M(\lambda)\begin{pmatrix}
	-1 & 6i\\
	6i & 1\\
	\end{pmatrix}\big(M(\lambda)\big)^{-1}+\mathcal{O}\big(s^{-2}\big),\hspace{0.5cm}\lambda\in C_{r_1}\\
	G_R(\lambda)-I &=&\frac{i}{48}\,\zeta^{-\frac{3}{2}}(\lambda)M(\lambda)\begin{pmatrix}
	1 & 6i e^{2\pi isV}\\
	6i e^{-2\pi isV} & -1\\
	\end{pmatrix}\big(M(\lambda)\big)^{-1} +\mathcal{O}\big(s^{-2}\big),\hspace{0.5cm}\lambda\in C_{l_1}\\
	G_R(\lambda)-I&=&\frac{i}{8}\,\zeta^{-\frac{1}{2}}(\lambda)M(\lambda)\begin{pmatrix}
	1 & 2i\\
	2i & -1\\
	\end{pmatrix}\big(M(\lambda)\big)^{-1}+\mathcal{O}\big(s^{-2}\big),\hspace{0.5cm}\lambda\in C_{r_2}\\
	G_R(\lambda)-I&=&\frac{i}{8}\,\zeta^{-\frac{1}{2}}(\lambda)M(\lambda)\begin{pmatrix}
	-i & -2e^{2\pi isV}\\
	-2e^{-2\pi isV} & i\\
	\end{pmatrix}\big(M(\lambda)\big)^{-1}+\mathcal{O}\big(s^{-2}\big),\hspace{0.5cm}\lambda\in C_{l_2}.
\end{eqnarray*}
Next from \eqref{Mdef},
\begin{equation*}
	M(\lambda) = \frac{1}{2}\begin{pmatrix}
	\big(\beta(\lambda)+\beta^{-1}(\lambda)\big)N_{11}(\lambda) & -i\big(\beta(\lambda)-\beta^{-1}(\lambda)\big)N_{12}(\lambda)\\
	i\big(\beta(\lambda)-\beta^{-1}(\lambda)\big)N_{21}(\lambda) & \big(\beta(\lambda)+\beta^{-1}(\lambda)\big)N_{22}(\lambda)\\
	\end{pmatrix}
\end{equation*}
with
\begin{align*}
	N_{11}(\lambda)&= \frac{\theta(0)}{\theta(sV)}\frac{\theta(u(\lambda)+sV+d)}{\theta(u(\lambda)+d)}, & N_{12}(\lambda)&=\frac{\theta(0)}{\theta(sV)}\frac{\theta(-u(\lambda)+sV+d)}{\theta(-u(\lambda)+d)}\\
	N_{21}(\lambda)&=\frac{\theta(0)}{\theta(sV)}\frac{\theta(u(\lambda)+sV-d)}{\theta(u(\lambda)-d)},    & N_{22}(\lambda)&=\frac{\theta(0)}{\theta(sV)}\frac{\theta(-u(\lambda)+sV-d)}{\theta(-u(\lambda)-d)}.
\end{align*} 
Here the functions $N_{jk}(\lambda)$ are analytic in $\mathbb{C}\backslash J$, and from Proposition \ref{prop5},
\begin{align*}
	\big(N_{11}(\lambda)\big)_{\pm} &= \big(N_{12}(\lambda)\big)_{\mp}, &\big(N_{21}(\lambda)\big)_{\pm} &= \big(N_{22}(\lambda)\big)_{\mp},&\lambda\in(a,1)\\
	\big(N_{11}(\lambda)\big)_{\pm} &= \big(N_{12}(\lambda)\big)_{\mp}e^{-2\pi isV},&\big(N_{21}(\lambda)\big)_{\pm} &= \big(N_{22}(\lambda)\big)_{\mp}e^{-2\pi isV},&\lambda\in(-1,-a).
\end{align*}
Thus
\begin{eqnarray}
	\big(G_R(\lambda)-I\big)_{11} &=&\frac{1}{192}\,\zeta^{-\frac{3}{2}}(\lambda)\bigg[-\big(\beta^2+\beta^{-2}\big)\left(N_{11}N_{22}+N_{12}N_{21}\right)\label{G11}\\
	&&+6\big(\beta^2-\beta^{-2}\big)\left(N_{11}N_{21}+N_{12}N_{22}\right)-2\left(N_{11}N_{22}-N_{12}N_{21}\right)\bigg]+\mathcal{O}\big(s^{-2}\big)\nonumber\\
	&=&-\big(G_R(\lambda)-I\big)_{22}+\mathcal{O}\left(s^{-2}\right),\hspace{0.5cm}\lambda\in C_{r_1}\nonumber
\end{eqnarray}
and we notice (see again Proposition \ref{prop5}) that $N_{11}N_{22}+N_{12}N_{21}$ and $N_{11}N_{21}+N_{12}N_{22}$ have no jumps in a small neighborhood of $\lambda=a$ and are bounded at $\lambda=a$ (compare Proposition \ref{prop8}), hence are locally analytic. The function $N_{11}N_{22}-N_{12}N_{22}$, on the other hand, satisfies the relation:
\begin{equation}\label{tev1}
	\big(N_{11}N_{22}-N_{12}N_{22}\big)_+ = -\big(N_{11}N_{22}-N_{12}N_{22}\big)_-,\hspace{0.5cm}\lambda\in(a,1).
\end{equation}
Similar statements hold for the remaining cases. First, for $\lambda\in C_{l_1}$,
\begin{eqnarray*}
	\big(G_R(\lambda)-I\big)_{11} &=&\frac{i}{192}\,\zeta^{-\frac{3}{2}}(\lambda)\bigg[\big(\beta^2+\beta^{-2}\big)\left(N_{11}N_{22}+N_{12}N_{21}\right)+6\big(\beta^2-\beta^{-2}\big)\\
	&&\times\left(N_{11}N_{21}e^{2\pi isV}+N_{12}N_{22}e^{-2\pi isV}\right)+2\left(N_{11}N_{22}-N_{12}N_{21}\right)\bigg]+\mathcal{O}\big(s^{-2}\big)\\
	&=&-\big(G_R(\lambda)-I\big)_{22}+\mathcal{O}\left(s^{-2}\right),\hspace{0.5cm}\lambda\in C_{l_1}
\end{eqnarray*}
with functions $N_{11}N_{22}+N_{12}N_{21}$ and $N_{11}N_{21}e^{2\pi isV}+N_{12}N_{22}e^{-2\pi isV}$ analytic at $\lambda=-a$
and
\begin{equation}\label{tev2}
	\big(N_{11}N_{22}-N_{12}N_{21}\big)_+ = -\big(N_{11}N_{22}-N_{12}N_{21}\big)_-,\hspace{0.5cm}\lambda\in(-1,-a).
\end{equation}
Next, for $\lambda\in C_{r_2}$,
\begin{eqnarray*}
	\big(G_R(\lambda)-I\big)_{11} &=&\frac{i}{32}\,\zeta^{-\frac{1}{2}}(\lambda)\bigg[\big(\beta^2+\beta^{-2}\big)\left(N_{11}N_{22}+N_{12}N_{21}\right)\\
	&&+2\big(\beta^2-\beta^{-2}\big)\left(N_{11}N_{21}+N_{12}N_{22}\right)+2\left(N_{11}N_{22}-N_{12}N_{21}\right)\bigg]+\mathcal{O}\big(s^{-2}\big)\\
	&=&-\big(G_R(\lambda)-I\big)_{22}+\mathcal{O}\left(s^{-2}\right),\hspace{0.5cm}\lambda\in C_{r_2},
\end{eqnarray*}
where $N_{11}N_{22}+N_{12}N_{21}$ and $N_{11}N_{21}+N_{12}N_{22}$ are also analytic at $\lambda=1$, and the jump of $N_{11}N_{22}-N_{12}N_{21}$ is given in \eqref{tev1}. Finally, for $\lambda\in C_{l_2}$,
\begin{eqnarray*}
	\big(G_R(\lambda)-I\big)_{11} &=&\frac{1}{32}\,\zeta^{-\frac{1}{2}}(\lambda)\bigg[\big(\beta^2+\beta^{-2}\big)\left(N_{11}N_{22}+N_{12}N_{21}\right)-2\big(\beta^2-\beta^{-2}\big)\\
	&&\times\left(N_{11}N_{21}e^{2\pi isV}+N_{12}N_{22}e^{-2\pi isV}\right)+2\left(N_{11}N_{22}-N_{12}N_{21}\right)\bigg]+\mathcal{O}\big(s^{-2}\big)\\
	&=&-\big(G_R(\lambda)-I\big)_{22}+\mathcal{O}\left(s^{-2}\right),\hspace{0.5cm}\lambda\in C_{l_2}
\end{eqnarray*}
where $N_{11}N_{22}+N_{12}N_{21}$ and $N_{11}N_{21}e^{2\pi isV}+N_{12}N_{22}e^{-2\pi isV}$ are analytic at $\lambda=-1$ and the jump relation of $N_{11}N_{22}-N_{12}N_{21}$ is given in \eqref{tev2}. We note the following expansions (cf. Proposition \ref{prop4}):
\begin{equation*}
	u_{\pm}(\lambda) = \mp\frac{1}{2}+\frac{2c}{\sqrt{2a(1-a^2)}}\sqrt{a-\lambda}\left[1+\frac{1-5a^2}{12a(1-a^2)}(a-\lambda)+\mathcal{O}\big((a-\lambda)^2\big)\right],\,\lambda\in(a-r_1,a),
\end{equation*}
\begin{equation*}
	u_{\pm}(\lambda) = \mp\frac{1}{2}+\frac{\tau}{2}-\frac{2c}{\sqrt{2a(1-a^2)}}\sqrt{\lambda+a}\left[1+\frac{1-5a^2}{12a(1-a^2)}(\lambda+a)+\mathcal{O}\big((\lambda+a)^2\big)\right],\,\lambda\in(-a,-a+r_1),
\end{equation*}
where $c$ is given in \eqref{Anorm} and $\tau$ in \eqref{Bper}. Furthermore,
\begin{equation*}
	u(\lambda) = \frac{2c}{\sqrt{2(1-a^2)}}\sqrt{\lambda-1}\left[1-\frac{5-a^2}{12(1-a^2)}(\lambda-1)+\mathcal{O}\big((\lambda-1)^2\big)\right],\hspace{0.5cm}\lambda\in(1,1+r_2),
\end{equation*}
\begin{equation*}
	u(\lambda) = \frac{\tau}{2}-\frac{2c}{\sqrt{2(1-a^2)}}\sqrt{-1-\lambda}\left[1-\frac{5-a^2}{12(1-a^2)}(-1-\lambda)+\mathcal{O}\big((-1-\lambda)^2\big)\right], \ \ \ \ \lambda\in(-1-r_2,-1).
\end{equation*}
Note also that
\begin{equation*}
	\beta^2(\lambda)\pm \beta^{-2}(\lambda) = \sqrt{\frac{(\lambda+a)(\lambda-1)}{(\lambda+1)(\lambda-a)}}\pm\sqrt{\frac{(\lambda+1)(\lambda-a)}{(\lambda+a)(\lambda-1)}}
\end{equation*}
and, from \eqref{change1},\eqref{p6},\eqref{p16} and \eqref{p20}, with the appropriate choice of the branches,
\begin{eqnarray*}
 \zeta^{-\frac{3}{2}}(\lambda)&=&\frac{i}{s}\sqrt{\frac{1-a^2}{2a}}(a-\lambda)^{-\frac{3}{2}}\left[1+\frac{3(1+3a^2)}{20a(1-a^2)}(a-\lambda)+\mathcal{O}\big((a-\lambda)^2\big)\right],\,\lambda\in(a-r_1,a),\\
 \zeta^{-\frac{3}{2}}(\lambda)&=&\frac{1}{s}\sqrt{\frac{1-a^2}{2a}}(\lambda+a)^{-\frac{3}{2}}\left[1+\frac{3(1+3a^2)}{20a(1-a^2)}(\lambda+a)+\mathcal{O}\big((\lambda+a)^2\big)\right],\,\lambda\in(-a,-a+r_1),\\
 \zeta^{-\frac{1}{2}}(\lambda)&=&\frac{1}{s\sqrt{2(1-a^2)}}(\lambda-1)^{-\frac{1}{2}}\left[1-\frac{3+a^2}{12(1-a^2)}(\lambda-1)+\mathcal{O}\big((\lambda-1)^2\big)\right],\,\lambda\in(1,1+r_2),
\end{eqnarray*}
and
\begin{equation*}
 \zeta^{-\frac{1}{2}}(\lambda)=\frac{-i}{s\sqrt{2(1-a^2)}}(-1-\lambda)^{-\frac{1}{2}}\left[1-\frac{3+a^2}{12(1-a^2)}(-1-\lambda)+\mathcal{O}\big((-1-\lambda)^2\big)\right],\,\lambda\in(-1-r_2,-1).
\end{equation*}
It now follows that the functions $(\beta^2\pm \beta^{-2})\zeta^{-\frac{3}{2}}(\lambda)$ and $(\beta^2\pm \beta^{-2})\zeta^{-\frac{1}{2}}(\lambda)$ are locally meromorphic in a neighborhood of the relevant branch points. More specifically, near $\lambda=a$,
\begin{equation*}
	\big(\beta^2\pm \beta^{-2}\big)\zeta^{-\frac{3}{2}}(\lambda) = \frac{i}{s}\bigg[\frac{1-a}{(a-\lambda)^2}+\frac{1}{a-\lambda}\left\{\pm\frac{1+a}{2a}-\frac{1-10a-7a^2}{10a(1+a)}\right\}+\mathcal{O}(1)\bigg],\ \ \ \lambda\rightarrow a,
\end{equation*}
and near $\lambda=-a$,
\begin{equation*}
	\big(\beta^2\pm \beta^{-2}\big)\zeta^{-\frac{3}{2}}(\lambda) = \pm\frac{1}{s}\bigg[\frac{1-a}{(\lambda+a)^2}+\frac{1}{\lambda+a}\left\{\pm\frac{1+a}{2a}-\frac{1-10a-7a^2}{10a(1+a)}\right\}+\mathcal{O}(1)\bigg],\ \ \ \lambda\rightarrow-a.
\end{equation*}
Next, near $\lambda=1$,
\begin{equation*}
	\big(\beta^2\pm\beta^{-2}\big)\zeta^{-\frac{1}{2}}(\lambda) = \frac{1}{s}\bigg[\pm \frac{1}{1+a}\frac{1}{\lambda-1}+\left\{\frac{1}{2(1-a)}\pm\frac{a(3-a)}{3(1+a)(1-a^2)}\right\}+\mathcal{O}(\lambda-1)\bigg],\ \ \ \lambda\rightarrow 1,
\end{equation*}
and, near $\lambda=-1$,
\begin{equation*}
	\big(\beta^2\pm\beta^{-2}\big)\zeta^{-\frac{1}{2}}(\lambda) =-\frac{i}{s}\bigg[\frac{1}{1+a}\frac{1}{(-1-\lambda)}\pm\left\{\frac{1}{2(1-a)}\pm\frac{a(3-a)}{3(1+a)(1-a^2)}\right\}+\mathcal{O}(-1-\lambda)\bigg],\ \ \ \lambda\rightarrow-1.
\end{equation*}
As we noticed above, certain combinations of the functions $N_{jk}(\lambda)$ are in fact analytic in a neighborhood of the relevant branch points. Indeed, at $\lambda=a$, by the expansion of $u(\lambda)$ near $a$ given above
\begin{eqnarray*}
	N_{11}(\lambda)N_{22}(\lambda)+N_{12}(\lambda)N_{21}(\lambda) &=& \Xi_0+\frac{c^2}{a(1-a^2)}\frac{\d^2}{\d u^2}\big(N_{11}N_{22}+N_{12}N_{21}\big)\Big|_{\lambda=a}(a-\lambda)\\ &&\ \ \ +\mathcal{O}\big((a-\lambda)^2\big),\ \ \ \lambda\rightarrow a,\\
	N_{11}(\lambda)N_{21}(\lambda)+N_{12}(\lambda)N_{22}(\lambda) &=& \Xi_0+\frac{c^2}{a(1-a^2)}\frac{\d^2}{\d u^2}\big(N_{11}N_{21}+N_{12}N_{22}\big)\Big|_{\lambda=a}(a-\lambda)\\ &&\ \ \ +\mathcal{O}\big((a-\lambda)^2\big),\ \ \ \lambda\rightarrow a,
\end{eqnarray*}
with the abbreviation
\begin{equation}\label{abbrev1}
	\Xi_k =\Xi_k(sV,\kappa)= 2\frac{\theta_3^2(0)}{\theta_3^2(sV)}\frac{\theta_k(sV+d)\theta_k(sV-d)}{\theta_k^2(d)},\hspace{0.5cm}k=0,1,2,3,
\end{equation}
involving $V$ given by \eqref{Vconst}, $d$ given by \eqref{dchoice},\eqref{Bper},\eqref{Anorm}, and the Jacobi theta functions $\theta_0,\theta_1,\theta_2$ and $\theta_3$ (see Appendix \ref{thetaid}). Next, at $\lambda=-a$
\begin{eqnarray*}
	N_{11}(\lambda)N_{22}(\lambda)+N_{12}(\lambda)N_{21}(\lambda)&=&-\Xi_1+\frac{c^2}{a(1-a^2)}\frac{\d^2}{\d u^2}\big(N_{11}N_{22}+N_{12}N_{21}\big)\Big|_{\lambda=-a}(\lambda+a)\\ &&\ \ \ +\mathcal{O}\big((\lambda+a)^2\big),
\end{eqnarray*}
and at $\lambda=-a$,
\begin{align*}
	N_{11}(\lambda)N_{21}(\lambda)e^{2\pi isV}+N_{12}(\lambda)N_{22}(\lambda)e^{-2\pi isV}&=-\Xi_1+\frac{c^2}{a(1-a^2)}\frac{\d^2}{\d u^2}\big(N_{11}N_{21}e^{2\pi isV}\\
	&+N_{12}N_{22}e^{-2\pi isV}\big)\Big|_{\lambda=-a}(\lambda+a)+\mathcal{O}\big((\lambda+a)^2\big).
\end{align*}
At $\lambda=1$,
\begin{eqnarray*}
	N_{11}(\lambda)N_{22}(\lambda)+N_{12}(\lambda)N_{21}(\lambda) &=&\Xi_3+\frac{c^2}{1-a^2}\frac{\d^2}{\d u^2}\big(N_{11}N_{22}+N_{12}N_{21}\big)\Big|_{\lambda=1}(\lambda-1)\\
	&&\ \ \ +\mathcal{O}\big((\lambda-1)^2\big),\\
	N_{11}(\lambda)N_{21}(\lambda)+N_{12}(\lambda)N_{22}(\lambda)&=&\Xi_3+\frac{c^2}{1-a^2}\frac{\d^2}{\d u^2}\big(N_{11}N_{21}+N_{12}N_{22}\big)\Big|_{\lambda=1}(\lambda-1)\\
	&&\ \ \ +\mathcal{O}\big((\lambda-1)^2\big),
\end{eqnarray*}
and finally at $\lambda=-1$,
\begin{eqnarray*}
	N_{11}(\lambda)N_{22}(\lambda)+N_{12}(\lambda)N_{21}(\lambda)&=&\Xi_2+\frac{c^2}{1-a^2}\frac{\d^2}{\d u^2}\big(N_{11}N_{22}+N_{12}N_{21}\big)\Big|_{\lambda=-1}(-1-\lambda)\\
	&&\ \ \ +\mathcal{O}\big((-1-\lambda)^2\big),
\end{eqnarray*}
\begin{align*}
	N_{11}(\lambda)N_{21}(\lambda)e^{2\pi isV}+N_{12}(\lambda)N_{22}(\lambda)e^{-2\pi isV} &=\Xi_2+\frac{c^2}{1-a^2}\frac{\d^2}{\d u^2}\big(N_{11}N_{21}e^{2\pi isV}\\
	&+N_{12}N_{22}e^{-2\pi isV}\big)\Big|_{\lambda=-1}(-1-\lambda)+\mathcal{O}\big((-1-\lambda)^2\big).
\end{align*}
Since the function $N_{11}N_{22}-N_{12}N_{21}$ is an odd function of $u$, we can use the previously derived expansions of $u(\lambda)$ near the branch points, and obtain
\begin{align*}
	N_{11}(\lambda)N_{22}(\lambda)&-N_{12}(\lambda)N_{21}(\lambda)=\Xi_0\left[\frac{\theta_0'(sV+d)}{\theta_0(sV+d)}-\frac{\theta_0'(sV-d)}{\theta_0(sV-d)}-2\frac{\theta_0'(d)}{\theta_0(d)}\right]\\
	&\times\frac{2c\sqrt{a-\lambda}}{\sqrt{2a(1-a^2)}}\Big(1+\mathcal{O}\big(a-\lambda\big)\Big),\hspace{0.5cm}\lambda\in(a-r_1,a),
\end{align*}
\begin{align*}
	N_{11}(\lambda)N_{22}(\lambda)&-N_{12}(\lambda)N_{21}(\lambda)=\Xi_1\left[\frac{\theta_1'(sV+d)}{\theta_1(sV+d)}-\frac{\theta_1'(sV-d)}{\theta_1(sV-d)}-2\frac{\theta_1'(d)}{\theta_1(d)}\right]\\
	&\times\frac{2c\sqrt{\lambda+a}}{\sqrt{2a(1-a^2)}}\Big(1+\mathcal{O}\big(\lambda+a\big)\Big),\hspace{0.5cm}\lambda\in(-a,-a+r_1),
\end{align*}
\begin{align*}
	N_{11}(\lambda)N_{22}(\lambda)&-N_{12}(\lambda)N_{21}(\lambda)=\Xi_3\left[\frac{\theta_3'(sV+d)}{\theta_3(sV+d)}-\frac{\theta_3'(sV-d)}{\theta_3(sV-d)}-2\frac{\theta_3'(d)}{\theta_3(d)}\right]\\
	&\times\frac{2c\sqrt{\lambda-1}}{\sqrt{2(1-a^2)}}\Big(1+\mathcal{O}\big(\lambda-1)\Big),\hspace{0.5cm}\lambda\in(1,1+r_2),
\end{align*}
\begin{align*}
	N_{11}(\lambda)N_{22}(\lambda)&-N_{12}(\lambda)N_{21}(\lambda)=-\Xi_2\left[\frac{\theta_2'(sV+d)}{\theta_2(sV+d)}-\frac{\theta_2'(sV-d)}{\theta_2(sV-d)}-2\frac{\theta_2'(d)}{\theta_2(d)}\right]\\
	&\times\frac{2c\sqrt{-1-\lambda}}{\sqrt{2(1-a^2)}}\Big(1+\mathcal{O}\big(-1-\lambda\big)\Big),\hspace{0.5cm}\lambda\in(-1-r_2,-1).
\end{align*}
Substituting these expansions into \eqref{G11}, we obtain
\begin{align*}
	\big(G_R(\lambda)&-I\big)_{11} =\frac{i}{192s}\Bigg[\frac{5(1-a)\Xi_0}{(a-\lambda)^2}+\frac{1}{a-\lambda}\Bigg\{\frac{6c^2}{a(1+a)}\frac{\d^2}{\d u^2}\big(N_{11}N_{21}+N_{12}N_{22}\big)\Big|_{\lambda=a}\\
	&-\frac{c^2}{a(1+a)}\frac{\d^2}{\d u^2}\big(N_{11}N_{22}+N_{12}N_{21}\big)\Big|_{\lambda=a}-\frac{2c\,\Xi_0}{a}
	\left(\frac{\theta_0'(sV+d)}{\theta_0(sV+d)}-\frac{\theta_0'(sV-d)}{\theta_0(sV-d)}-2\frac{\theta_0'(d)}{\theta_0(d)}\right)\\
	&-\frac{2(2+a)}{a(1+a)}\Xi_0\Bigg\}+\mathcal{O}(1)\Bigg]+\mathcal{O}\left(s^{-2}\right),\hspace{0.5cm}\lambda\in C_{r_1},\ \ s\rightarrow\infty.
\end{align*}
Similarly,
\begin{align*}
	\big(G_R(\lambda)&-I\big)_{11} = \frac{i}{192s}\Bigg[\frac{5(1-a)\Xi_1}{(\lambda+a)^2}-\frac{1}{\lambda+a}\Bigg\{\frac{6c^2}{a(1+a)}\frac{\d^2}{\d u^2}\big(N_{11}N_{21}e^{2\pi isV}+N_{12}N_{22}e^{-2\pi isV}\big)\Big|_{\lambda=-a}\\
	&-\frac{c^2}{a(1+a)}\frac{\d^2}{\d u^2}\big(N_{11}N_{22}+N_{12}N_{21}\big)\Big|_{\lambda=-a}-\frac{2c\,\Xi_1}{a}\left(\frac{\theta_1'(sV+d)}{\theta_1(sV+d)}
	-\frac{\theta_1'(sV-d)}{\theta_1(sV-d)}-2\frac{\theta_1'(d)}{\theta_1(d)}\right)\\
	&+\frac{2(2+a)}{a(1+a)}\Xi_1\Bigg\}+\mathcal{O}(1)\Bigg]+\mathcal{O}\left(s^{-2}\right),\hspace{0.5cm}\lambda\in C_{l_1},
\end{align*}
and
\begin{eqnarray*}
	\big(G_R(\lambda)-I\big)_{11} &=& \frac{i}{32s}\Bigg[-\frac{\Xi_3}{(1+a)(\lambda-1)}+\mathcal{O}(1)\Bigg]+\mathcal{O}\left(s^{-2}\right),\hspace{0.5cm}\lambda\in C_{r_2},\\
	\big(G_R(\lambda)-I\big)_{11} &=& \frac{i}{32s}\Bigg[\frac{\Xi_2}{(1+a)(-1-\lambda)}+\mathcal{O}(1)\Bigg]+\mathcal{O}\left(s^{-2}\right),\hspace{0.5cm}\lambda\in C_{l_2}.
\end{eqnarray*}
We now compute the terms involving the second order derivatives. First, at $\lambda=a$,
\begin{eqnarray*}
	\frac{\d^2}{\d u^2}\big(N_{11}N_{21}+N_{12}N_{22}\big)\Big|_{\lambda=a}&=&\Xi_0\Bigg[\frac{\theta_0''(sV+d)}{\theta_0(sV+d)}-2\left(\frac{\theta_0'(d)}{\theta_0(d)}\right)'
	+\frac{\theta_0''(sV-d)}{\theta_0(sV-d)}\\
	&&+2\frac{\theta_0'(sV-d)}{\theta_0(sV-d)}\frac{\theta_0'(sV+d)}{\theta_0(sV+d)}\Bigg],
\end{eqnarray*}
then
\begin{align*}
	\frac{\d^2}{\d u^2}\big(N_{11}N_{22}&+N_{12}N_{21}\big)\Big|_{\lambda=a}=\Xi_0\Bigg[\frac{\theta_0''(sV+d)}{\theta_0(sV+d)}-2\left(\frac{\theta_0'(d)}{\theta_0(d)}\right)'
	+\frac{\theta_0''(sV-d)}{\theta_0(sV-d)}\\
	&-2\frac{\theta_0'(sV-d)}{\theta_0(sV-d)}\frac{\theta_0'(sV+d)}{\theta_0(sV+d)}+4\bigg\{\frac{\theta_0'(sV-d)}{\theta_0(sV-d)}+\frac{\theta_0'(d)}{\theta_0(d)}
	-\frac{\theta_0'(sV+d)}{\theta_0(sV+d)}\bigg\}\frac{\theta_0'(d)}{\theta_0(d)}\Bigg],
\end{align*}
and the second derivatives near $\lambda=1$ are, up to the subscript replacement $0\mapsto 3$, identical to the latter. At $\lambda=-a$
\begin{eqnarray*}
	\frac{\d^2}{\d u^2}\big(N_{11}N_{21}e^{2\pi isV}+N_{12}N_{22}e^{-2\pi isV}\big)\Big|_{\lambda=-a}&=&-\Xi_1\Bigg[\frac{\theta_1''(sV+d)}{\theta_1(sV+d)}-2\left(\frac{\theta_1'(d)}{\theta_1(d)}\right)'+\frac{\theta_1''(sV-d)}{\theta_1(sV-d)}\\
	&&+2\frac{\theta_1'(sV-d)}{\theta_1(sV-d)}\frac{\theta_1'(sV+d)}{\theta_1(sV+d)}\Bigg],
\end{eqnarray*}
and 
\begin{align*}
	\frac{\d^2}{\d u^2}\big(N_{11}N_{22}&+N_{12}N_{21}\big)\Big|_{\lambda=-a}=-\Xi_1\Bigg[\frac{\theta_1''(sV+d)}{\theta_1(sV+d)}-2\left(\frac{\theta_1'(d)}{\theta_1(d)}\right)'
	+\frac{\theta_1''(sV-d)}{\theta_1(sV-d)}\\
	&-2\frac{\theta_1'(sV-d)}{\theta_1(sV-d)}\frac{\theta_1'(sV+d)}{\theta_1(sV+d)}+4\bigg\{\frac{\theta_1'(sV-d)}{\theta_1(sV-d)}+\frac{\theta_1'(d)}{\theta_1(d)}
	-\frac{\theta_1'(sV+d)}{\theta_1(sV+d)}\bigg\}\frac{\theta_1'(d)}{\theta_1(d)}\Bigg],
\end{align*}
furthermore, at $\lambda=-1$,
\begin{eqnarray*}
	\frac{\d^2}{\d u^2}\big(N_{11}N_{21}e^{2\pi isV}+N_{12}N_{22}e^{-2\pi isV}\big)\Big|_{\lambda=-1}&=&\Xi_2\Bigg[\frac{\theta_2''(sV+d)}{\theta_2(sV+d)}-2\left(\frac{\theta_2'(d)}{\theta_2(d)}\right)'+\frac{\theta_2''(sV-d)}{\theta_2(sV-d)}\\
	&&+2\frac{\theta_2'(sV-d)}{\theta_2(sV-d)}\frac{\theta_2'(sV+d)}{\theta_2(sV+d)}\Bigg],
\end{eqnarray*}
and
\begin{align*}
	\frac{\d^2}{\d u^2}\big(N_{11}N_{22}&+N_{12}N_{21}\big)\Big|_{\lambda=-1}=\Xi_2\Bigg[\frac{\theta_1''(sV+d)}{\theta_1(sV+d)}-2\left(\frac{\theta_1'(d)}{\theta_1(d)}\right)'
	+\frac{\theta_1''(sV-d)}{\theta_1(sV-d)}\\
	&-2\frac{\theta_1'(sV-d)}{\theta_1(sV-d)}\frac{\theta_1'(sV+d)}{\theta_1(sV+d)}+4\bigg\{\frac{\theta_1'(sV-d)}{\theta_1(sV-d)}+\frac{\theta_1'(d)}{\theta_1(d)}
	-\frac{\theta_1'(sV+d)}{\theta_1(sV+d)}\bigg\}\frac{\theta_1'(d)}{\theta_1(d)}\Bigg].
\end{align*}
In order to evaluate the integrals in \eqref{lint1}, we apply the residue theorem and obtain:
\begin{equation*}
	\oint_{C_{r_1}}\!\!\!\!\big(G_R(w)-I\big)_{11}\,\d w =-\frac{2\pi\,\Xi_0}{192s}\frac{\Theta_0}{a(1+a)}+\mathcal{O}\left(s^{-2}\right),\  \oint_{C_{l_1}}\!\!\!\big(G_R(w)-I\big)_{11}\,\d w =\frac{2\pi\,\Xi_1}{192s}\frac{\Theta_1}{a(1+a)}+\mathcal{O}\left(s^{-2}\right)
\end{equation*}
where we set
\begin{align}
	\Theta_0=\Theta_0(sV,\kappa)&=5c^2\bigg\{\frac{\theta_0''(sV+d)}{\theta_0(sV+d)}-2\left(\frac{\theta_0'(d)}{\theta_0(d)}\right)'+\frac{\theta_0''(sV-d)}{\theta_0(sV-d)}\bigg\}
+14c^2\frac{\theta_0'(sV-d)}{\theta_0(sV-d)}\frac{\theta_0'(sV+d)}{\theta_0(sV+d)}\nonumber\\&-4c^2\left\{\frac{\theta_0'(sV-d)}{\theta_0(sV-d)}
	+\frac{\theta_0'(d)}{\theta_0(d)}-\frac{\theta_0'(sV+d)}{\theta_0(sV+d)}\right\}\frac{\theta_0'(d)}{\theta_0(d)}
	-2c(1+a)\bigg\{\frac{\theta_0'(sV+d)}{\theta_0(sV+d)}-2\frac{\theta_0'(d)}{\theta_0(d)}\nonumber\\
	&-\frac{\theta_0'(sV-d)}{\theta_0(sV-d)}\bigg\}-2(2+a)\label{ell1}
\end{align}
and
\begin{align}
	\Theta_1=\Theta_1(sV,\kappa)&=5c^2\bigg\{\frac{\theta_1''(sV+d)}{\theta_1(sV+d)}-2\left(\frac{\theta_1'(d)}{\theta_1(d)}\right)'+\frac{\theta_1''(sV-d)}{\theta_1(sV-d)}\bigg\}
	+14c^2\frac{\theta_1'(sV-d)}{\theta_1(sV-d)}\frac{\theta_1'(sV+d)}{\theta_1(sV+d)}\nonumber\\
&-4c^2\left\{\frac{\theta_1'(sV-d)}{\theta_1(sV-d)}
	+\frac{\theta_1'(d)}{\theta_1(d)}-\frac{\theta_1'(sV+d)}{\theta_1(sV+d)}\right\}\frac{\theta_1'(d)}{\theta_1(d)}
	+2c(1+a)\bigg\{\frac{\theta_1'(sV+d)}{\theta_1(sV+d)}-2\frac{\theta_1'(d)}{\theta_1(d)}\nonumber\\
	&-\frac{\theta_1'(sV-d)}{\theta_1(sV-d)}\bigg\}-2(2+a).\label{ell2}
\end{align}
Furthermore,
\begin{equation*}
	\oint_{C_{r_2}}\big(G_R(w)-I\big)_{11}\,\d w = -\frac{2\pi}{32s}\frac{\Xi_3}{1+a}+\mathcal{O}\left(s^{-2}\right),\hspace{0.35cm}\oint_{C_{l_2}}\big(G_R(w)-I\big)_{11}\,\d w =- \frac{2\pi}{32s}\frac{\Xi_2}{1+a}+\mathcal{O}\left(s^{-2}\right),
\end{equation*}
and substituting these results into \eqref{lint0}, we obtain 
\begin{eqnarray}
	\frac{\partial}{\partial s}\ln\det\left(I-\gamma K_s\right)&=&-s\big(1-a^2\big)+\frac{\partial}{\partial s}\ln\theta(sV)-\frac{i}{4\pi s}\frac{\kappa}{\theta(sV|\tau)}\frac{\partial^2}{\partial x^2}\theta(x|\tau)\big|_{x=sV}\frac{\partial\tau}{\partial\kappa}\nonumber\\
	&&\hspace{0.4cm}-\frac{\Xi_0\Theta_0-\Xi_1\Theta_1+6a(\Xi_3+\Xi_2)}{96sa(1+a)}+\mathcal{O}\left(s^{-2}\right),\label{exp1}
\end{eqnarray}
uniformly for $\kappa\in[\delta,1-\delta],\delta>0$.
In order to simplify the numerator in the $\mathcal{O}\left(s^{-1}\right)$ contribution in \eqref{exp1}, we now use some identities for the theta functions, see Appendix \ref{thetaid}. First, using the connection formula \eqref{conn} between $\theta_0(z)$ and $\theta_1(z)$ and its derivative, and recalling that $d=-\frac{\tau}{4}$, we obtain (away from zeros of denominators)
\begin{equation}\label{Ig1}
	\frac{\theta_1'(z+d)}{\theta_1(z+d)}=\frac{\theta_0'(z-d)}{\theta_0(z-d)}+i\pi.
\end{equation}
Using the second derivative of \eqref{conn}, we similarly obtain (provided $z$ is not a zero of a denominator)
\begin{equation}\label{Ig2}
	\frac{\theta_1''(z+d)}{\theta_1(z+d)}=\frac{\theta_0''(z-d)}{\theta_0(z-d)}+2\pi i\frac{\theta_0'(z-d)}{\theta_0(z-d)}-\pi^2.
\end{equation}
Substituting \eqref{Ig1} and \eqref{Ig2} with $z=sV,z=-sV,z=0$ into \eqref{ell2} and using the fact that $\theta_0(z)$ is even and $\theta_1(z)$ is odd, we obtain that
\begin{equation*}
	\Theta_1=\Theta_0.
\end{equation*}
Moreover, it easily follows from \eqref{conn} that
\begin{equation*}
	\Xi_1=-\Xi_0,\hspace{0.5cm}\Xi_2=\Xi_3.
\end{equation*}
This implies that \eqref{exp1} reduces to
\begin{eqnarray}
	\frac{\partial}{\partial s}\ln\det\left(I-\gamma K_s\right) &=& -s(1-a^2)+\frac{\partial}{\partial s}\ln\theta(sV)-\frac{i}{4\pi s}\frac{\kappa}{\theta(sV|\tau)}\frac{\partial^2}{\partial x^2}\theta(x|\tau)\big|_{x=sV}\frac{\partial\tau}{\partial\kappa}\nonumber \\
	&&\hspace{0.4cm}-\frac{\Xi_0\Theta_0+6a\Xi_2}{48sa(1+a)}+\mathcal{O}\left(s^{-2}\right),\label{exf1}
\end{eqnarray}
where the error term is uniform for $\kappa\in[\delta,1-\delta],\delta>0$.\smallskip

As can be seen from the definitions \eqref{abbrev1},\eqref{ell1},\eqref{ell2},\eqref{Vconst} $\Xi_j=\Xi_j(sV,\kappa)$ and $\Theta_j=\Theta_j(sV,\kappa), V=V(\kappa)$. The following Proposition will be useful:
\begin{prop}\label{integrab} The function
\begin{equation*}
	M(x,\kappa) = \frac{\Xi_0(x,\kappa)\Theta_0(x,\kappa)+6a(\kappa)\Xi_2(x,\kappa)}{48a(\kappa)(1+a(\kappa))}+\frac{i}{4\pi}\frac{\kappa}{\theta(x|\tau)}\frac{\partial^2}{\partial y^2}\theta(y|\tau)\big|_{y=x}\frac{\partial\tau}{\partial\kappa},
\end{equation*}
with $x\in\mathbb{R},\kappa=\frac{v}{s}\in[0,1-\delta]$, where $\Xi_k$ are given by \eqref{abbrev1}, $\Theta_0$ by \eqref{ell1}, and $a(\kappa)$ by \eqref{Dya}, is periodic in the $x$-variable, namely
\begin{equation*}
	M(x+n,\kappa) = M(x,\kappa),\ \ \ x\in\mathbb{R},\ \ n\in\mathbb{Z}.
\end{equation*}
Moreover, for $s\geq s_0$ chosen so that $\kappa\in[0,1-\delta]$,
\begin{equation}\label{integrab111}
	\left|\int_s^{\infty}M(tV\big(vt^{-1}\big),vt^{-1})\frac{\d t}{t}\right|\leq c_0(\delta)<\infty,
\end{equation}
where $c_0(\delta)$ does not depend on $s,v$. In the limit $\kappa\downarrow 0$,
\begin{equation}\label{integrab112}
	\int_s^{\infty}M(tV\big(vt^{-1}\big),vt^{-1})\frac{\d t}{t}=\mathcal{O}\left(\kappa\right),\ \ \ \ \kappa\downarrow 0.
\end{equation}
\end{prop}
\begin{proof}
The periodicity is a consequence of the periodicity of the theta functions. To verify \eqref{integrab111}, \eqref{integrab112} first change the variable
\begin{equation}\label{ii1}
	\int_s^RM(tV\big(vt^{-1}\big),vt^{-1})\frac{\d t}{t} = \int_{\frac{v}{R}}^{\kappa}M\left(\frac{v}{u}V(u),u\right)\frac{\d u}{u}
\end{equation}
and use Corollary \ref{cor1} and the definitions of $\theta$-functions. This gives, as $u\downarrow 0$, uniformly in $v\geq v_0$
\begin{equation}\label{est0}
	M\left(\frac{v}{u}V(u),u\right) = \mathcal{O}\left(u\right)
\end{equation}
(since $\tau = -\frac{2i}{\pi}\ln\frac{u}{4\pi}+o(1)$ for $u\downarrow 0$) and therefore in the same limit
\begin{equation}\label{est}
	\theta\left(\frac{v}{u}V(u)\right) = 1+\mathcal{O}\left(u^2\right),\  \ \Xi_0=2\big(1+\mathcal{O}\left(u\right)\big),\ \ \Xi_2=2\big(1+\mathcal{O}\left(u\right)\big),\ \ \ \Theta_0=-6\big(1+\mathcal{O}\left(u\right)\big).
\end{equation}
Letting $R\rightarrow\infty$ in \eqref{ii1}, we see that the original integral converges at infinity. Moreover, in the limit $\kappa\downarrow 0$, we obtain the estimate \eqref{integrab112}. To prove the estimate \eqref{integrab111}, we write for $0<\epsilon<\kappa$
\begin{align}
	\left|\int_s^{\infty}M(tV\big(vt^{-1}\big),vt^{-1})\frac{\d t}{t}\right|&=\left|\int_0^{\kappa}M\left(\frac{v}{u}V(u),u\right)\frac{\d u}{u}\right|\nonumber\\
	&\leq \left|\int_0^{\epsilon}M\left(\frac{v}{u}V(u),u\right)\frac{\d u}{u}\right|+\left|\int_{\epsilon}^{\kappa}M\left(\frac{v}{u}V(u),u\right)\frac{\d u}{u}\right|.\label{b1}
\end{align} 
The first integral in the r.h.s. is bounded by a constant uniformly in $v\geq v_0$ by \eqref{est0}. We now estimate the second integral for $0<\kappa\leq 1-\delta<1$. Note that in this region, $a=a(\kappa)$ is bounded away from zero. To obtain a bound for $M$, we need to provide upper bounds for the $\theta$-functions and their derivatives and a lower bound for $|\theta_3(sV|\tau)|$ on the interval of integration. We have
\begin{equation*}
	\big|\theta_3(sV|\tau)\big|\leq\sum_{k=-\infty}^{\infty}e^{i\pi\tau k^2}\leq 1+2\sum_{k=1}^{\infty}e^{i\pi\tau k}=\frac{1+e^{i\pi\tau}}{1-e^{i\pi\tau}},
\end{equation*}
which is bounded by some $\tilde{c}_0(\delta)$ for $0<\kappa\leq 1-\delta$ (cf. Corollary \ref{cor1}). Similar estimates hold for other $\theta$-functions and their derivatives.\smallskip

Furthermore, using the Poisson summation formula, we obtain
\begin{equation*}
	\theta_3(sV|\tau)=\frac{1}{\sqrt{-i\tau}}\sum_{n=-\infty}^{\infty}e^{\frac{\pi}{i\tau}(sV-n)^2}\geq \frac{1}{\sqrt{-i\tau}}e^{\frac{\pi}{4i\tau}},
\end{equation*}
which is bounded away from zero by some $\hat{c}_0(\epsilon)$ (since $0<c_0'(\delta)\leq-i\tau\leq c_0''(\epsilon)<\infty$) for $1-\delta\geq\kappa\geq\epsilon$.\smallskip

Collecting our estimates together, we obtain that, uniformly for $v\geq v_0,s\geq s_0$,
\begin{equation*}
	\left|M\left(\frac{v}{u}V(u),u\right)\right|\leq c_1(\epsilon,\delta)<\infty
\end{equation*}
for $\epsilon\leq u\leq\kappa\leq 1-\delta$. Using this estimate in the second integral in the r.h.s. of \eqref{b1}, we finally prove \eqref{integrab111}.
\end{proof}

\subsection{Integration of the logarithmic $s$-derivative}\label{A:int}
For $\kappa\in[0,1-\delta]$, we can integrate the first term in \eqref{exf1} as follows
\begin{equation*}
	\int s\left(1-a^2(s)\right)\d s = \frac{s^2}{2}(1-a^2) +\int s^2a(s)\frac{\partial a}{\partial s}\,\d s.
\end{equation*}
However by \eqref{Dya},\eqref{Anorm} we have for $v=\kappa s$ fixed, 
\begin{equation*}
	\frac{\partial a}{\partial s} =\frac{\partial a}{\partial\kappa}\frac{\partial\kappa}{\partial s}= -\frac{2ic}{a}\frac{\kappa}{s},
\end{equation*}
and recalling \eqref{diffrel}, we obtain
\begin{equation*}
	-\int s\left(1-a^2(s)\right)\d s = -\frac{s^2}{2}\big(1-a^2\big)+vsV. 
\end{equation*}
In view of \eqref{exf1} and Proposition \ref{integrab} we define
\begin{equation}\label{Dsv}
	D(s,v)=-\frac{s^2}{2}\left(1-a^2(\kappa)\right)+vsV(\kappa)+\ln\theta_3\big(sV(\kappa)\big|\tau(\kappa)\big)+\int_s^{\infty}M\big(tV(vt^{-1}),vt^{-1}\big)\frac{\d t}{t}
\end{equation}
where $\kappa=\frac{v}{s}\in[0,1-\delta]$. Let in addition
\begin{equation}\label{jsvform}
	j(s,v)=\frac{\partial}{\partial s}\Big(\ln\det(I-\gamma K_s)-D(s,v)\Big),
\end{equation}
so that expansion \eqref{exf1} is equivalent to the following estimation
\begin{cor}\label{corset:1} For any given $\delta\in(0,\frac{1}{2}]$ there exists positive constants $s_0=s_0(\delta),v_0=v_0(\delta)$ and $c=c(\delta)$ such that
\begin{equation*}
	\big|j(s,v)\big|\leq cs^{-2},\ \ \ \ \forall\,s\geq s_0(\delta),\ \ v\geq v_0(\delta),\ \ \ \ \ s\delta\leq v\leq s(1-\delta).
\end{equation*}
\end{cor}
Define
\begin{equation*}
	F(s,v)=-\frac{4vs}{\pi}+\frac{2v^2}{\pi^2}\ln(4s)+2\ln b(v)
\end{equation*}
with $b(v)$ as in \eqref{BWBB2}. Observe that from \eqref{aasy} and \eqref{c:1},\eqref{c:2}, and \eqref{integrab112} we have
\begin{equation}\label{Dsvlower}
	D(s,v)=-\frac{4vs}{\pi}+\frac{2v^2}{\pi^2}\ln(4s)+\frac{v^2}{\pi^2}\left(3+2\ln\left(\frac{\pi}{v}\right)\right)+\mathcal{O}\left(s^{2-3\epsilon}\ln s\right),\ s\rightarrow\infty
\end{equation}
uniformly for $s^{1-\epsilon}\geq v>0$ with any given $\frac{2}{3}<\epsilon<1$. Now keep $v$ fixed through the rest of this section. For fixed $v\in(0,\infty)$, it is known from \eqref{BWBB1} that as $s\rightarrow\infty$,
\begin{equation*}\label{Acomp:3}
	\ln\det(I-\gamma K_s)=F(s,v)+\mathcal{O}\left(s^{-1}\right).
\end{equation*}
Hence, for fixed $v$, as $s\rightarrow\infty$,
\begin{eqnarray}
	\ln\det(I-\gamma K_s)-D(s,v) &=& F(s,v)-D(s,v)+\mathcal{O}\left(s^{-1}\right)\nonumber\\
	&=&-\frac{v^2}{\pi^2}\left(3+2\ln\left(\frac{\pi}{v}\right)\right)+2\ln b(v)+\mathcal{O}\left(s^{2-3\epsilon}\ln s\right),\ \ \ \frac{2}{3}<\epsilon<1.\label{n:1}
\end{eqnarray}
Set
\begin{equation}\label{AVdef}
	A(v)=2\ln b(v)-\frac{v^2}{\pi^2}\left(3+2\ln\left(\frac{\pi}{v}\right)\right),\ \ \ \ 0<v<\infty.
\end{equation}
Then, from \eqref{n:1} we obtain
\begin{prop}\label{Avfin} For every fixed $v\in(0,\infty)$,
\begin{equation}\label{DI:1}
	\ln\det(I-\gamma K_s)=D(s,v)+A(v)+o(1),\ \ \ s\rightarrow\infty,
\end{equation}
where the $s$ independent term $A(v)$ is given by formula \eqref{AVdef}.
%
\end{prop}
\begin{rem} Note that, apart from the asymptotics in \eqref{Dsvlower} which involves an estimation of known functions, all we need to determine $A(v)$ is the fixed $v$ result \eqref{BWBB1} in \cite{BW,BB}. 
\end{rem}
Let
\begin{equation}\label{Jsvdef}
	J(s,v)=\ln\det(I-\gamma K_s)-D(s,v)-A(v),
\end{equation}
be the error term in \eqref{DI:1}. This function is defined for all $s>0,0<v<s(1-\delta)$ and, by its very definition, satisfies the estimate,
\begin{equation}\label{estJ0}
	J(s,v)=o(1),\ \ \ s\rightarrow\infty,
\end{equation}
for every fixed $v$. In the next section we will extend this estimate to the whole domain $0<v<s(1-\delta)$, and hence complete the proof of Theorem \ref{theo1}. In order to make this extension we note that $J(s,v)$ can also be expressed as the integral of the function $j(s,v)$ in \eqref{jsvform},
\begin{equation}\label{n:4}
	J(s,v)=-\int_s^{\infty}j(t,v)\,\d t.
\end{equation}
The constant of integration is fixed by the a-priori fact that $J(s,v)$, defined by \eqref{Jsvdef}, obeys \eqref{estJ0}. The existence of the integral on the right hand side is guaranteed by the estimate,
\begin{equation*}
	j(s,v)=\mathcal{O}\left(s^{1-3\epsilon}\ln s\right),\ \ \ \frac{2}{3}<\epsilon<1,
\end{equation*}
which in turn follows from the differentiability of \eqref{n:1} with respect to $s$ (\cite{BDIK}, see also Remark \ref{Busref}).
\section{Estimation of the ``error"-term $J(s,v)$}\label{error}
Our goal will be to obtain control over the error term $J(s,v)$ in a sufficiently large region in the $(s,v)$-plane below the main diagonal as shown in Figure \ref{Deiftdia2}. Note that in view of Corollary \ref{corset:1} we can already control the integrand $j(s,v)$ in the region
\begin{equation*}
	\big\{(s,v)\,|\,s\geq s_0(\delta),\ v\geq v_0(\delta):\ \ s\delta\leq v\leq s(1-\delta)\big\}.
\end{equation*}
This, however, is not yet sufficient to estimate $J(s,v)$ in \eqref{n:4}.
\subsection{Relaxation at the lower end}
Choose $s\geq s_0,v\geq v_0$ throughout such that
\begin{equation}\label{m1}
	0<s^{1-\epsilon}\leq v \leq s(1-\delta),\hspace{0.5cm} 0<\epsilon<1.
\end{equation}
In this case the radii of $C_{l_j}$ and $C_{r_j}$ in \eqref{it1} cannot be fixed any longer since
\begin{equation*}
	a=1-\frac{\hat{c}}{s^{\epsilon}}+\mathcal{O}\left(s^{-2\epsilon}\right),\ \ \ s\geq s_0,\ \ \textnormal{for some}\ \hat{c}>0.
\end{equation*}
Instead, we define $R(\lambda)$ as in \eqref{it1}, but with contracting as $s\rightarrow\infty$ radii
\begin{equation*}
	r_1 = c_1s^{-\epsilon}\hspace{0.5cm}r_2=c_2s^{-\epsilon},\hspace{1cm} 0<c_1,c_2:\ \ c_1+c_2<\hat{c}.
\end{equation*}
In this case, $G_R(\lambda)-I$ satisfies the following estimate: 
\begin{equation}\label{m2}
	\sup_{\lambda\in C_{r_1}\cup C_{l_1}}\left|G_R(\lambda)-I\right|=\mathcal{O}\left(s^{-(1-\epsilon)}\right)=o(1)
\end{equation}
which can be read off from our computations in Section \ref{firstcorr}. We used, in particular, the fact,  that for $\kappa=cs^{-\epsilon}$ by Corollary \ref{cor1},
\begin{equation*}
	\tau=\frac{2i\epsilon}{\pi}\ln s+\mathcal{O}(1),
\end{equation*}
hence all expressions involving the theta functions are bounded, compare also the proof of Proposition \ref{integrab}. In order to give a similar estimate for $\lambda\in C_{r_2}\cup C_{l_2}$, we first recall \eqref{p18} and \eqref{p22}. In these formulae the exponential vanishing of the ``logarithmic corrections'' still takes place, since by \eqref{m1}
\begin{equation*}
	s\kappa\rightarrow+\infty
\end{equation*}
and we chose $r_2 = \mathcal{O}\left(s^{-\epsilon}\right)$. Hence our computations in Section \ref{firstcorr} lead to
\begin{equation}\label{m3}
	\sup_{\lambda\in C_{r_2}\cup C_{l_2}}\left|G_R(\lambda)-I\right|=\mathcal{O}\left(s^{-(1-\epsilon)}\right)=o(1).
\end{equation}
On the lens boundaries $\gamma^+\cup\gamma^-$ we first notice that
\begin{equation*}
	\sup_{\lambda\in\gamma^+\cup\gamma^-}\left|G_R(\lambda)-I\right|=\sup_{\lambda\in(\gamma^+\cup\gamma^-)\cap(C_{r_1}\cup C_{l_1})}\left|G_R(\lambda)-I\right|,
\end{equation*}
and for $\lambda\in(\gamma^+\cup\gamma^-)\cap(C_{r_1}\cup C_{l_1})$ we have
\begin{equation*}
	\beta(\lambda) = \left(\frac{(\lambda+a)(\lambda-1)}{(\lambda+1)(\lambda-a)}\right)^{\frac{1}{4}} = \mathcal{O}(1),
\end{equation*}
and therefore, by our previous remark on the boundedness of the theta functions,
\begin{equation*}
	M(\lambda) = \mathcal{O}(1),\hspace{0.5cm}\lambda\in(\gamma^+\cup\gamma^-)\cap(C_{r_1}\cup C_{l_1}).
\end{equation*}
Since 
\begin{equation*}
	\Re\left(\pm is\Omega(\lambda)\right)\leq -\tilde{c}s^{1-\epsilon},\ \ \tilde{c}>0,\ \ \lambda\in(\gamma^+\cup\gamma^-)\cap(C_{r_1}\cup C_{l_1})
\end{equation*}
we finally obtain
\begin{equation}\label{m4}
	\sup_{\lambda\in\gamma^+\cup\gamma^-}\left|G_R(\lambda)-I\right|=\mathcal{O}\left(e^{-\tilde{c}s^{1-\epsilon}}\right)=o(1).
\end{equation}
Next consider the contracting line segments $(-1+r_2,-a-r_1)\cup(a+r_1,1-r_2)$. Note that for some $\tilde{c}>0$,
\begin{equation*}
 \gamma=1-e^{-2\kappa s} =1+\mathcal{O}\left(e^{-\tilde{c}s^{1-\epsilon}}\right)
\end{equation*}
and
\begin{equation*}
	M_-(\lambda) = \mathcal{O}(1),\hspace{0.5cm}\lambda\in(-1+r_2,-a-r_1)\cup(a+r_1,1-r_2).
\end{equation*}
Now estimate, for $\lambda\in(a+r_1,1-r_2)$, say,
\begin{eqnarray*}
	\Pi(\lambda) &=&-2\int_{\lambda}^1\sqrt{\frac{\mu^2-a^2}{1-\mu^2}}\,\d\mu\leq -d_1\int_{1-r_2}^1\sqrt{\frac{\mu-a}{1-\mu}}\,\d\mu\leq -d_2\sqrt{r_2}\sqrt{1-a-r_2}\\
	2\kappa+\Pi(\lambda)&=&2\int_a^{\lambda}\sqrt{\frac{\mu^2-a^2}{1-\mu^2}}\,\d\mu\geq d_3\int_a^{a+r_1}\sqrt{\frac{\mu-a}{1-\mu}}\,\d\mu\geq \frac{d_4r_1^{\frac{3}{2}}}{\sqrt{1-a}},\ \ d_j>0
\end{eqnarray*}
and deduce for some $\tilde{c}>0$,
\begin{equation*}
	e^{s\Pi(\lambda)}=\mathcal{O}\left(e^{-\tilde{c}s^{1-\epsilon}}\right),\hspace{0.5cm} e^{-s(2\kappa+\Pi(\lambda))}=\mathcal{O}\left(e^{-\tilde{c}s^{1-\epsilon}}\right).
\end{equation*}
Using similar estimates for $\lambda\in(-1+r_2,-a-r_1)$, we obtain
\begin{equation}\label{m5}
	\sup_{\lambda\in(-1+r_2,-a-r_1)}\left|G_R(\lambda)-I\right|=\mathcal{O}\left(e^{-\tilde{c}s^{1-\epsilon}}\right)=\sup_{\lambda\in(a+r_1,1-r_2)}\left|G_R(\lambda)-I\right|.
\end{equation}
For the parts $\dot{C}_1$ of the original jump contours inside $C_{l_1}\cup C_{r_1}$ we recall that $U_-(\lambda)$ and $V_-(\lambda)$ are bounded, so we can apply similar estimates (see \eqref{Gr1} and \eqref{Gl1}) as above. This leads to
\begin{equation}\label{m6}
	\sup_{\lambda\in\dot{C}_1}\left|G_R(\lambda)-I\right|=\mathcal{O}\left(e^{-\tilde{c}s^{1-\epsilon}}\right).
\end{equation}
Combining \eqref{m2}--\eqref{m6} we obtain
\begin{equation}\label{DZ9}
	\|G_R-I\|_{L^2\cap L^{\infty}(\Sigma_R)}\leq \tilde{c}s^{-(1-\epsilon)}
\end{equation}
which holds under the assumption \eqref{m1}. This should be compared to \eqref{DZ7} in the case $\kappa\in[\delta,1-\delta]$. The estimate \eqref{DZ9} ensures solvability of the R-RHP and we have for its solution
\begin{equation*}
	\|R_--I\|_{L^2(\Sigma_R)}\leq \tilde{c}s^{-(1-\frac{\epsilon}{2})}
\end{equation*}
subject to \eqref{m1}. This last estimation follows from the general small norm-theory of Riemann-Hilbert problems \cite{DZ}, since the scaling invariance of  $\frac{\d\lambda}{\lambda}$ on $\partial\mathbb{D}=\{\lambda\in\mathbb{C}:\ |\lambda|=1\}$ implies a slightly better $L^2$-estimation than \eqref{DZ9}, namely
\begin{equation*}
	\|G_R-I\|_{L^2(\Sigma_R)}\leq\tilde{c}s^{-(1-\frac{\epsilon}{2})}.
\end{equation*}
Now use Cauchy-Schwarz inequality and obtain
\begin{eqnarray*}
	\int_{\Sigma_R}R_-(w)\big(G_R(w)-I\big)\d w &=& \int_{\Sigma_R}\big(G_R(w)-I\big)\d w+\int_{\Sigma_R}\big(R_-(w)-I\big)\big(G_R(w)-I\big)\d w\\
	 &=& \int_{\Sigma_R}\big(G_R(w)-I\big)\d w+\mathcal{O}\left(s^{-2+\epsilon}\right).
\end{eqnarray*}
Hence, reproducing the steps which first lead us to \eqref{lint0} and then \eqref{exf1}, also recalling Proposition \ref{integrab}, we have for
\begin{equation*}
	j(s,v)=\frac{\partial}{\partial s}\Big(\ln\det(I-\gamma K_s)-D(s,v)\Big)
\end{equation*}
\begin{prop}\label{LowerAv} Given any $\delta,\epsilon\in(0,1)$, there exist positive constants $s_0=s_0(\delta,\epsilon),v_0=v_0(\delta,\epsilon)$ and $c=c(\epsilon,\delta)$ such that
\begin{equation*}
	\big|j(s,v)\big|\leq cs^{-2+\epsilon},\ \ \forall\,s\geq s_0(\delta,\epsilon),\ \ v\geq v_0(\delta,\epsilon),\ \ \ s^{1-\epsilon}\leq v\leq s(1-\delta).
\end{equation*}
\end{prop}
\subsection{Estimation of the error term}\label{esticruc}
Let us choose $\epsilon=\epsilon_k=\frac{k}{k+1}$ with $k\in\mathbb{Z}_{\geq 3}$ in Proposition \ref{LowerAv} so that
\begin{equation}\label{es:1}
	\big|j(s,v)\big|\leq c s^{-2+\frac{k}{k+1}}
\end{equation}
is valid for $(1-\delta)s\geq v\geq v_0>0$ and $0<s_0\leq s<v^{k+1}$ with $s_0>(1-\delta)^{-\frac{1}{\epsilon_k}}$. We first consider the case where $0<s_0\leq s<v^{k+1}$ and split the integral
\begin{equation}\label{Aint}
	J(s,v) = -\int_s^{v^{k+1}}j(t,v)\d t-\int_{v^{k+1}}^{\infty}j(t,v)\d t\equiv J_1(s,v)+B(v)
\end{equation}
and obtain from \eqref{es:1},
\begin{equation}\label{J1esti}
	\big|J_1(s,v)\big|\leq c(\delta,\epsilon_k)s^{-\frac{1}{k+1}},\ \ \ \ 0<s_0\leq s<v^{k+1},\ \ v\leq s(1-\delta),\ \ \ k\in\mathbb{Z}_{\geq 3}.
\end{equation}
In order to obtain a good estimation for $B(v)$, we now make use of \eqref{BBrefined} which will be derived in \cite{BDIK}. Namely, with
\begin{equation*}
	F(s,v)=-\frac{4vs}{\pi}+\frac{2v^2}{\pi^2}\ln(4s)+2\ln b(v);\hspace{1.5cm} q(s,v)=\frac{\partial}{\partial s}\Big(\ln\det(I-\gamma K_s)-F(s,v)\Big),
\end{equation*}
there exist positive constants $s_0=s_0(\epsilon_k)$ and $c=c(\epsilon_k)$ such that
\begin{equation}\label{Av:4}
	\big|q(s,v)\big|\leq cs^{-2+\frac{3}{k+1}},\ \ \ \ \forall\,s\geq v^{k+1}>0,\ s\geq s_0,\ \ k\in\mathbb{Z}_{\geq 3}.
\end{equation}
We have
\begin{equation}\label{Av:5}
	j(s,v)=q(s,v)+\frac{\partial}{\partial s}\Big(F(s,v)-D(s,v)\Big).
\end{equation}
\begin{rem}
From \eqref{Dsvlower} we have
\begin{equation}\label{diffexp}
	\frac{\partial}{\partial s}\big(F(s,v)- D(s,v)\big)=\mathcal{O}\left(s^{-2+\frac{3}{k+1}}\ln s\right)
\end{equation}
uniformly for $s\geq v^{k+1}>0,s\geq s_0,k\in\mathbb{Z}_{\geq 3}$.
\end{rem}	
Now combine \eqref{diffexp} with \eqref{Av:4} and obtain in \eqref{Av:5} 
\begin{prop}\label{Aproof} There exist constants $s_0=s_0(\epsilon_k)$ and $c=c(\epsilon_k)$ such that
\begin{equation*}
	\big|j(s,v)\big|\leq cs^{-2+\frac{3}{k+1}}\ln s,\ \ \ \ \forall\,s\geq v^{k+1}>0,\ s\geq s_0,\ \ k\in\mathbb{Z}_{\geq 3}.
\end{equation*}
\end{prop}	
With Proposition \ref{Aproof} at hand, we go back to \eqref{Aint} and obtain first
\begin{equation*}
	\big|B(v)\big|\leq c(\delta,\epsilon_k)s^{-1+\frac{3}{k+1}}\ln s,\ \ \ 0<s_0\leq s<v^{k+1},\ \ 0<v\leq s(1-\delta),\ \ k\in\mathbb{Z}_{\geq 3},
\end{equation*}
so that combined with the estimation \eqref{J1esti}, for any given $\delta\in(0,1)$, for $k=3$,
\begin{equation}\label{Av:7}
	\big|J(s,v)\big|\leq c(\delta)s^{-\frac{1}{4}}\ln s,\ \ \ \forall\,s\geq s_0,\ \ s^{\frac{1}{4}}\leq v\leq s(1-\delta).\smallskip
\end{equation}
On the other hand, for $s\geq s_0,0<v<s^{\frac{1}{4}}$, we can use Proposition \ref{Aproof} directly
and deduce, 
\begin{equation*}
	\big|J(s,v)\big|\leq \int_s^{\infty}\big|j(t,v)\big|\d t\leq c(\delta)s^{-\frac{1}{4}}\ln s,\ \ \ \ \ \ 0<v<s^{\frac{1}{4}}.
\end{equation*}
We summarize
\begin{prop}\label{Errorfin} Given any $\delta\in(0,1)$ there exist positive constants $s_0=s_0(\delta),v_0=v_0(\delta)$ and $c=c(\delta)$ such that
\begin{equation*}
	\big|J(s,v)\big|\leq cs^{-\frac{1}{4}}\ln s,\ \ \ \forall\,s\geq s_0,\ \ 0<v\leq s(1-\delta).
\end{equation*}
\end{prop}
Combining the result of this Proposition with \eqref{Dsv}, \eqref{AVdef}, \eqref{Jsvdef}, \eqref{n:4}, we finish the proof of Theorem \ref{theo1}.
\begin{appendix}

\section{Some identities for theta functions}\label{thetaid}
The Jacobi theta functions $\theta_0(z)\equiv \theta_4(z),\theta_1(z),\theta_2(z),\theta_3(z)$ are given by the formulae (see e.g. \cite{BE,WW}) where $\Im\,\tau>0$,
\begin{eqnarray*}
	\theta_3(z|\tau) &\equiv& \theta(z|\tau) = 1+2\sum_{k=1}^{\infty}e^{\pi ik^2\tau}\cos(2\pi kz),\\
	\theta_0(z|\tau) &=&\theta_3\left(z+\frac{1}{2}\bigg|\,\tau\right) = 1+2\sum_{k=1}^{\infty}(-1)^ke^{\pi ik^2\tau}\cos\big(2\pi kz\big)\equiv \theta_4(z|\tau),\\
	\theta_2(z|\tau) &=& e^{\pi i\frac{\tau}{4}+\pi iz}\theta_3\left(z+\frac{\tau}{2}\bigg|\,\tau\right)=2\sum_{k=0}^{\infty}e^{\pi i\tau(k+\frac{1}{2})^2}\cos\big((2k+1)\pi z\big),\\
	\theta_1(z|\tau) &=& -ie^{\pi i\frac{\tau}{4}+\pi iz}\theta_3\left(z+\frac{1}{2}+\frac{\tau}{2}\bigg|\,\tau\right)=2\sum_{k=0}^{\infty}(-1)^ke^{\pi i(k+\frac{1}{2})^2\tau}\sin\big((2k+1)\pi z\big).
\end{eqnarray*}
We list certain identities for the theta functions. First,
\begin{equation}\label{conn}
	\theta_1\left(z-\frac{1}{2}\right) = -\theta_2(z),\hspace{0.5cm}\theta_0\left(z+\frac{1}{2}\right)=\theta_3(z),\hspace{0.5cm}\theta_0\left(z+\frac{\tau}{2}\right) = ie^{-\pi i\frac{\tau}{4}-i\pi z}\theta_1(z).
\end{equation}
The theta functions satisfy the following periodicity conditions with respect to $\Lambda$:
\begin{eqnarray*}
	\theta_1(z+n+m\tau) &=&\theta_1(z)(-1)^{n+m}e^{-\pi i\tau m^2-2\pi imz},\\
	\theta_2(z+n+m\tau) &=&\theta_2(z)(-1)^ne^{-\pi i\tau m^2-2\pi imz},\\
	\theta_0(z+n+m\tau) &=&\theta_0(z)(-1)^me^{-\pi i\tau m^2-2\pi imz},
\end{eqnarray*}
valid for any $n,m\in\mathbb{Z}$, and $z\in\mathbb{C}$. The zeros of the theta functions are given by
\begin{equation*}
	\theta_0\left(\frac{\tau}{2}\right)=0,\hspace{0.5cm}\theta_1(0) = 0,\hspace{0.5cm}\theta_2\left(\frac{1}{2}\right)=0,\hspace{0.5cm}\theta_3\left(\frac{1}{2}+\frac{\tau}{2}\right)=0,
\end{equation*}
up to shifts by vectors of the lattice $\Lambda$. All zeros are simple. Note also the addition formulae 
\begin{eqnarray*}
	\theta_0^2(0)\theta_0(w+z)\theta_0(w-z) &=&\theta_0^2(w)\theta_0^2(z)-\theta_1^2(w)\theta_1^2(z)\\
	\theta_0^2(0)\theta_1(w+z)\theta_1(w-z) &=&\theta_1^2(w)\theta_0^2(z)-\theta_0^2(w)\theta_1^2(z)\\
	\theta_0^2(0)\theta_2(w+z)\theta_2(w-z) &=&\theta_0^2(w)\theta_2^2(z)-\theta_1^2(w)\theta_3^2(z)\\
	\theta_0^2(0)\theta_3(w+z)\theta_3(w-z) &=&\theta_0^2(w)\theta_3^2(z)-\theta_1^2(w)\theta_2^2(z)
\end{eqnarray*}
and a duplication formula
\begin{equation}\label{double}
	\theta_0(2z)\theta_0^3(0) = \theta_3^4(z)-\theta_2^4(z)=\theta_0^4(z)-\theta_1^4(z),\hspace{0.5cm}z\in\mathbb{C}.
\end{equation}
Furthermore, all Jacobi theta functions $\theta=\theta(z|\tau)$ solve the partial differential equation
\begin{equation}\label{de}
	\frac{\partial^2}{\partial z^2}\theta(z|\tau)=4\pi i\frac{\partial}{\partial\tau}\theta(z|\tau).
\end{equation}
\end{appendix}

\end{document}